\documentclass[10pt,journal,compsoc]{IEEEtran}

\usepackage{cite}
\usepackage{amsmath,amssymb,amsfonts}
\usepackage{graphicx}
\usepackage{textcomp}
\usepackage{xcolor}
\usepackage{booktabs} 
\usepackage{algpseudocode}
\usepackage{color}
\usepackage{soul}
\usepackage{subcaption}
\usepackage{url}
\usepackage[linesnumbered,ruled]{algorithm2e}
\usepackage{hhline}
\usepackage{multirow}
\usepackage{makecell}
\usepackage[flushleft]{threeparttable}

\usepackage{framed}
\usepackage{multicol}
\usepackage{lipsum}
\usepackage{amsthm}

\newtheorem{definition}{Definition}
\newtheorem{theorem}{Theorem}
\newtheorem{proposition}{Proposition}

\newcolumntype{E}{>{\centering\arraybackslash}p{0.8cm}}
\newcolumntype{I}{>{\centering\arraybackslash}p{1cm}}
\newcolumntype{B}{>{\centering\arraybackslash}p{4.1cm}}
\newcolumntype{P}{>{\centering\arraybackslash}p{4.1cm}}
\newcolumntype{D}{>{\centering\arraybackslash}p{5.6cm}}

\newcolumntype{Z}{>{\centering\arraybackslash}p{0.6cm}}
\newcolumntype{Y}{>{\centering\arraybackslash}p{0.8cm}}
\newcolumntype{L}{>{\centering\arraybackslash}p{2.41cm}}
\newcolumntype{V}{>{\centering\arraybackslash}p{2.5cm}}
\newcolumntype{N}{>{\centering\arraybackslash}p{2.5cm}}
\newcolumntype{O}{>{\centering\arraybackslash}p{2.53cm}}
\newcolumntype{R}{>{\centering\arraybackslash}p{2.43cm}}

\let\oldnl\nl
\newcommand{\nonl}{\renewcommand{\nl}{\let\nl\oldnl}}%

\newcommand{\system}{ShieldDB}

\ifCLASSINFOpdf
\else
\fi
\begin{document}

\title{ShieldDB: An Encrypted Document Database with Padding Countermeasures}

\author{Viet~Vo, Xingliang~Yuan, Shi-Feng~Sun, Joseph~K.~Liu, Surya~Nepal, and~Cong~Wang
\IEEEcompsocitemizethanks{\IEEEcompsocthanksitem Viet~Vo, Xingliang~Yuan, and Joseph~K.~Liu are with the Faculty of Information Technology, Monash University, Melbourne, Australia. 
E-mail: {\{viet.vo,xingliang.yuan,joseph.liu\}}@monash.edu
\IEEEcompsocthanksitem Viet Vo and Surya Nepal are with the Data 61, CSIRO, Australia. 
E-mail: {\{viet.vo,surya.nepal\}}@data61.csiro.au
\IEEEcompsocthanksitem Shi-Feng~Sun is with the Department of Computer Science and Engineering, Shanghai Jiao Tong University, China. 
E-mail: shifeng.sun@sjtu.edu.cn
\IEEEcompsocthanksitem Cong Wang is with the Department of Computer Science, City University of Hong Kong, Hong Kong, China. 
E-mail: congwang@cityu.edu.hk
}}


\markboth{Journal of \LaTeX\ Class Files,~Vol.~14, No.~8, August~2015}%
{Shell \MakeLowercase{\textit{et al.}}: Bare Demo of IEEEtran.cls for Computer Society Journals}

\IEEEtitleabstractindextext{
\begin{abstract}
Cloud storage systems have seen a growing number of clients due to the fact that more and more businesses and governments are shifting away from in-house data servers and seeking cost-effective and ease-of-access solutions. However, the security of cloud storage is underestimated in current practice, which resulted in many large-scale data breaches. To change the status quo, this paper presents the design of ShieldDB, an encrypted document database. ShieldDB adapts the searchable encryption technique to preserve the search functionality over encrypted documents without having much impact on its scalability. However, merely realising such a theoretical primitive suffers from real-world threats, where a knowledgeable adversary can exploit the leakage (aka access pattern to the database) to break the claimed protection on data confidentiality. To address this challenge in practical deployment, ShieldDB is designed with tailored padding countermeasures. Unlike prior works, we target a more realistic adversarial model, where the database gets updated continuously, and the adversary can monitor it at an (or multiple) arbitrary time interval(s). ShieldDB's padding strategies ensure that the access pattern to the database is obfuscated all the time. We present a full-fledged implementation of ShieldDB and conduct intensive evaluations on Azure Cloud.
\end{abstract}

%

}

\maketitle
\IEEEdisplaynontitleabstractindextext

\IEEEpeerreviewmaketitle

\vspace{-10pt}
\IEEEraisesectionheading{\section{Introduction}\label{sec:introduction}}

The adaptation of cloud storage by various governments and businesses is rapid and inexorable.
It has brought many benefits to the economy and user productivity improvements, and caused a huge upheaval in the data center infrastructure development.
However, data breaches in cloud stores are happening quite frequently in recent time, affecting millions of individuals~\cite{InformationisBeautiful,IBMAU-breach19,verizon-breach20}. 
This phenomenon calls for increased control and security for private and sensitive data~\cite{CLiu14,Fadolalkarim20,Ji19,Yi16}. 
To combat against such ``breach fatigue'', encrypted database systems recently receive wide attention~\cite{Poparz11,microsoftEDB,PappasVK14,PapadimitriouBCRHSMB16,PoddarBP16,YuanGWWLJ17}. 
Their objective is to preserve the query functionality of databases over encrypted data; that is, the server can process a client's encrypted query without decryption. The first generation of encrypted databases~\cite{Poparz11,microsoftEDB,Huanchen20,Meng18} implements property-preserving encryption (PPE) in a way that a ciphertext inherits equality and/or order properties of the underlying plaintext. 
However, inference attacks can compromise these encryption schemes by exploiting the above properties preserved in the ciphertexts~\cite{NaveedKW15,Vincent18}.

In parallel, dedicated privacy-preserving query schemes are investigated intensively in the past decade. Among others, searchable symmetric encryption (SSE)~\cite{SoWa00, CurtmolaGKO06} is well known for its  applications to ubiquitous keyword based search. In general, SSE schemes utilise an encrypted index to enable the server to search over encrypted documents. The server is restricted such that only if a query token (keyword ciphertext) is given, the search operation against the index will be triggered to output the matched yet encrypted documents. 
This ensures that an adversary with a full image of the encrypted database learns no useful information about the documents. 
In that sense, SSE outperforms PPE in terms of security. 
Apart from security, SSE is scalable, because it is realised via basic symmetric primitives.
%

In this paper, we design an encrypted document database system based on SSE, for which the client can add encrypted documents to the database and query encrypted keywords against it. However, deploying SSE in practice is non-trivial.
With recent emerging inference attacks against SSE~\cite{IslamKK12,CashGP15,ZhangKP16} introduced, it raises doubts whether SSE achieves an acceptable tradeoff between the efficiency and security. 
As a noteworthy threat, the count attack~\cite{CashGP15} demonstrates that an adversary with extra background knowledge of a dataset can analyse the size of the query result set to recover keywords from the query tokens. 
The above information is known as \emph{access pattern}, and can be monitored via the server's memory access and communication between the server and client.
If SSE is deployed to a database, access patterns can also be derived from database logs~\cite{GrubbsRS17}. This situation further reduces the security claim of SSE, since the adversary does not have to stay online for monitoring. 

Using oblivious-RAM (ORAM) is an quintessential approach to enable encrypted search without exposing access pattern~\cite{Saba19,Liu20}, but it is shown as an expensive tool~\cite{Naveed15,IslamKK12,Cash15,Mishra18}.
Alternatively, using padding (bogus documents) for inverted index solution~\cite{CurtmolaGKO06} is proven as a conceptually simple but effective countermeasure to obfuscate the access pattern against the aforementioned attacks~\cite{IslamKK12,CashGP15,BostF17}. Unfortunately, existing padding countermeasures only consider a static database, where padding is only added at the setup~\cite{BostF17,Cash15}. 
They are not sufficient for real-world applications. In practice, the states of database change over the time. Specifically, the updates of documents change the access pattern for a given keyword, and new keywords can be introduced randomly at any time. Hence, exploring to what extend adversaries can exploit such changes to compromise the privacy of data and how padding countermeasures can be applied in a dynamic environment are essential to make SSE deployable in practice.

\vspace{2pt}
\noindent \textbf{Contributions}: To address the above issues, we propose and implement an encrypted document database named \system, in which the data and query security in realistic and dynamic application scenarios is enhanced via effective padding countermeasures. Our contributions are as follows:   
\begin{itemize}

\item \system~is the first encrypted database that supports encrypted keyword search, while equipping with padding countermeasures against inference attacks launched by  adversaries with database background knowledge. 

\item We define two new types of attack models, i.e., \emph{non-persistent} and \emph{persistent} adversaries, which faithfully reflect different real-world threats in a continuously updated database. Accordingly, we propose padding countermeasures to address these two adversaries. 

\item \system~ is designed with a dedicated system architecture to achieve the functionality and security goals. Apart from the client and server modules for encrypted keyword search, a \textit{Padding Service} is developed. This service leverages two controllers, i.e., \emph{Cache Controller} and \emph{Padding Controller}, to enable efficient and effective database padding.  

\item \system~ implements advanced features to further improve the security and performance. These features include: 1) forward privacy that protects the newly inserted document, 2) flushing that can reduce the load of the padding service, and 3) re-encryption that refreshes the ciphertexts while realising deletion and reducing padding overhead.  

\item We present the implementation and optimization of \system, and deploy it in Azure Cloud. We build a streaming scenario for evaluation. 
In particular, we implement an aggressive padding mode (\emph{high} mode) and a conservative padding mode (\emph{low} mode), and compare them with padding strategies against \textit{non-persistent} and \textit{persistent adversaries}, respectively. 
We perform a comprehensive set of evaluations on the load of the cache, system throughput, padding overhead, and search time.  
Our results show that the \emph{high} mode results in much larger padding overhead than the \emph{low} mode does, while achieving lower cache load.
In contrast, the \emph{low} mode results in higher system throughput (accumulated amount of real data) but requiring a higher cache load.
%
%
The evaluations of flushing and re-encryption demonstrate the effective reduction of the cache load and padding overhead.

\end{itemize}

\noindent \textbf{Organisation}: We present important related works in Section~\ref{sec:preliminary}. Section~\ref{sec:system} overviews the design of~\system~ and demonstrates \textit{non-persistent} and \textit{persistent} adversaries. Section~\ref{sec:construction} demonstrates  padding countermeasures and other optimisation features supported in~\system. Section~\ref{subsec:security_shielddb} analyses the security of the system. Sections~\ref{sec:implementation} and~\ref{sec:exp} present system deployment and evaluations. Conclusion is in Section~\ref{sec:conclusion}.

\section{Background} 
\label{sec:preliminary}
In this section, we introduce the background knowledge related to the design of our system.  

\vspace{2pt}
\noindent \textbf{Dynamic symmetric searchable encryption}. Considering a client \textit{C} and a server \textit{S}, the client encrypts documents in a way that the server can query keywords over the encrypted documents. Functions included in an SSE scheme are \textsf{setup} and \textsf{search}. If the scheme is dynamic, \textsf{update} functions (data addition and deletion) are also supported.  
Let DB represent a database of  documents, and each document is a variable-length set of unique keywords. We use $\Delta=\{w_1,..., w_m\}$ to present all keywords occurred in \textnormal{DB}, \textnormal{DB}($w$) to present documents that contain $w$, and $|\textnormal{DB} (w)|$ to denote the number of those documents, i.e., the size of the query result set for $w$. In SSE, the encrypted database, named EDB, is a collection of encrypted keyword and document id pairs $(w,id)$'s, aka an encrypted searchable index. 

In \textsf{setup}, client encrypts \textnormal{DB} by using a cryptographic key \textit{k}, and sends EDB to server. During \textsf{search}, client takes \textit{k} and a query keyword $w$ as an input, and outputs a query token \textit{tk}. \textit{S} uses \textit{tk} to query EDB to get the pseudo-random identifies of the matched documents, so as to return the encrypted result documents. In \textsf{update}, \textit{C} takes an input of \textit{k}, a document $\textit{D}_i$ parsed as a set of $(w,id)$ pairs, and an operation \textsf{op} $\in \{add, del\}$. For addition, the above pairs are  inserted to EDB. For deletion, the server no longer returns encrypted documents in subsequent search queries. As an output, the server returns an updated EDB. We refer readers to~\cite{Kamara12,Cash14} for more details about dynamic SSE.

The security of dynamic SSE can be quantified via a tuple of stateful leakage functions $\mathcal{L}=(\mathcal{L}^{Setup},\mathcal{L}^{Search},\mathcal{L}^{Update})$. They define the side information exposed in \textsf{setup}, \textsf{search}, and \textsf{update} operations, respectively. 
The detailed definitions can be referred to Section~\ref{subsec:security_shielddb}. 

\vspace{2pt}
\noindent \textbf{Count Attack}. Cash et al.~\cite{Cash15} propose a practical attack that exploits the leakage in the \textsf{search} operation of SSE. It is assumed that an adversary with full or partial prior knowledge of DB can uncover keywords from query tokens via \textit{access pattern}. Specifically, the prior knowledge allows the adversary to learn the documents matching a given keyword before queries. Based on this, she can construct a keyword co-occurrence matrix indicating keyword coexisting frequencies in known documents. As a result, if the result length $|\textnormal{DB}(w)|$ for a query token $tk$ is unique and matches with the prior knowledge, the adversary directly recovers $w$. For tokens with the same result length, the co-occurrence matrix can be leveraged to narrow down the candidates. In this work, we extend the threat model of the count attack to the dynamic setting, which will be introduced in Section~\ref{subsec:attackmodels}. 

\vspace{2pt}
\noindent \textbf{Forward Privacy}. Forward privacy in SSE prevents the adversary from exploiting the leakage in \textsf{update} (addition) operations. Given previously collected query tokens, this security notion ensures that these tokens cannot be used to match newly added documents. As our system considers the scenario, where the documents are continuously inserted, we adapt an efficient scheme with forward privacy~\cite{SongDYXZ17} proposed by Song et al. This scheme follows Bost's scheme~\cite{Bost16} that employs trapdoor permutation to secure states associated to newly added $(w, id)$'s. Without being given new states, the server cannot perform search on the new data, and those states can be used to recover old states via trapdoor permutation. 
Specifically, we optimise the adapted scheme in the context of batch insertion and improve the efficiency, which will be introduced in Section~\ref{subsec:optimisation}.

\vspace{-5pt}
\section{Overview}
\label{sec:system}

\subsection{System Overview}
\label{subsec:architecture}

%
\system~is a document-oriented database, where semi-structured records are modeled and stored as documents, and can be queried via keywords or associated attributes. 

\noindent \textbf{Participants and scenarios}: As illustrated in Figure~\ref{fig:system}, \system~consists of a query client \textit{C}, a trusted padding service \textit{P} and an untrusted storage server \textit{S}. 
In our targeted scenario, new documents are continuously inserted to \textit{S}, and required to be encrypted. Meanwhile, \textit{C} expects \textit{S} to retain search functionality over the encrypted documents. 
To enhance the security, \textit{P} adapts padding countermeasures during encryption. 
In this paper, we consider an enterprise that utilises outsourced storage. \textit{P} is deployed at the enterprise gateway and in the same network with \textit{C}, and  \textit{P} encrypts and uploads the documents created by its employees, while \textit{C} is deployed for employees to search the encrypted documents at \textit{S}. 
Note that the deployment of \textit{P} is flexible. It can be separated from or co-located with \textit{C}. 
%




\vspace{2pt}
\noindent \textbf{Overview}: \system~supports three main operations, i.e., \textsf{setup}, \textsf{streaming}, and \textsf{search}. The full protocols of \system ~is presented in Figure~\ref{fig:shieldDB_protocols}.
Apart from the main functions, \system~also supports optimisation features \textsf{deletion} and \textsf{re-encryption}, and \textsf{flushing} operations (Section~\ref{subsec:optimisation}). 

During \textsf{setup}, \textit{P} receives a sample training dataset $\Delta_{stp}$ from \textit{C}, and then it groups keywords into clusters $L=\{L_1, \cdots, L_m\}$ based on keyword frequencies. After that, the module \emph{App Controller} in \textit{P} notifies $L$ to the module \emph{Cache Controller} to initialise a cache capacity $L_i$ for each keyword cluster. \emph{App Controller} also notifies $L$ to the module \textit{Padding Controller} to generate a padding dataset $B$.

\begin{figure}
\centering
\includegraphics[width=0.5\textwidth]{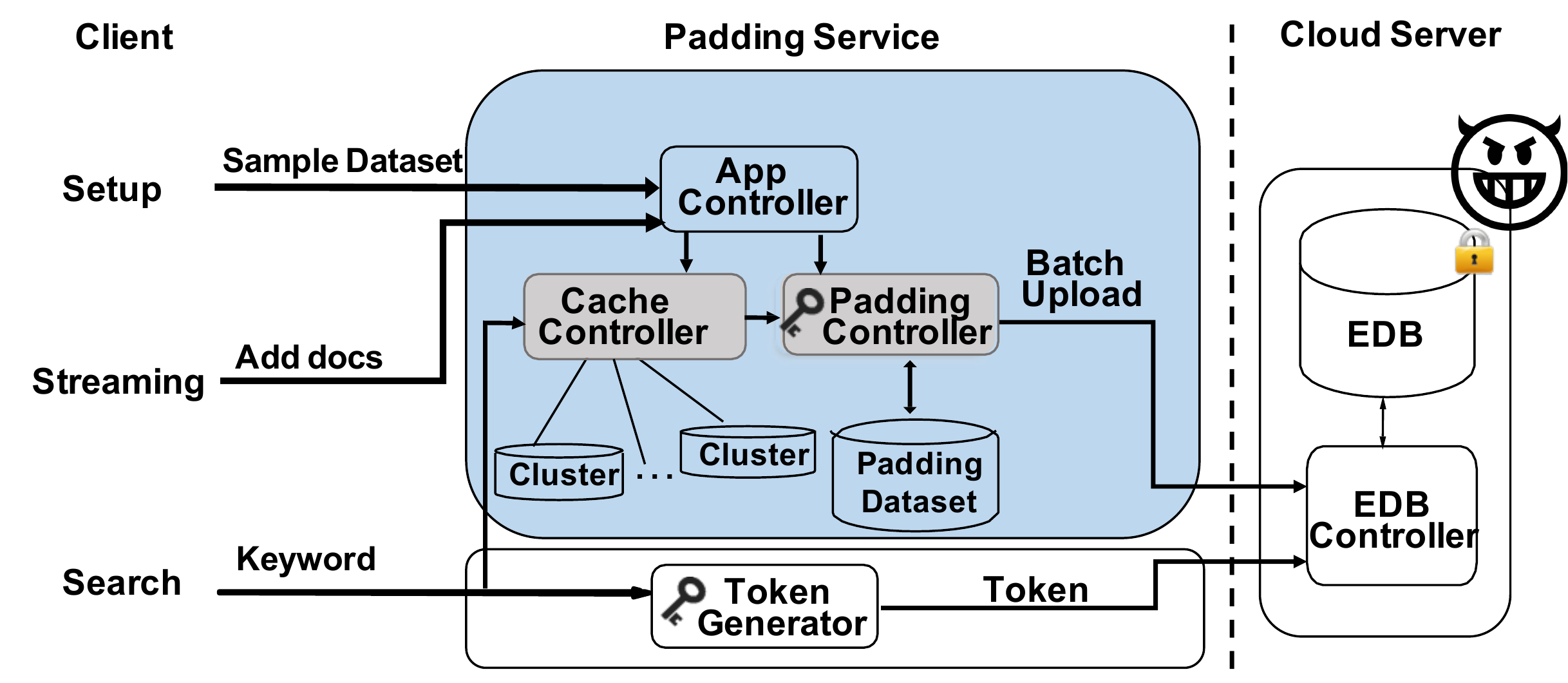} 
\vspace{-20pt}
\caption{High-level design of \system}
\vspace{-15pt}
\label{fig:system}
\end{figure}

During \textsf{streaming}, \textit{P} receives an input, $\textsf{in}=\{(\textsf{doc},id)\}$ containing a collection of documents, each element is a document $\textsf{doc}$ with identifier $id$, sent from \textit{C}.
Then, \textit{P} parses them into a set of keyword and document identifier $(w, id)$ pairs, i.e., index entries for search.
Then, \textit{Cache Controller} stores these pairs to the caches of the corresponding keyword clusters.
Based on the targeted attack model, \textit{Cache Controller} applies certain constraints in ${\tt PaddingCheck}$ to flush the cache (Algorithm~\ref{alg:paddingstrategy}).
Once the constraints on a cluster are met, \textit{Cache Controller} notifies the satisfied cluster to \textit{Padding Controller} for padding.
In particular, \textit{Padding Controller} adds bogus $(w, id)$ pairs extracted from the padding dataset to make the keywords in this cluster have equal frequency.
Then, \textit{P} encrypts and inserts all those real and bogus index entries as a data collection in a batch to EDB (see \textsf{streaming} lines $7-30$ in Figure~\ref{fig:shieldDB_protocols}).

During \textsf{search}, for a given single query keyword $w$, \textit{C} wants to retrieve documents matching that keyword from \textit{S} and \textit{P}.  First, \textit{C} retrieves the local results from \textit{Cache Controller} in \textit{P}, since some index entries might have not been sent to EDB yet. 
After that, \textit{C} sends a query token generated from this keyword to \textit{S} to retrieve the rest of the encrypted results.
After decryption, \textit{C} filters padding and combines the result set with the local one. 
For security, \textit{C} will not generate query tokens against the data collection which is currently in streaming; this constraint enforces \textit{S} to query only over data collections which are already inserted to EDB.    
Following the setting of SSE~\cite{CurtmolaGKO06,KamaraPR12}, \textsf{search} is performed over the encrypted index entries in EDB, and document identifiers are pseudo-random strings. In response to query, \textit{S} will return the encrypted documents via recovered identifiers in the result set after \textsf{search}.

Apart from padding countermeasures, \system~provides several other salient features. First, it realises forward privacy~\cite{Bost16} (an advanced notion of SSE) for the \textsf{streaming} operation. Our realisation is customised for efficient batch insertion and can prevent \textit{S} from searching the data collection in streaming. 
Second, \system~integrates the functionality of \textsf{re-encryption}. Within this operation, index entries in a targeted cluster are fetched back to \textit{P} and the redundant padding is removed. At the same time, \textsf{deletion} can be triggered, where the deleted index entries issued and maintained at \textit{P} are removed and will not be re-inserted. 
After that, real entries combing with new bogus entries are re-encrypted and inserted to EDB.  
Third, \textit{Cache Controller} can issue a secure \textsf{flushing} operation before meeting the constraints for padding. This reduces the overhead of \textit{P} while preserving the security of padding. 

\vspace{2pt}
\noindent \textbf{Remark}: \system~ assigns \textit{P} for key generation and management, and \textit{P} issues the key for \textit{C} to query. 
In addition, \textit{C} also gets the latest state of the query keyword from \textit{P}, and together with the key, to generate query tokens and send them to \textit{S}.
In our current implementation,  \textit{P} and \textit{C} use the same key for index encryption, just as most SSE schemes do. 
This is practical because SSE index only stores pseudorandom identifies of documents, and documents can separately be encrypted via other encryption algorithms. 
Advanced key management schemes of SSE~\cite{SunLSSY16,JareckiJK13} can readily be adapted; yet, this is not relevant to our problem.
%

\begin{figure}
\centering
\includegraphics[width=0.5\textwidth]{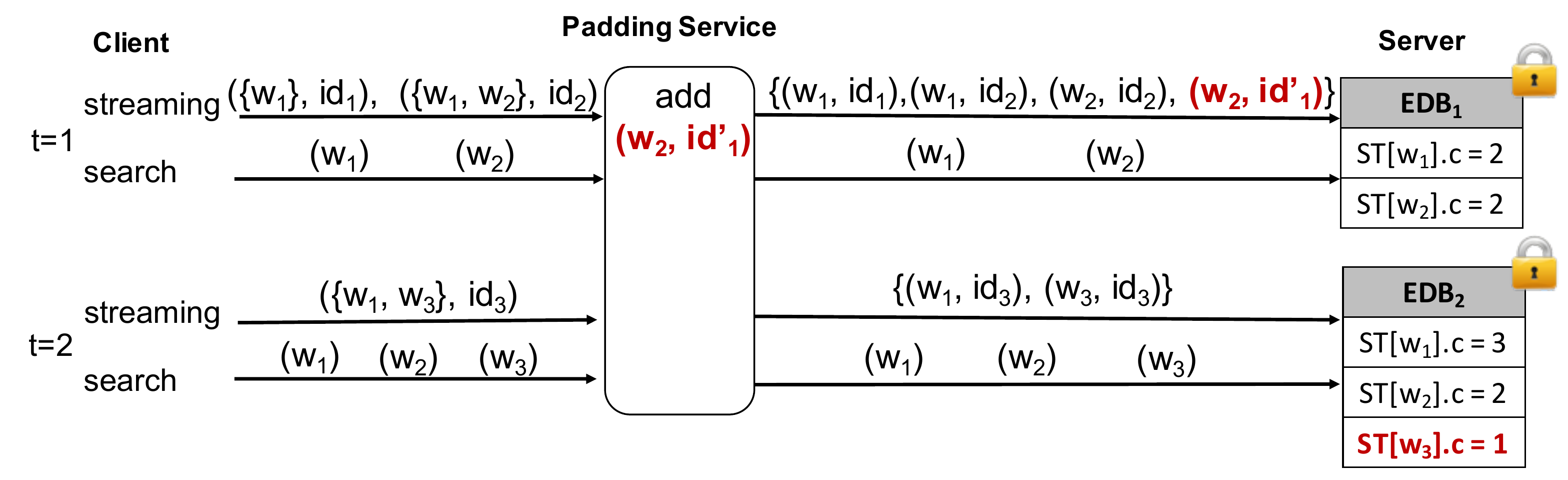} 
\vspace{-20pt}
\caption{Strawman padding against non-persistent adversary}
\vspace{-10pt}
\label{fig:np_strawman}
\end{figure}

Like many other SSE works~\cite{BostF17,RaphaelBO17,SunYLS18,Vo20} that focus on search document index, we only present that \textsf{streaming} in \system~ updates real/bogus document identifiers via $(w,id)$ pairs to the index map EDB. 
Real and padding documents containing these pairs can be uploaded separately by \textit{P} to \textit{S} via other encryption algorithms. 
In \textsf{search}, once the identifiers of real/bogus documents matching the query keyword are uncovered, \textit{S} retrieves the corresponding documents and returns them to \textit{C}.
Note that, we omit presenting the physical document management in \textit{S} in the rest of the paper since it does not affect the security of \system~against the non-persistent and persistent adversaries as proposed in Section~\ref{subsec:attackmodels}.

\begin{figure}[!t]
\centering
\includegraphics[width=0.5\textwidth]{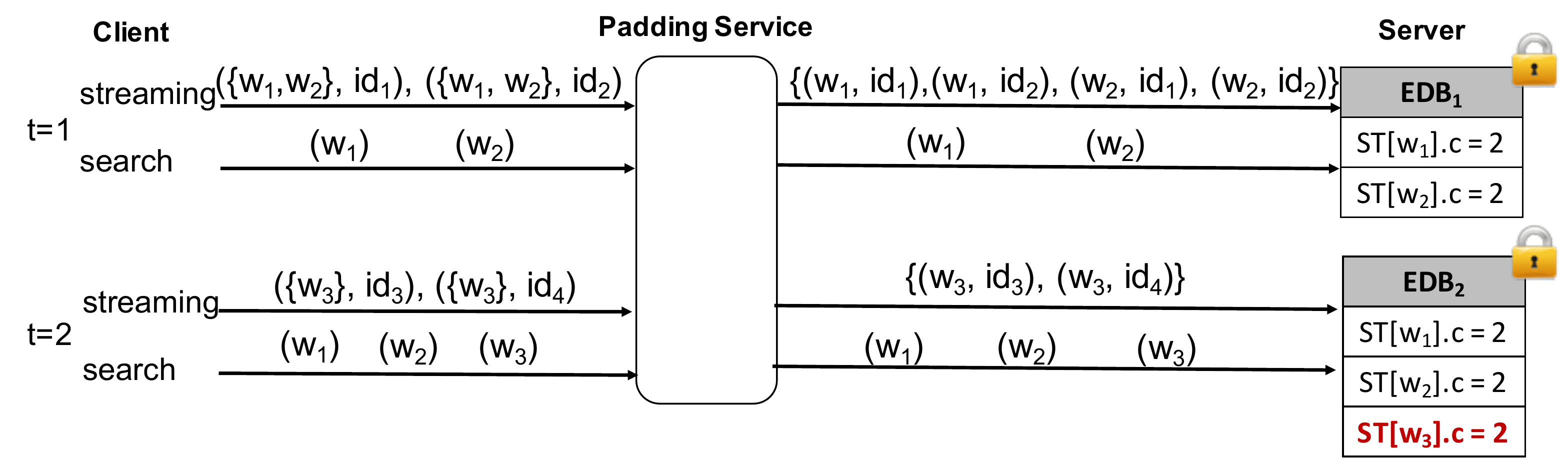} 
\vspace{-20pt}
\caption{Strawman padding against persistent adversary}
\vspace{-15pt}
\label{fig:p_strawman}
\end{figure}

\begin{figure*}
\begin{framed}
\vspace{-10pt}
\begin{multicols}{3}
\underline{\textsf{setup}}~($1^\lambda,\Delta_{stp}$)\\[6pt]
\hspace*{7pt}\textit{Client:}\\
1. Transfer dataset $\Delta_{stp}$ to \textit{P};\\
\hspace*{7pt}\textit{Padding Service:}\\
2: $\{k_1, k_2\}\xleftarrow{\$}\{0, 1\}^\lambda$;\\
3: Initialise a map $ST^\prime$ and a tuple $T$;\\  
4: Run ${\tt Setup}(\Delta_{stp})$ {\normalsize (see Section~\ref{subsec:setup})};\\
\hspace*{7pt}\textit{Server:}\\
5: Initialise an index map EDB;\\

\underline{\textsf{streaming}}~($\textsf{in}=\{(\textsf{doc},id)\}$)\\[6pt]
\hspace*{7pt}\textit{Client:}\\
1. Transfer $\textsf{in}$ to \textit{P};\\[5pt]
\hspace*{7pt}\textit{Padding Service:}\\
2: Parse $\textsf{in}$ to $M=\{(w,id)\}$;\\
//cache and check padding constraints\\
3: $V \leftarrow {\tt PaddingCheck}(M)$ (Algorithm~\ref{alg:paddingstrategy});\\
//if there is no real/bogus pairs returned by \textit{Padding Controller}\\
4: \textbf{if} $V=\{\emptyset\}$ \textbf{then}\\
5:\hspace*{9pt} return;//not sending to \textit{Server}\\
6: \textbf{else:} //perform  encryption\\  
7: \textbf{foreach} $w$ $\textbf{in}$ $V$ \textbf{do}\\
8: \hspace*{9pt}$k_e\xleftarrow{\$}\{0, 1\}^\lambda$;\\
\columnbreak

9: \hspace*{13pt}$k_w\leftarrow F(k_1,w)$;\\
10: \hspace*{9pt}$k_{id}\leftarrow F(k_{2},w)$;\\
//let $b$ is the current batch\\
11: \hspace*{9pt}\textbf{if} $ST^\prime[w]$ $\neq$ $\perp$ \textbf{then}\\
12: \hspace*{20pt}$(st_{w_{(b-1)}},c_{w_{(b-1)}})\leftarrow $ $ST^\prime[w]$;\\[2pt]
13:\hspace*{9pt}\textbf{else}\\
14:\hspace*{20pt}$st_{w_0}\xleftarrow{\$}\{0, 1\}^\lambda$, $c_{w_0} \leftarrow$ 0;\\
15:\hspace*{9pt}\textbf{end if}\\
16:\hspace*{9pt}$st_{w_{b}} \leftarrow F (k_e,st_{w_{(b-1)}});$\\
17:\hspace*{9pt}$i\leftarrow0$;\\
18.\hspace*{9pt}\textbf{foreach} $id$ that matches $w$ \textbf{do}\\
19:\hspace*{15pt}$u\leftarrow H_1(F(st_{w_{b}},i) \parallel k_w);$\\
20:\hspace*{15pt}$v\leftarrow H_2(F(st_{w_{b}},i) \parallel k_{id})\oplus id$;\\ 
21:\hspace*{15pt}$T\leftarrow T$ $\cup$ $(u,v)$;\\
22:\hspace*{15pt}$i\leftarrow i+1;$\\
23:\hspace*{9pt}\textbf{end foreach} \\
24:\hspace*{9pt}$c_{w_b}\leftarrow i$;\\ 25:\hspace*{9pt}$ST^\prime[w]\leftarrow(st_{w_{b}},{c_{w_b}})$;\\
26:\hspace*{9pt}$u_{w_b}\leftarrow H_1(F(st_{w_{b}},{c_{w_b}}) \parallel k_w);$\\
27:\hspace*{9pt}$v_{w_b}\leftarrow H_2(F(st_{w_{b}},{c_{w_b}}) \parallel k_{id}) \\ 
					\hspace*{50pt}\oplus (k_e \parallel c_{w_{(b-1)}})$;\\
28:\hspace*{9pt}$T\leftarrow T$ $\cup$ $(u_{w_b},v_{w_b})$;\\
29: \textbf{end foreach}\\
30: Send $T$ to Server;\\[3pt]
\hspace*{7pt}\textit{Server:}\\
31: \textbf{foreach} $(u,v)$  $\textbf{in}$ $T$ \textbf{do}\\
32:\hspace*{9pt} EDB$[u] = v $;\\
33: \textbf{end foreach}\\
\columnbreak

\underline{\textsf{search}}$(w)$\\[5pt]
\hspace*{7pt}\textit{Client:} //receive $ST^\prime[w]$ from \textit{P}\\
1: \textbf{if} $ST^\prime[w]$ $\neq$ $\perp$ \textbf{then}\\
2: \hspace*{9pt}$k_w\leftarrow F(k_1,w)$;\\
3: \hspace*{9pt}$k_{id}\leftarrow F(k_{2},w)$;\\
4: \hspace*{9pt}$(st_{w_{b}},c_{w_{b}})$$\leftarrow$ $ST^\prime[w]$;\\
5: \hspace*{8pt}Send {\normalsize $(k_w, k_{id},st_{w_{b}},c_{w_{b}})$} to \textit{Server}\\
6: \textbf{else}\\ 
7: \hspace*{9pt}Search $w$ in \textit{P}, return $R$;\\
8: \textbf{end if}\\
\hspace*{7pt}\textit{Server:}\\
9: \hspace*{3pt}$R \leftarrow \emptyset$, $st_{i} \leftarrow$ $st_{w_{b}}$, $c_{i} \leftarrow$ $c_{w_{b}}$;\\
10: \textbf{while} $c_i \neq 0 $ \textbf{do}\\
11: \hspace*{9pt}\textbf{for} $j=0$ to $(c_i-1)$ \textbf{do}\\
12: \hspace*{16pt}$u\leftarrow H_1(F(st_i,j) \parallel k_w)$\\
13: \hspace*{16pt}$v \leftarrow \textnormal{EDB}[u]$\\
14: \hspace*{16pt}$id\leftarrow$  $v \oplus H_2(F(st_i,j) \parallel k_{id});$\\
15: \hspace*{16pt}$R\leftarrow R$ $\cup$ $(u,v)$;\\
16: \hspace*{9pt}\textbf{end for} \\
17: \hspace*{9pt}$u_k\leftarrow H_1(F(st_i,c_i) \parallel k_w)$;\\
18: \hspace*{9pt}$v_k \leftarrow \textnormal{EDB}[u_k]$;\\
19: \hspace*{9pt}$ (k_i \parallel c_{i-1})\leftarrow$  $v_k \oplus \\
            \hspace*{70pt} H_2(F(st_i,c_i) \parallel k_{id})$\\
20: \hspace*{9pt}$st_{i-1} \leftarrow F^{-1}(k_i,st_i)$\\
21: \hspace*{9pt}$st_{i} \leftarrow st_{i-1}, c_i \leftarrow c_{i-1};$\\
22: \textbf{end while}\\
23: send $R$ to \textit{Client}\\ 

\end{multicols}
\vspace{-30pt}
\end{framed}
\vspace{-10pt}
\caption{Forward-private SSE protocols in ShieldDB.
%
%
In \textsf{streaming}, $k_e$ is an ephemeral key generated for batch insertion.
\textit{P} maintains the master state $ST^\prime[w]=(st_{w_b}, {c_{w_b}})$ for each keyword $w$, where 
$st_{w_b}$ is the master key to derive entries for $(w,id)$ pairs in the same latest batch $b$, and
${c_{w_b}}$ presents the result length $w$ (i.e., the number of real and bogus $id$s containing $w$) in that batch $b$. 
The result length of $w$ in the previous batch ${c_{w_{(b-1)}}}$ is embedded in $v_{w_b}$ (see \textsf{streaming} protocol line 26).
In \textsf{search}, ${st_i}$ and $c_i$ present the state key and the result length of $w$ in batch $i$.
$H_1$ and $H_2$ are hash functions, and $F$ is AES cipher.}
\label{fig:shieldDB_protocols}
\vspace{-15pt}
\end{figure*}

\vspace{-10pt}
\subsection{Attack Models}
\label{subsec:attackmodels}

\system~mainly considers a passive adversary who monitors the server \textit{S}'s memory access and the communication between the \textit{Server} and other participants. 
Following the assumption of the count attack~\cite{CashGP15}, the adversary has access to the background knowledge of the dataset and aims to exploit this information with the access pattern in \textsf{search} operations to recover query keywords. 
In this paper, we extend this attack model to the dynamic (streaming) setting.

Before elaborating the attack models, we define the streaming setting.
In our system, \textsf{streaming} performs batch insertion on a collection of encrypted $(w,id)$ pairs to \textit{S}. 
Giving a number of continuous \textsf{streaming} operations, encrypted collections are added to a sequence over time. 
Accordingly, \textit{S} orders the sequence of data collections by the timestamp.  
We define the gap between any two consecutive timestamps is a time interval $t$, and \textit{C} is allowed to search at any time interval. 
Note that at a given $t$, \textit{S} can only perform \textsf{search} operations against the collections that have been completely inserted to EDB.  

In the dynamic setting, we observe two new attack models, which we refer to as non-persistent and persistent adversaries, respectively.


\noindent \textbf{Non-persistent adversary}: This adversary controls \textit{S} within one single arbitrary time interval $t_{i}$, where $i$ is a system parameter that monotonically increases and $i \geq 0$. During $t_{i}$, she observes query tokens that \textit{C} issued to \textit{S}, and the access patterns returned by \textit{S}. She knows the accumulated (not separate) knowledge of the document sets inserted from $t_{0}$ to $t_{i}$. 

\noindent \textbf{Persistent adversary}:   This adversary controls \textit{S} across multiple arbitrary time intervals, for example, from $t_{0}$ to $t_{i}$. She persistently observes query tokens and access patterns at those intervals, and knows the separate knowledge of the document sets inserted from $t_{0}$ to $t_{i}$.


For both attack models, \textit{S} cannot obtain the query tokens against the encrypted data collections streamed in the current time interval. It is enforced by our streaming operation with forward privacy (see Section~\ref{subsec:optimisation}). 

\vspace{2pt}
\noindent \textbf{Strawman padding service against the adversaries}: We note that a basic \textit{Padding Service} \textit{P} that only maintains one single cache for batch \textsf{streaming} cannot mitigate the proposed adversaries as presented in Figures~\ref{fig:np_strawman} and \ref{fig:p_strawman}.

In Figure~\ref{fig:np_strawman}, we show that the non-persistent adversary, capturing query tokens of $w_1$, $w_2$, and $w_3$, and their corresponding access patterns (i.e., result lengths $ST[w_1].c$, $ST[w_2].c$, and $ST[w_3].c$) at time $t=2$, can uncover which tokens used for what keywords if she has the corresponding background knowledge of DB at time $t=2$ (i.e., DB$_2$). The reason is due to the unique result lengths introduced in EDB at time $t=2$ (i.e., EDB$_2$) when \textit{P} adds bogus pairs to equalise the number of pairs for keywords sent to EDB during every \textsf{streaming} operation.

In Figure~\ref{fig:p_strawman}, we demonstrate that the persistent adversary can detect when new keywords are inserted in EDB. For example, she might know the states of the database DB at time $t=1$ and $t=2$ as her background knowledge, and the query results of $w_3$ in EDB$_1$ and EDB$_2$ are different. Then, she knows the occurrence of a new keyword $w_3$ is introduced in EDB$_2$ at $t=2$. Then, she is able to identify the query token of $w_3$ during \textsf{search} at $t=2$.

\vspace{2pt}
\noindent \textbf{Real-world implication of the adversaries}: We note that the proposed attack models are new for leakage-abuse attacks, which have not been investigated and formalised in any of the prior works~\cite{IslamKK12,CashGP15,BlackstoneKM20}. We stress that non-persistent adversary could be any external attackers, e.g., hackers or organised cyber criminals. They might compromise the server at a certain time window. We also assume that this adversary could obtain a snapshot of the database via public channels, e.g., a prior data breach~\cite{NaveedKW15}.  
Because the database is changed dynamically, the snapshot might only reflect some historical state of the database. 
On the contrary, the persistent adversary is more powerful and could be database administrators or insiders of an enterprise. They might have long term access to the server and could obtain multiple snapshots of the database via internal channels.

\vspace{2pt}
\noindent \textbf{Other threats}: Apart from the above adversaries, \system~ considers another specific rational adversary~\cite{ZhangKP16} who can inject documents to compromise query privacy. As mentioned, this threat can be mitigated via forward privacy SSE. Note that \system~currently does not address an active adversary who sabotages the search results. 

\vspace{-5pt}
\section{Design of ShieldDB} 
\label{sec:construction}

In this section, we present the detailed design of \system~ in \textsf{Setup}, \textsf{Stream}, and \textsf{Search}. 
%
%
Then, we present some advanced features of \system~ to further improve the security and efficiency. 

\vspace{-10pt}
\subsection{Setup}
\label{subsec:setup}

We consider $\Delta_{stp}=\{w_1,w_2,...,w_l\}$ is the training dataset for the system. During ${\tt Setup}(\Delta_{stp})$, \textit{P} invokes \textit{Cache Controller} to initialise the cache for batch insertion, and \textit{Padding Controller} to generate bogus documents for padding.

To reduce padding overhead, \system~implements cache management in a way that it groups keywords with similar frequencies together and performs padding at each individual keyword cluster. We denote  $L=\{L_1, \cdots, L_m\}$ as caching clusters managed by \textit{Cache Controller}, where $m$ is the number of cache clusters. This approach is inspired from existing padding countermeasures in the static setting~\cite{CashGP15,BostF17}. 
The idea of doing this in a static database is intuitive; the variance between the result lengths of keywords with similar frequencies is small, which can minimize the number of bogus entries added to the database. 
We note that it is also reasonable in the dynamic setting, where the keyword frequencies in specific applications can be stable in the long run. If a keyword is popular, it is likely to appear frequently during \textsf{streaming}, and vice versa.  
Therefore, we assume that the existence of the training dataset, where the keyword frequencies are close to the real ones during \textsf{streaming} is reasonble (see Section~\ref{subsec:evaluation} for that distribution evaluation). 
%
%
We further suggest alternative training data collection approaches in Section~\ref{sec:discussion}. 
%

\vspace{-2pt}
%

Given $\Delta_{stp}$, \textit{Cache Controller} partitions keywords based on their frequencies by using a heuristic algorithm. The objective function in Eq.~\ref{eq:1}, such that 
the clustering can be formed as $\left[\left(w_1,\ldots,w_i\right),\left(w_{i+1},\ldots,w_j\right),\ldots,\left(w_k,\ldots,w_l\right)\right]$.
We note that the minimum size of each group $\alpha$ is subjected to $\alpha \geq 2$. 
For security, the keyword frequency in each cluster after padding should be the same, i.e., the maximum one, and thus \textit{Cache Controller} computes the padding overhead $\gamma$ as follows: 
\begin{equation}
\label{eq:1}
\begin{split}
\gamma = &\left({i*{f_{{w_i}}} - \sum\limits_{t = 1}^i {{f_{{w_t}}}} } \right)  + \left( {(j - i)*{f_{{w_j}}} - \sum\limits_{t = i + 1}^j {{f_{{w_t}}}} } \right) + \\
&\ldots+\left({(l - k - 1)*{f_{{w_l}}}-\sum\limits_{t=k}^l {{f_{{w_t}}}}}\right)
\end{split}
\raisetag{20pt}
\end{equation}
This algorithm iterates evaluating $\gamma$ for every combination of the partition. 
We denote by $m$ the number of clusters. 
After that, the \textit{Cache Controller} allocates the capacity of the cache based on the  aggregated keyword frequencies of each cluster, i.e., $|L| \sum\limits_{t = 1}^i {{f_{{w_t}}}}$, $|L| \sum\limits_{t = i + 1}^j {{f_{{w_t}}}}$, $\ldots$, $|L|  \sum\limits_{t=k}^l {{f_{{w_t}}}}$, where $|L|$ is the total capacity assigned for the local cache. 
We denote by $L_i.threshold()$ the function that outputs the caching capacity of cluster $L_i$.

After that, \textit{Padding Controller} initialises a bogus dataset $B$ with size $|B|$, where the number of bogus keyword/id pairs for each keyword $w_i$ is determined via the frequency, i.e., $|B|  (f_{w} - f_{w_i})$, where $f_{w}$ is the maximum frequency in the cluster of $w_i$.  
The reason of doing so is that it still follows the assumption in cache allocation. 
If the keyword is less frequent in a cluster, it needs more bogus pairs to achieve the maximum result length after padding, comparing other keywords with higher frequency, and vice versa. 
Then the controller generates bogus index pairs. Once the bogus pairs for a certain keyword $w_{i}$ is run out, the controller is invoked again to generate padding for it through the same way. 

\begin{algorithm}[!t]
    \DontPrintSemicolon
    \SetKwInOut{Input}{Input}
    \SetKwInOut{Output}{Output}
    \textbf{function} ${\tt PaddingCheck}$()\\
    \Input{$M=\{(w,id)\}$: entries for streaming,\\ $\{L_1, \cdots, L_m\}$: cache clusters, \\ $ST$: a map that tracks keyword states, \\ $B$: bogus document set; \\
    $mode$: padding mode (\emph{high} or \emph{low});\\
    }
    \Output{$V$: a set of real and bogus entries}

push entries in $M$ to $\{L_1, \cdots, L_m\}$;\\
$V \leftarrow \{\emptyset\}$;\\
	\If{padding against non-persistent adversary} 
	{	\For{cluster $L_i$ $\in$ $\{L_1, \cdots, L_m\}$}
		{
			\If{$L_i.capacity() \geq L_i.threshold()$} 
			{
				\For{$w \in L_i $}
		        {
    				\If {$ ST[w].flag = false $}
    				{
    				  skip padding for $w$ when executing $\tt{PaddingByMode}()$; //not occurred yet\\
    				}
    			}
				$M_i \leftarrow {\tt PaddingByMode}(L_i,ST,B,mode)$;\\	
				add $M_i$ to $V$;\\ 		
			}
		}
	}
	 \If{padding against persistent adversary} 
	 {
		\For{cluster $L_i$ $\in$ $\{L_1, \cdots, L_m\}$}
		{
			\uIf{$L_i$.firstBatch=true \& $ST[w].flag = true$ \textnormal{for} $\forall w\in L_{i}$ } 
			{

				$M_i \leftarrow {\tt PaddingByMode}(L_i,ST,B,mode)$;\\	
				add $M_i$ to $V$;\\ 
			}
			\ElseIf{$L_i.capacity() \geq L_i.threshold()$}
			{

				$M_i \leftarrow {\tt PaddingByMode}(L_i,ST,B,mode)$;\\	
				add $M_i$ to $V$;\\	
			}		
		}
	 }
return $V$;\\
\caption{Padding strategies}
   \label{alg:paddingstrategy}
\end{algorithm}

\begin{algorithm}[!t]
    \DontPrintSemicolon
    \SetKwInOut{Input}{Input}
    \SetKwInOut{Output}{Output}
    \textbf{function} ${\tt PaddingByMode}$($L_i,ST,P,mode$)\\
    \Input{$L_i$: cluster for padding, \\ $ST$: a map that tracks keyword states, \\ $B$: bogus dataset; \\
    $mode$: padding mode (\emph{high} or \emph{low})}
    \Output{$M_i$: a set of real and bogus entries}

$M_i \leftarrow \emptyset$;\\
 $st_{max}$ $\leftarrow$ $\max \{ST[w].c\}$ for $\forall w\in L_{i}$;\\

let $c_w$ is the length of the currently matching list of $w$ cached in $L_i$; \\
	\eIf{$mode = high$} 
	{
	    let $c_{max}$ is max$\{c_w,\forall w \in L_i\}$;\\
		$C \leftarrow st_{max} + c_{max}$;\\
		
		\For{ $w \in L_i ~\textnormal{with}~ST[w].flag=true$ } 
		 {
		 add $(C - c_w)$ bogus entries from $B[w]$ to $M_i$; \\
		 add all $c_w$ cached entries of $w$ in $L_i$ to  $M_i$; \\
		 $ST[w].c \leftarrow  C$; \\
		 }
	}
	{ // $mode$ is $low$\\
	    let $c_{min}$ is min$\{c_w >0,\forall w \in L_i\}$; \\
		$C \leftarrow st_{max} + c_{min}$;\\
		\For{ $w \in L_i ~\textnormal{with}~ST[w].flag=true$ } 
		 {
		 $m \leftarrow C - ST[w].c;$\\ 
    		 \eIf{$m > c_w$}
    		 {
    		 add $(m-c_w)$ bogus entries from $B[w]$ to $M_i$; \\
    		 add all $c_w$ cached entries of $w$ in $L_i$ to  $M_i$ ; \\
             }
             {
             //do not add bogus entries for $w$;\\
             put $(m)$ cached entries of $w$ in $L_i$ to $M_i$;\\
             //the remaining $(c_w - m)$ entries of $w$ are still cached in $L_i$;\\
             }
        $ST[w].c \leftarrow  C$; \\
 	    }
	}
return $M_i$;\\  

\caption{Padding modes}
   \label{alg:paddingmode}
\end{algorithm}

\vspace{2pt}
\noindent \textbf{Remark}: We assume that the distribution of the sample dataset is close to the one of the streaming data in a running period.
We acknowledge that it is non-trivial to obtain an optimal padding overhead in the dynamic setting due to the variation of streaming documents in different time intervals. Nevertheless, if the distribution of the database varies during the runtime, the \textsf{setup} can be re-invoked. Namely, keyword clustering algorithm can be re-activated based on the up-to-date streaming data (e.g., in a sliding window), and the cache can be re-allocated. Additionally, our proposed \textsf{re-encryption} operation can further reduce the padding overhead (see Section~\ref{subsec:optimisation}) if the streaming distribution only differs on particular keyword clusters.
We discuss the distribution difference detection in Sections~\ref{subsec:evaluation} and~\ref{sec:discussion}. 
We also note that there are  applications and scenarios where the distribution does not vary much, like IoT streaming data for environment sensors. In such applications, the range of numbers/indicators are already specified by the vendors.

\subsection{Padding Strategies}
\label{subsec:padding}

During \textsf{streaming}, documents are continuously collected and parsed as  $M=\{(w,id)\}$ in \textit{P}.
Then, \textit{P} executes ${\tt PaddingCheck}(M)$ to cache and check padding constraints.
In details, these $(w,id)$ pairs are cached at their corresponding clusters by \textit{Cache Controller}. 
Once a cluster $L_i$ is full, \textit{Padding Controller} adapts the corresponding padding strategy to the targeted adversary, encrypts and inserts all real and bogus pairs to EDB in a batch manner.   
We elaborate on the padding strategies against the non-persistent and persistent adversaries, respectively. The details are given in Algorithms~\ref{alg:paddingstrategy} and~\ref{alg:paddingmode}.
%
%

\vspace{2pt}
\noindent \textbf{Padding strategy against the non-persistent adversary}: Recall that this adversary controls \textit{S} within a certain time interval $t$. 
From the high level point of view, an effective padding strategy should ensure that all keywords occurred in EDB at $t$ do not have unique result lengths. 
There are two challenges to achieve this goal. 
First, $t$ can be an arbitrary time interval. Therefore, the above guarantee needs to be held at any certain time interval. 
Second, not all the keywords in the keyword space would appear at each time interval. It is non-trivial to deal with this situation to preserve the security of padding.

\begin{figure}
\centering
\includegraphics[width=0.5\textwidth]{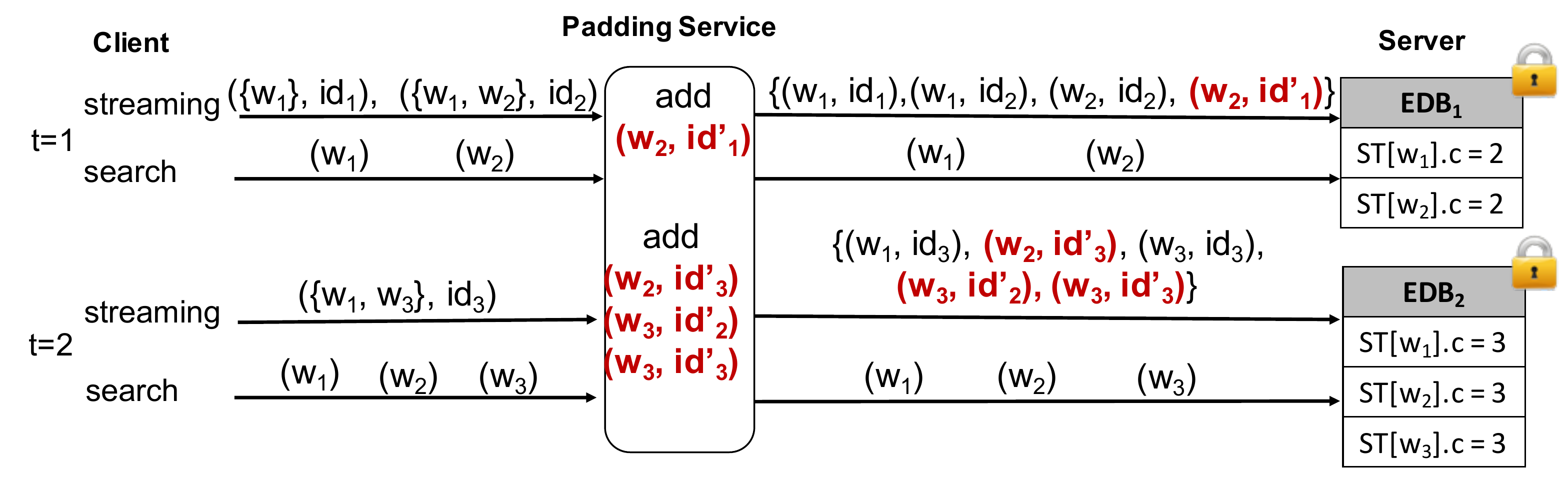} 
\vspace{-20pt}
\caption{\textit{High} mode padding against non-persistent adversary}
\vspace{-15pt}
\label{fig:np_shieldDB}
\end{figure}

To address the above challenges, \system~programs \textit{Padding Controller} to track the states of keywords over the time intervals from the beginning. Specifically, each keyword state $ST[w]$ includes two components, a flag $ST[w].flag$ that indicates whether the keyword has existed before in the streamed documents, and a counter $ST[w].c$ that presents the number of total real and bogus $(w,id)$ pairs already uploaded in EDB of the keyword $w$.
Note that $ST[w].flag=true$ is kept permanently once $w$ has existed in the documents streamed to the server.
\textit{Padding Controller} only pads the keywords in a cluster $L_i$ if the number of cached real $(w,id)$ pairs of the cluster, denoted by $L_i.capacity()$, exceeds $L_i.threshold()$ defined in  \textsf{setup} (see Algorithm~\ref{alg:paddingstrategy} line 6).

Based on the states of keywords, \textit{Padding Controller} performs the following actions.
If the keyword has not existed yet, the controller will not pad it even its cluster is full (see Algorithm~\ref{alg:paddingstrategy} line 8). 
The reason is that the adversary might also know the information of keyword existence. 
If \textit{C} queries a keyword which does not exist, \textit{S} should return an empty set. 
Otherwise, the adversary can identify the token of this keyword if padded. 
Accordingly, only when a keyword $w$ appears at the first time (i.e., $ST[w].flag=true$), padding over this keyword will be invoked  (see Algorithm~\ref{alg:paddingmode} lines 8 and 17). 
Once $ST[w].flag=true$,  this keyword will always get padded later, as long as the cache of its cluster $L_i.capacity() \geq L_i.threshold()$, no matter it exists in a certain time interval or not. The padding ensures that all existing keywords in the cluster always have the same result length at any following time interval. 

\vspace{2pt}
\noindent \textbf{Padding strategy against the persistent adversary}: Recall that this adversary can monitor the database continuously and obtain multiple references of the database across multiple time intervals. 
Likewise, the padding strategy against the persistent adversary should ensure that all keywords have no unique access pattern in all time intervals from the very beginning. 
However, directly using the strategy against the non-persistent adversary here does not address the leakage of keyword existence. 

To address this issue,  \textit{Padding Controller} is programmed to enforce another necessary constraint to invoke padding. That is, all keywords in the cluster at the first batch have to exist before streaming.
Formally, we let $L_i.firstBatch$ be the constraint that evaluates the existence of all predefined keywords in $L_i$. Then, $L_i.firstBatch=true$ implies $ST[w].flag = true$ for $\forall w \in L_i$. 
The constraint remains $false$ if there is $\exists w_j \in L_i$, $ST[w_j]=false$.
As a trade-off, \textit{Cache Controller} has to hold all the pairs in the cluster even the cache is full, if there are still keywords yet to appear (see Algorithm~\ref{alg:paddingstrategy} line 18). 
%
%
For subsequent batches of the cluster, the padding constraint follows the same strategy for the non-persistent adversary (see Algorithm~\ref{alg:paddingstrategy} line 21).

\noindent \textbf{Padding modes}: \system~implements two modes for padding, i.e., \emph{high} and \emph{low} modes. These two modes both are applicable to the above two padding strategies (see $mode$ in Algorithm~\ref{alg:paddingstrategy}). The padding mechanism of these modes are described in Algorithm~\ref{alg:paddingmode}. In the \emph{high} mode, once the constraint for the  cache of a cluster is met, the keywords to be padded have the maximum result length of keywords in this cluster (see Algorithm~\ref{alg:paddingmode} lines 9-11). Accordingly, the cache can be emptied since all entries are sent to \textit{Padding Controller} for streaming. 
On the contrary, the \emph{low} mode is invoked in a way that the keywords to be padded have the minimum result length of keywords in this cluster. Therefore, some entries of keywords might still be remained in the cache (see Algorithm~\ref{alg:paddingmode} lines 24-25). Yet, this mode only introduces necessarily minimum padding for keywords which do not occur in random time intervals. 
The two modes have their own merits. The \emph{high} mode consumes a larger amount of padding and execution time for padding and encryption, but it reduces the load of cache in \textit{P}. In contrast, the \emph{low} mode introduces relatively less padding overhead but heavier load of \textit{P}.

\begin{figure}
\centering
\includegraphics[width=0.5\textwidth]{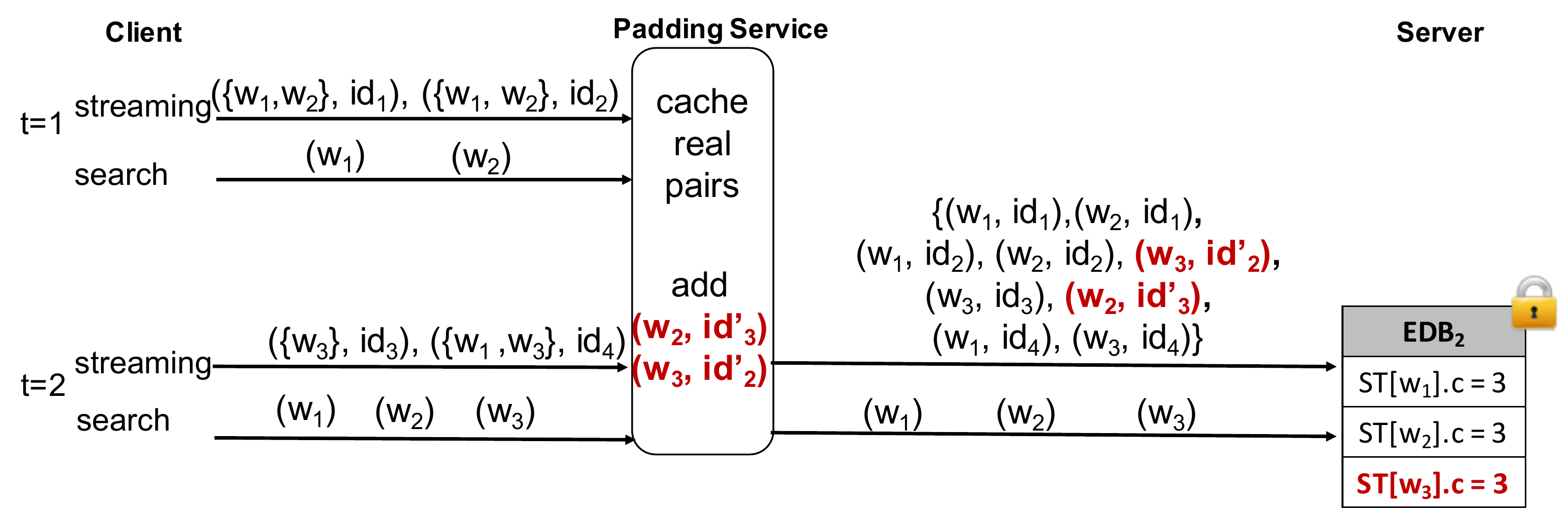} 
\vspace{-20pt}
\caption{\textit{High} mode padding against persistent adversary}
\vspace{-15pt}
\label{fig:p_shieldDB}
\end{figure}

Figure~\ref{fig:np_shieldDB} demonstrates the padding strategy against the non-persistent adversary by using \textit{high} padding mode.
Given a cluster $L_j$ containing three keywords $w_i, \forall i \in [1,3]$, \textit{P} tracks their keyword state $ST[w_i].c$ and $ST[w_i].flag$ and applies padding when real cached pairs exceed the capacity of the cluster.
Then, \textit{P} adds bogus pairs to ensure they have the same result length in EDB at $t=1$ and $t=2$. 
Note that, there is no padding applied for $w_3$ at time $t=1$ due to $ST[w_3].flag=false$ at that time. 
At $t=2$, \textit{P} pads all the keywords to the maximum result length among them in the cluster (i.e., $ST[w_i]=3, \forall i \in [1,3]$).

In Figure~\ref{fig:p_shieldDB}, we demonstrate the padding strategy against the persistent adversary by using \textit{high} padding mode for a cluster $L_j=\{w_i\}, \forall i \in [1,3]$.
\textit{P} only performs padding and inserts encrypted real/bogus entries to EDB when all keywords in $L_j$ have appeared (i.e., waiting for $w_3$ occurs at $t=1$).
Querying any $w_i$ of the cluster prior this time does not make the client \textit{C} send query tokens to \textit{S}.
The reason is because \textit{C} only receives $ST^\prime[w_i] =\perp$ sent by \textit{P} to do \textsf{search}.
We note that \textit{P} only updates $ST^\prime[w_i]$, (i.e., keyword state of $w_i$ for encryption) (see $ST^\prime[w]$ in \textsf{streaming} and \textsf{search} in Figure~\ref{fig:shieldDB_protocols}) when $w_i$ in the cluster is ready for encryption after padding applied.
We note that $L_j.firstBatch=true$ once all keywords in $L_j$ have been appeared. %
Then, subsequent batches at $t>1$ of $L_j$ does not need to check keyword occurrence.
Instead, it follows the padding strategy addressing the non-persistent adversary.
The reason is because $w_i$ is always padded in the following batches once it had appeared at the first time.
The strategy aims to ensure there is no new keyword introduced at any random time interval to address the persistent adversary.

\begin{table}
\caption{Time complexity of \textit{Padding Service} and \textit{Server}}
\begin{center}
\small
\begin{tabular}{|c|c|c|}
\hline
\multicolumn{2}{|c|}{\textit{Padding Service}} & \multicolumn{1}{c|}{\textit{Server}} \\ \hline 
{\textbf{padding}}&{\textbf{encryption}} &  \multirow{2}{*}{\textbf{update/search}} \\ 
(\textsf{streaming})& (\textsf{streaming})&  \\ \hline
$|B|(f_w-f_{w_i})$ &  $\mathcal{O}(n_r + n_b)$ &  $\mathcal{O}(n_r + n_b)$ \\ \hline
\end{tabular}%
\end{center}
\vspace{-15pt}
\label{table:complexity}
\end{table}

\vspace{2pt}
\noindent \textbf{Complexity analysis of padding}:
We note that during \textsf{setup}, the \textit{Padding Service} initialises a padding (bogus) dataset \textit{B} with size of $|B|$. Then, the total number of bogus pairs used for the keyword $w_i$ during the \textsf{streaming} phase is $|B|(f_w-f_{w_i})$, where  $f_w$ is the maximum frequency in the cluster of $w_i$ found in the \textsf{setup}. The asymptotic complexity to encrypt real/bogus pairs of $w_i$ in the \textsf{streaming} is $\mathcal{O}(n_r + n_b)$, where $n_r$ and $n_b$ present the total number of real and bogus pairs of $w_i$, respectively. In \textsf{streaming} (resp. \textsf{search}), upon receiving the update (resp. query) token(s) of $w_i$, the \textit{Server} inserts (resp. retrieves) the entries (resp. search result) based on encrypted labels of $w_i$ from the map EDB (i.e., $\mathcal{O}(n_r + n_b)$) by following the forward-private SSE protocols (see Figure~\ref{fig:shieldDB_protocols}). The performance is summarised in Table~\ref{table:complexity}.

\vspace{2pt}
\noindent \textbf{Security guarantees}:
Our padding countermeasures ensure that no unique access pattern exists for keywords which have occurred in EDB. 
For the persistent adversary, the padding countermeasure also ensures that the keyword occurrence is hidden across multiple time intervals. 
Note that padding not only protects the result lengths of queries, but also introduces false counts in keyword co-occurrence matrix, which further increases the efforts of the count attack. 
Regarding the formal security definition, we follow a notion recently proposed by Bost et al.~\cite{BostF17} for SSE schemes with padding countermeasures. 
This notion captures the background knowledge of the adversary and formalises the security strength of padding. That is, given any sequence of query tokens, it is efficient to find another same-sized sequence of query tokens with identical leakage. 
We extend this notion to make the above condition hold in the dynamic setting in Section~\ref{subsec:security_shielddb}.

\vspace{4pt}
\noindent \textbf{Remark}: Our padding strategies are different from the approach proposed by Bost et al.~\cite{BostF17}, which merely groups keywords into clusters and pads them to the same result length for a static database. 
Directly adapting their approach for different batches of incoming documents will fail to address persistent or even non-persistent adversaries. 
The underlying reason is that the above approach treats each batch individually, while the states of database are  accumulated. 
Effective padding strategies in the dynamic setting must consider the accumulated states of the database so that the adversaries can be addressed in arbitrary time intervals.


\vspace{-7pt}
\subsection{Other Features}
\label{subsec:optimisation}

\system~provides several other salient features to enhance its security, efficiency, and functionality. 

\vspace{2pt}
\noindent \textbf{Forward privacy}: In \textsf{streaming} and \textsf{search}, \system~ realises the notion of forward privacy~\cite{Bost16,SongDYXZ17} to protect the newly added documents and  mitigate the injection attacks~\cite{ZhangKP16}. 
In particular, our system customises an efficient SSE scheme with forward privacy~\cite{SongDYXZ17} to our context of batch insertion.
The detailed algorithm for encryption and search can be found in Figure~\ref{fig:shieldDB_protocols}.
Our forward-private scheme is built on symmetric-key based trapdoor permutation and is faster than the public-key based solution~\cite{Bost16}. 
The ephemeral key $k_e$ of permutation is embedded inside the index entry to recover the state $(st_{w_{(b-1)}}, c_{w_{(b-1)}})$ of the previous entries in batch $(b-1)$ (see Figure~\ref{fig:shieldDB_protocols}, lines 16-27 in \textsf{streaming}, and lines 17-21 in \textsf{search} protocol).
To reduce the computation and storage overhead, we link a master state $ST^\prime[w]=(st_{w_b},c_{w_b})$ to a set of entries with the same keyword in the batch $b$ (see \textsf{streaming} at lines 19-20).

Upon receiving $ST^\prime[w]$ of the query keyword $w$ sent from \textit{P}, \textit{C} generates a query token and sends to \textit{S} (see Figure~\ref{fig:shieldDB_protocols} line 5 in \textsf{search} protocol). We note that $ST^\prime[w]$ is different from the state $ST[w]$ used for padding in \textit{Padding Controller}.
The benefit of our forward-private design is that \textit{S} can be enforced to perform \textsf{search} operations over the completed batches. The batches which are still transmitted on the fly cannot be queried without the latest keyword state $ST^\prime[w]$ from \textit{C}.

\vspace{2pt}
\noindent \textbf{Re-encryption and deletion}: \system~also implements the \textsf{re-encryption} operation. This operation is periodically conducted over a certain keyword cluster. \textit{Padding Service} \textit{P} first fetches all entries in this cluster stored in EDB from \textit{S}. After that, \textit{P} removes all bogus entries and re-performs the padding over this cluster of keywords. All the real and bogus entires are then encrypted via a fresh key, and inserted back to EDB. The benefits of \textsf{re-encryption}  are two-fold: (1) redundant bogus entries in this cluster can be eliminated; and (2) the leakage function can be reset to protect the search and access patterns. 
During \textsf{re-encryption}, \system~can also execute deletion. A list of deleted document ids is maintained at \textit{P}, and the deleted entries are physically removed from the cluster before padding. 

\vspace{2pt}
\noindent \textbf{Cache flushing}: During \textsf{streaming}, the keywords in some clusters might not show up frequently. Even the cache capacity of such clusters is set relatively small, the constraint might still not be triggered very often. 
To reduce the load of the cache at \textit{P} and improve the streaming throughput, \system~develops an operation called \textsf{flushing} to deal with the above ``cold'' clusters. In particular,  \textit{Cache Controller} monitors all the caches of clusters, and sets a time limit to trigger \textsf{flushing}. If a cluster is not full after a period of this time limit, all entries in this cluster will be sent to \textit{Padding Controller}. Note that the padding strategies still need to be followed for security and the \emph{high} mode of padding is applied to empty the cache.

\vspace{-10pt}
\section{Security of ShieldDB}
\label{subsec:security_shielddb}

\system~implements a dynamic searchable encryption scheme $\Sigma = (\textsf{setup}, \textsf{streaming}, \textsf{search})$, consisting of three protocols between a padding service $P$, a storage server $S$, and an querying client $C$.
We note that a database DB$_t = (w_i,id_i)_{i=1}^{|DB_t|}$ is defined as a tuple of keyword and document id pairs with $w_i \subseteq \{0,1\}^*$ and $id_i \in \{0,1\}^l$ at the time interval $t \geq0$.
DB$_t$ presents the accumulated database up to time $t$.
The security of ShieldDB can be quantified via a leakage function $\mathcal{L}=(\mathcal{L}^{Stp},\mathcal{L}^{Stream},\mathcal{L}^{Srch})$. 
It defines the information exposed in \textsf{setup}, \textsf{streaming}, and \textsf{search}, respectively.
The security of~\system ~ensures that it does not reveal any information beyond the one that can be inferred from $\mathcal{L}^{Stp}$, $\mathcal{L}^{Stream}$, and $\mathcal{L}^{Srch}$.
%
%

We adapt the notion of \textit{constrained security}~\cite{BostF17} to formalise the knowledge the information known by the non-persistent and persistent adversaries.
Formally, we  model the fact that the non-persistent adversary knows the database by considering the constraint set $\mathfrak{C}^{\mathcal{DB}_t}=\{C^{\textnormal{DB}_{t}},{\textnormal{DB}_{t}}\in \mathcal{DB}_t\}$, where $\mathcal{DB}_t$ is a computable database set at the time $t$ (see Appendix A-B and A-C).
Then, with the knowledge of ${\textnormal{DB}_t}$, the leakage function of $\Sigma_{NP}$ is formally defined as $\mathcal{L}_{NP}=(\mathcal{L}^{Stp},\mathcal{L}^{Srch},\mathcal{L}^{\alpha-pad})$, where $\mathcal{L}^{Stp}$ and $\mathcal{L}^{Srch}$ reveals the leakage in \textsf{Setup} at $t$ and \textsf{Search} against $\textrm{EDB}_t$, respectively, and the new leakage $\mathcal{L}^{\alpha-pad}$ reveals the minimum size of clusters induced by $\mathcal{L}^{\alpha-pad}$.
We define the {\normalfont Ind$_{SSE,\mathcal{A},\mathcal{L}_{NP}, \mathfrak{C}^{\mathcal{DB}_t},\alpha}$} game for the non-persistent in Definition 5 in the supplementary material.

We state the following theorem regarding the security of \system~ against the adversary as follows (full proof is provided in the supplementary material Section A-C) .
\begin{theorem}\label{theo:theorem_non_persistent}
Let $\Sigma_{NP}$ = {\normalfont(\textsf{Setup,Search})} be our SSE scheme, and $\mathfrak{C}^{\mathcal{DB}_t}$ a set of knowledge constraints. If $\Sigma_{NP}$ is $\mathcal{L}_{NP}$-constrained-adaptively-indistinguishable secure, and  $\mathfrak{C}^{\mathcal{DB}_t}$ is   $(\mathcal{L}_{NP},\alpha)$-acceptable, then $\Sigma_{NP}$ is ($\mathcal{L}_{NP}$,$\mathfrak{C}^{\mathcal{DB}_t},\alpha$)-constrained-adaptively-indistinguishability secure.
\end{theorem}

\begin{figure}
\centering
\includegraphics[width=0.49\textwidth,height=3.6 cm]{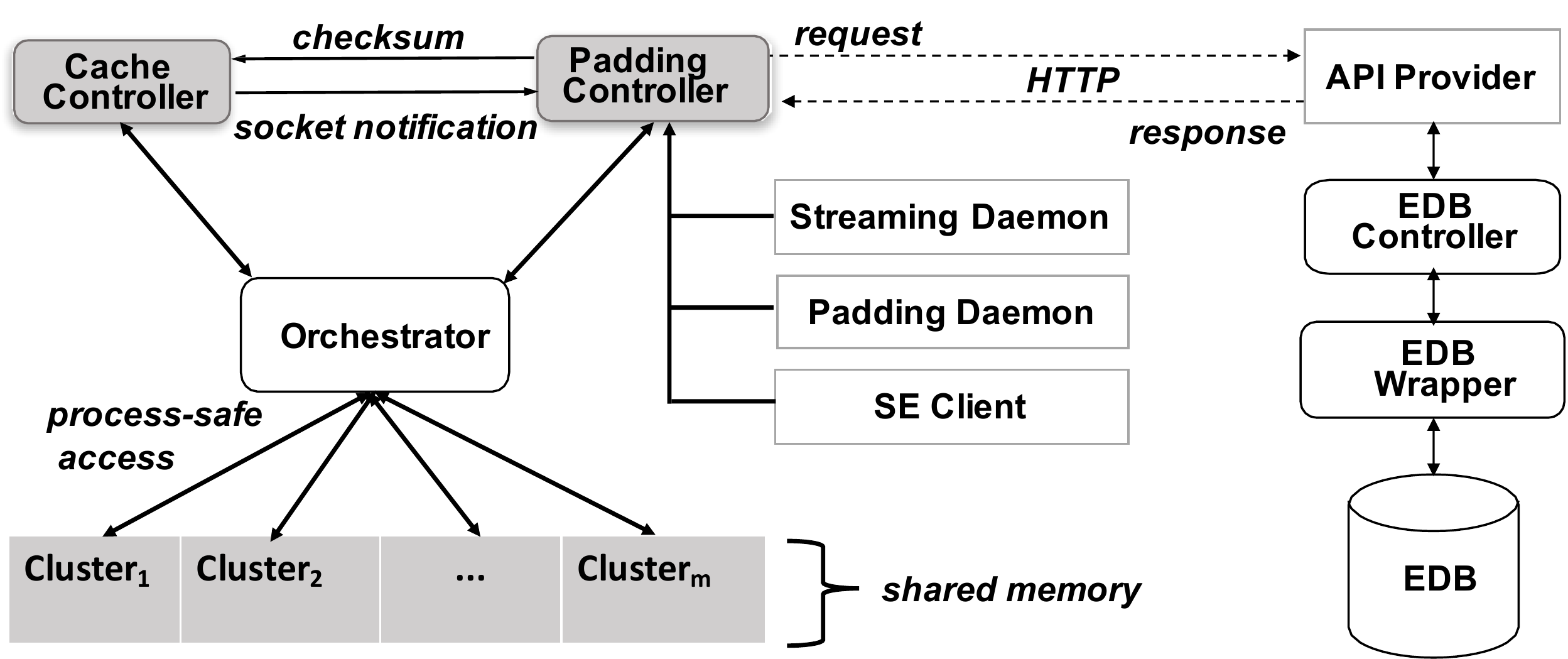}
\caption{Implementation of \system}
\vspace{-20pt}
\label{fig:implementation}
\end{figure} 

Then, we generalise the knowledge of the non-persistent adversary over the time to be the persistent adversary's knowledge of databases. 
Namely, we denote by $(\textnormal{DB}_{t=0},\dots, \textnormal{DB}_{t=T})$ such knowledge, where $T$ denotes the streaming period.
We also make use of the constraint set $\mathfrak{C}^{\mathcal{DB}_t}=\{C^{\textnormal{DB}_{t}},{\textnormal{DB}_{t}}\in \mathcal{DB}_t\}$ to formulate
$\mathbf{C}^{[1,\dots,T]}=\{\mathfrak{C}^{\mathcal{DB}_0},\dots,\mathfrak{C}^{\mathcal{DB}_T}\}$ the generalisation of constraint sets over the period such that we know every $\mathfrak{C}^{\mathcal{DB}_t}(\mathcal{L}_{NP},\alpha)$-acceptable, $\forall t\in [0,T]$.

The security of a scheme $\Sigma_{P}=(\textsf{Setup},\textsf{Stream},\textsf{Search})$ against the persistent adversary over a streaming period $T$.
We start adding a padding mechanism against the persistent adversary in Algorithm 1 (i.e., Padding Strategies) to $\Sigma_{P}$ such that, $\forall G_{i,t}$ in $\Gamma_t=\{G_{1,t},\dots,G_{m,t}\}$ induced by $\textsf{Search}_t$ (i.e., searching at time $t$) against  $\textrm{EDB}_t$, $G_{i,t}$ always has the same size $|G_{i,t}|=|G_{i,t^\prime}|$, where $G_{i,t^\prime}$ in $\Gamma_{t^\prime}=\{G_{1,t^\prime},\dots,G_{m,t^\prime}\}$, $\forall t^\prime \neq t$.
%
%
Let $\mathcal{L}^{Stream}_{[1,\dots,T]}=\{{L^{Stream}_1},\dots,{{L^{Stream}_T}}\}$, and $\mathcal{L}^{Search}_{[1,\dots,T]}=\{\mathcal{L}^{Srch}_{NP,i}\}$, for $i\in [1,\dots,T]$, where $\mathcal{L}^{Srch}_{NP,i}=(\mathcal{L}^{Srch}_i,\mathcal{L}^{\alpha-pad}_i)$ is the search leakage at time $i$, where  $\mathcal{L}^{\alpha-pad}_i)$ reveals the sizes of clusters induced by $\mathcal{L}^{\alpha-pad}_i$. 
Then, the leakage function of $\Sigma_{P}$ can be quantified via the leakage function $\mathcal{L}_{P}=(\mathcal{L}^{Stp},\mathcal{L}^{Stream}_{[1,\dots,T]},\mathcal{L}^{Search}_{[1,\dots,T]})$.
The {\normalfont Ind$_{\mathrm{DSSE},\mathcal{A},\mathcal{L}_{P}, \mathbf{C}^{[1,\dots,T]},\alpha,\mathcal{F}}$} game is given in Definition 7.

We state the below theorem regarding the security  of \system~ against the \textit{persistent} adversary. Security proof is provided in the supplementary material in Section A-D).
\begin{theorem}\label{theo:theorem_persistent}
Let $\Sigma_{P}$ = {\normalfont(\textsf{Setup,Streaming,Search})} be our DSSE scheme, and $\mathbf{C}^{[1,\dots,T]}=\{\mathfrak{C}^{\mathcal{DB}_0},\dots,\mathfrak{C}^{\mathcal{DB}_T}\}$ is a set of constraint sets. If $\Sigma$ is $\mathcal{L}_{P}$-constrained-adaptively-indistinguishable secure, and $\mathbf{C}^{[1,\dots,T]}$ is $(\mathcal{L}_{P},\alpha,\mathcal{F})$-acceptable, then $\Sigma$ is $(\mathcal{L}_{P},\mathbf{C}^{[1,\dots,T]},\alpha,\mathcal{F})$-constrained-adaptively-indistinguishability secure.
\end{theorem}

\vspace{-20pt}
\section{System Implementation} 
\label{sec:implementation}

A simple way to implement the padding service \textit{P} of \system~ is that \textit{Cache Controller} and \textit{Padding Controller} are maintained synchronously in a   single process.
That is, \textit{Cache Controller} is idle while \textit{Padding Controller} performs padding and encryption, and vice versa. Then, encrypted batches are uploaded to the server \textit{S}.
We observe that this single process becomes extremely slow in the long run because \textit{Cache Controller} and \textit{Padding Controller} cannot make use of CPU cores in parallel.
As a result, there are a very few batches uploaded to \textit{S}. 

To address the above bottleneck, we propose \textit{Orchestrator}, a component bridging data flow between \textit{Cache Controller} and \textit{Padding Controller}. 
\textit{Orchestrator} enables ShieldDB to maximise the usage of CPU cores by splitting two controllers to process in parallel. 
Figure~\ref{fig:implementation} depicts the implementation of the system. 
In details, \textit{Cache Controller} and \textit{Padding Controller} are spawned as separate system processes during \textsf{setup}. 
\textit{Orchestrator} acts as an independent proxy manager managing the cache clusters in \textit{P}'s shared memory. 
It provides process-safe access methods of collecting, clearing, and updating data in a given cluster. 

The communication between \textit{Cache Controller} and \textit{Padding Controller} is performed by sockets during the \textsf{streaming} operations. 
\textit{Cache Controller} acts as a client socket, and notifies \textit{Padding Controller} in the order of clusters that are ready for padding as per padding strategy. 
Then, \textit{Cache Controller} awaits a checksum notified by \textit{Padding Controller}. 
The checksum reports the number of keyword and document id pairs in the cached cluster. 
%
Note that \textit{Padding Controller} only collects necessarily cached data for padding upon the \textit{high} or \textit{low} padding mode. 
%

Apart from these components, ShieldDB contains \textit{Padding Daemon}, \textit{Streaming Daemon}, and \textit{SE Client}. 
They are activated by \textit{App Controller} during \textsf{setup}. 
\textit{Padding Daemon} provides \textit{Padding Controller} with the access to a bogus dataset, and maintains the track of remaining bogus entries for each keyword. It will generate a new bogus dataset if it is run out. 
\textit{Streaming Daemon} allows \textit{App Controller} to setup HTTP request/response methods to \textit{S}'s address.
\textit{SE Client} deploys our encryption protocols at \textit{C}, as presented in Figure~\ref{fig:shieldDB_protocols}.  This service is separated so that later protocol updates are compatible to other components in the system.

\begin{figure}
    \centering
    \begin{subfigure}{0.28\textwidth}
        \includegraphics[width=0.8\textwidth, height=2.cm]{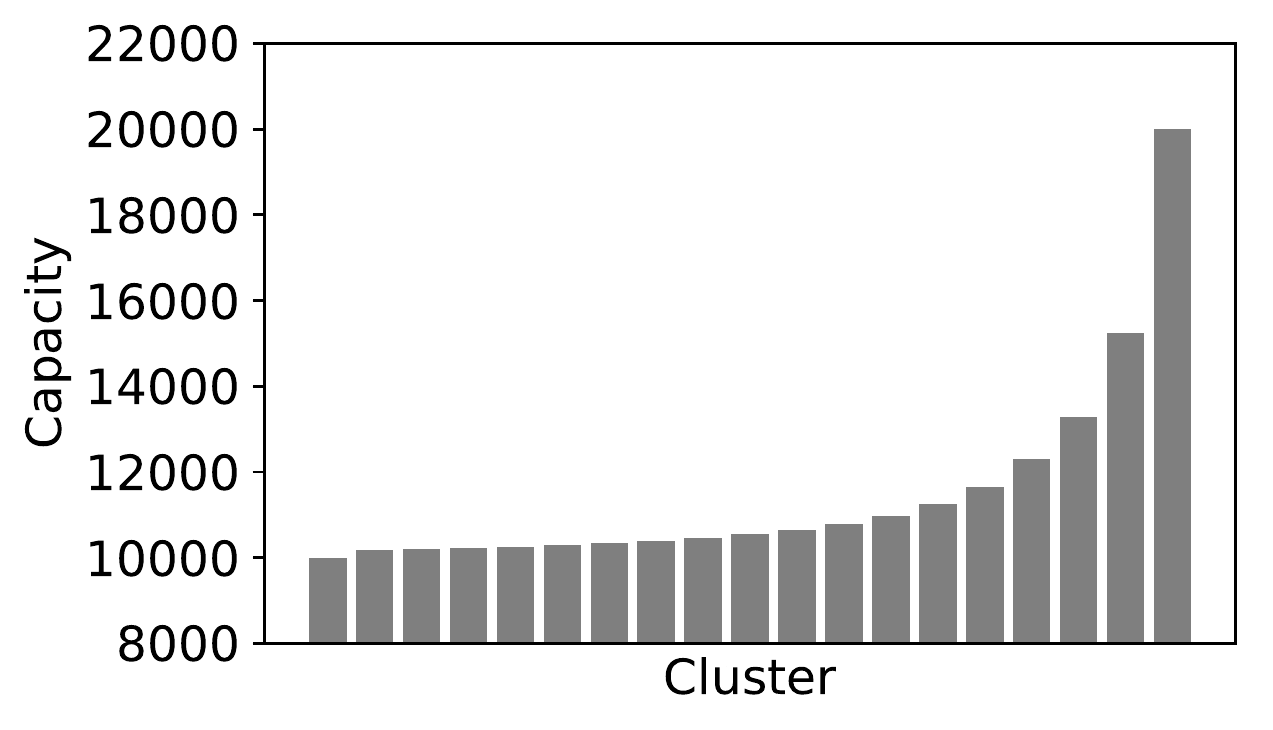}
        \vspace{-5 pt}
        \caption{$\alpha=256$}
        \label{fig:capacity256}
    \end{subfigure}~
    \begin{subfigure}{0.28\textwidth}
        \includegraphics[width=0.8\textwidth,height=2.cm]{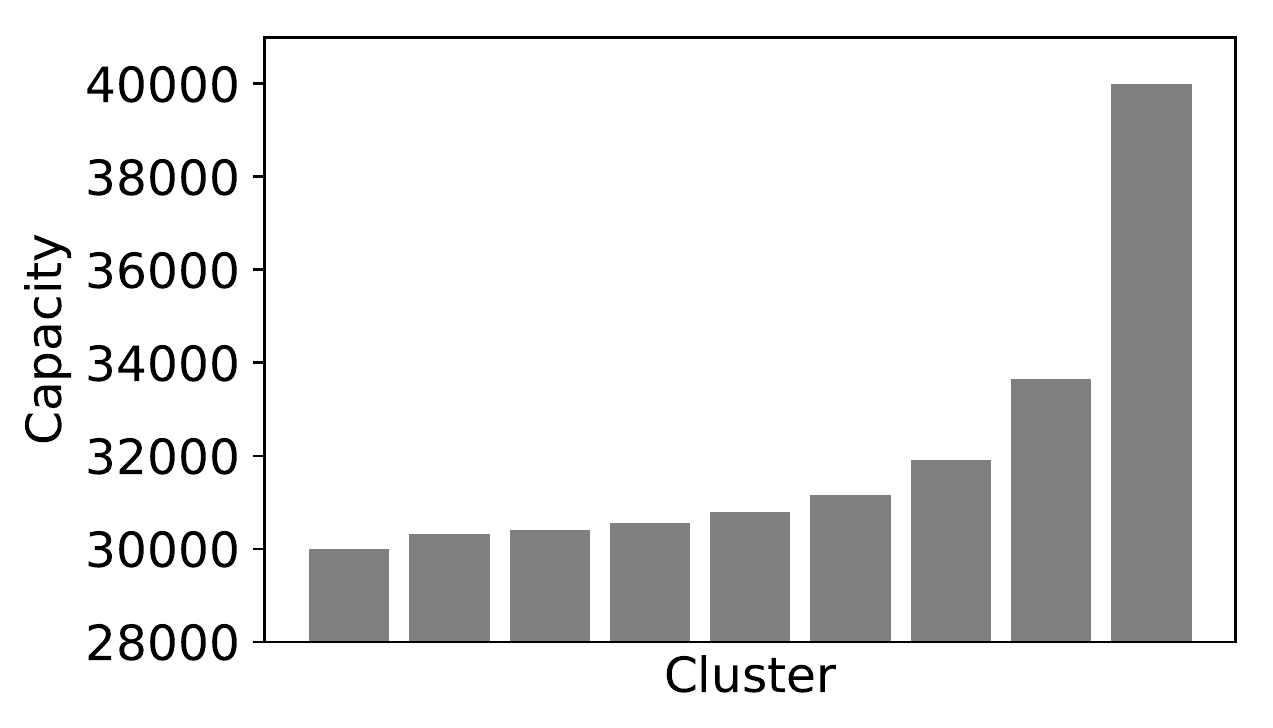}
        \vspace{-5 pt}
        \caption{$\alpha=512$}
        \label{fig:capacity512}
    \end{subfigure}
    \caption{Cache capacities}
    \vspace
    {-5pt}
    \label{fig:capacity}
\end{figure}

\begin{table}
\caption{Batch processing results}
\vspace{-15pt}
\begin{center}
\small
\begin{tabular}{|c|c|c|c|c|}
\hline
\multirow{2}{*}{Setting} & \multicolumn{2}{c|}{\textbf{Batch Insertions}} & \multicolumn{2}{c|}{\textbf{Avg. time/batch} (ms)} \\ \cline{2-5} 
& \multicolumn{1}{c|}{$\alpha=256$} & \multicolumn{1}{c|}{$\alpha=512$} & \multicolumn{1}{c|}{$\alpha=256$} & \multicolumn{1}{c|}{$\alpha=512$} \\ \hline
\textbf{NH}& $30$ & $5$ & $7047.2$ & $51384.41$ \\ \hline
\textbf{NL}& $1919$ & $497$ & $113.94$ & $456.87$ \\ \hline
\textbf{PH}& $45$ & $6$ & $5280.58$ & $45734.16$\\ \hline
\textbf{PL}& $1916$ & $549$ & $138.34$ & $465.09$\\ \hline
\end{tabular}%
\end{center}
\vspace{-15pt}
\label{table:batch}
\end{table}

\begin{figure*}[!t]
\centering
\begin{minipage}[t]{0.33\linewidth}
\centering
\includegraphics[width=0.9\textwidth,height=3.5cm]{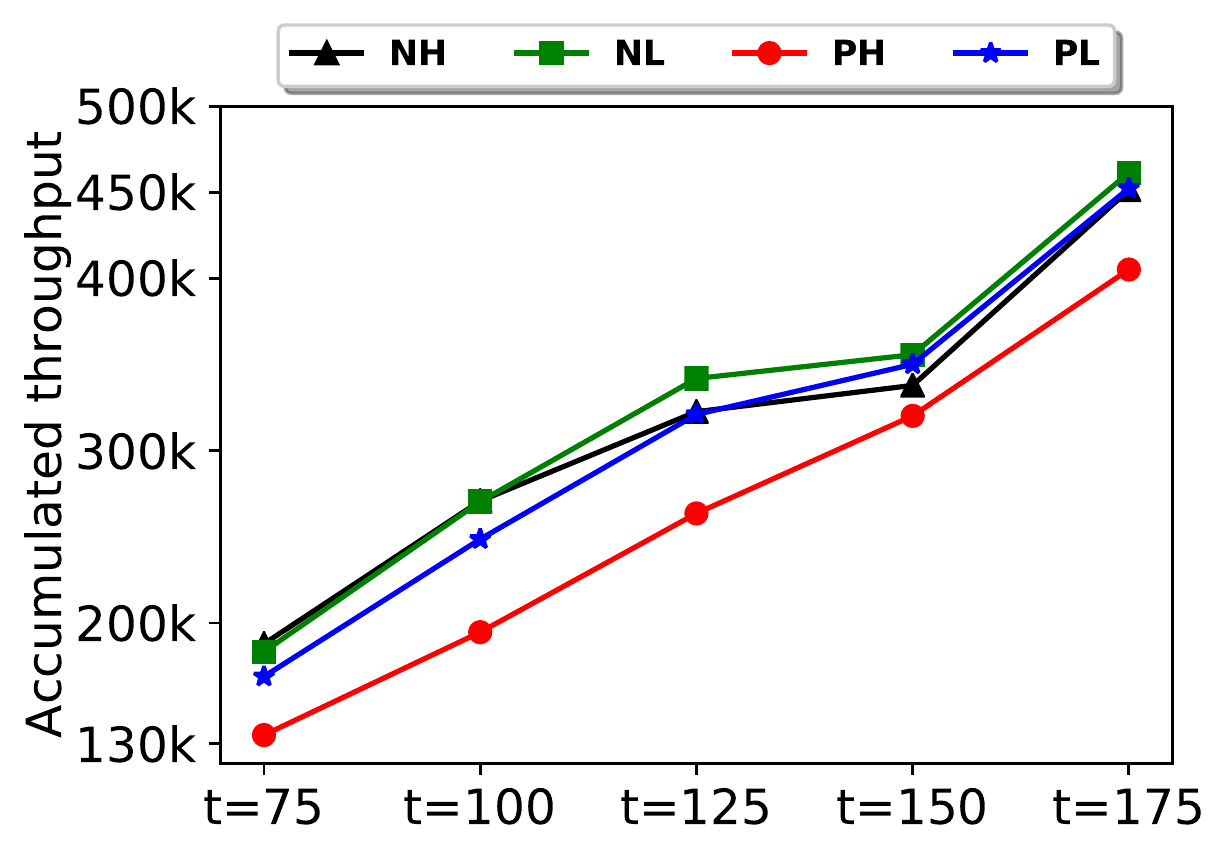}

\end{minipage}%
\begin{minipage}[t]{0.33\linewidth}
\centering
\includegraphics[width=0.9\textwidth,height=3.5cm]{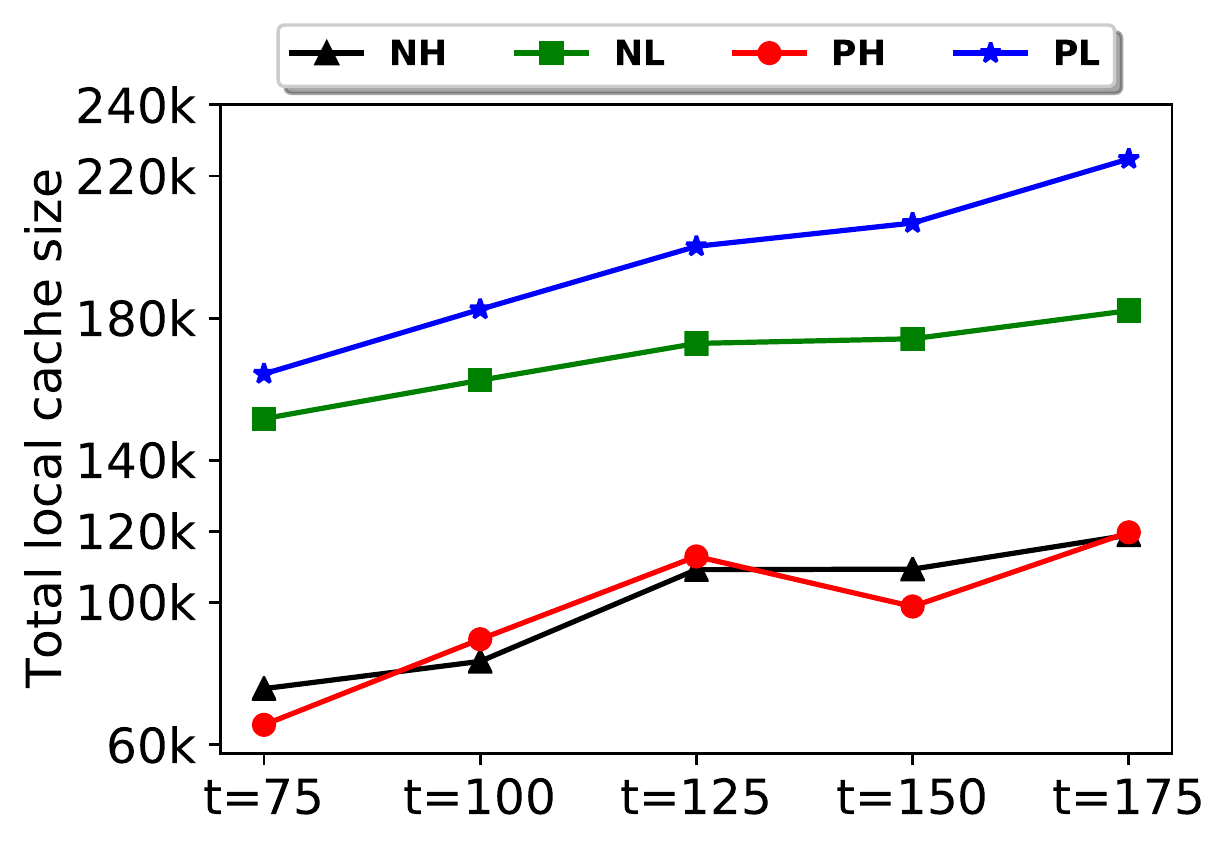}

\end{minipage}
\begin{minipage}[t]{0.33\linewidth}
\centering
\includegraphics[width=0.9\textwidth,height=3.5cm]{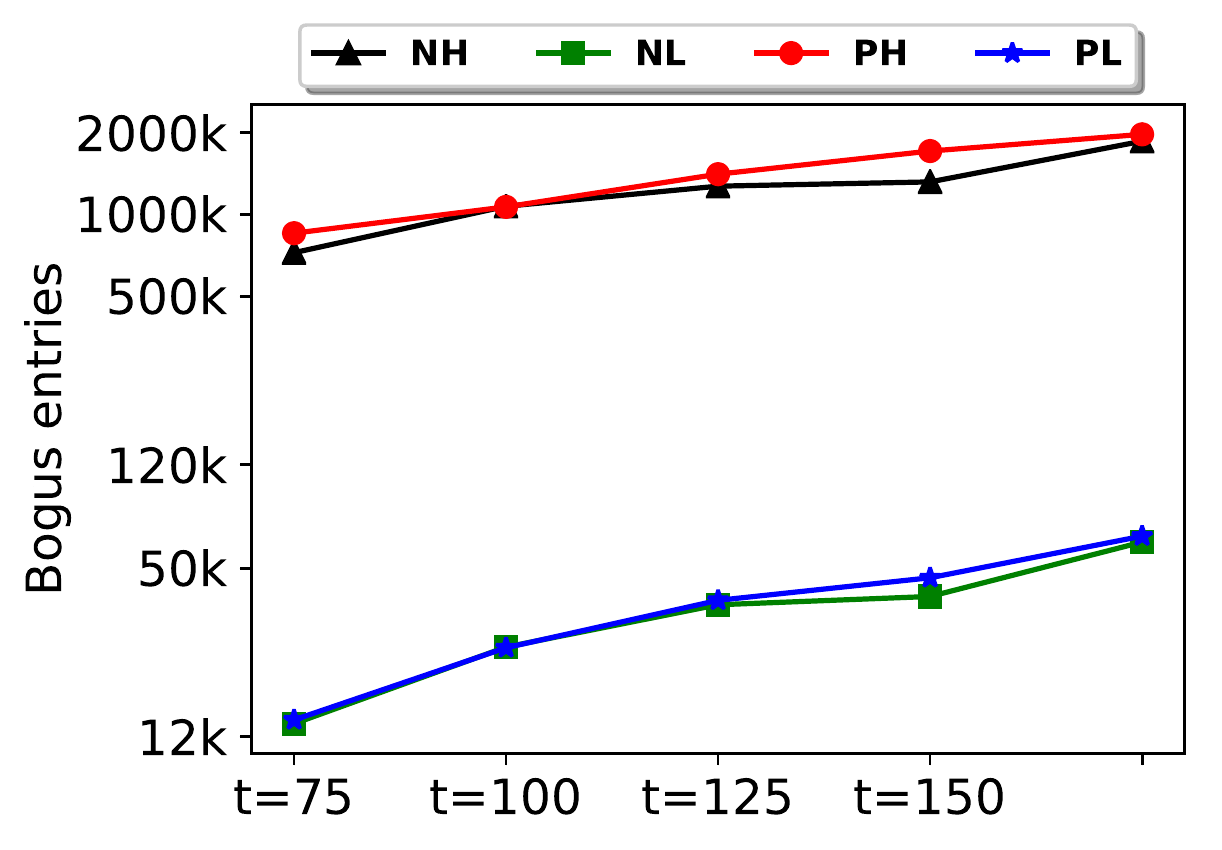}

\end{minipage}

(a) $\alpha=256$

\centering
\begin{minipage}[t]{0.33\linewidth}
\centering
\includegraphics[width=0.9\textwidth,height=3.5cm]{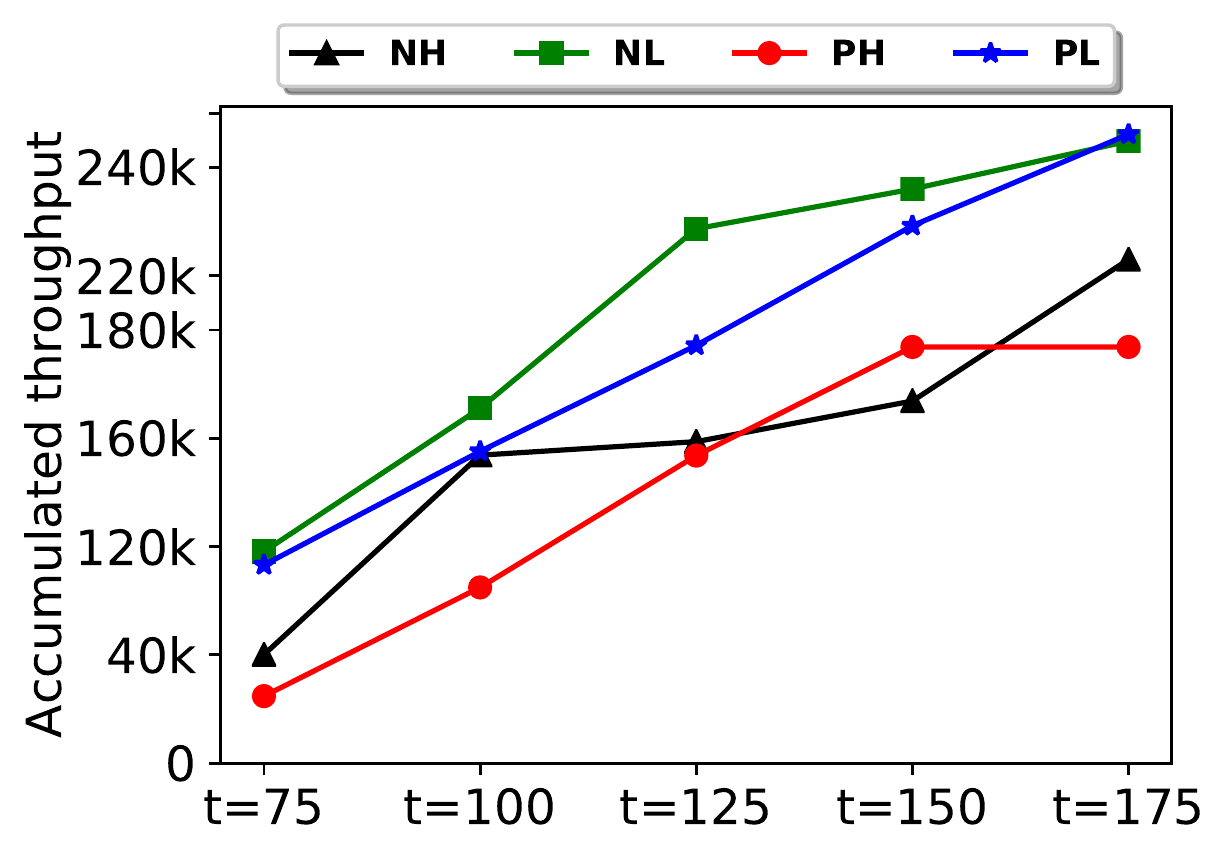}

\end{minipage}%
\begin{minipage}[t]{0.33\linewidth}
\centering
\includegraphics[width=0.9\textwidth,height=3.5cm]{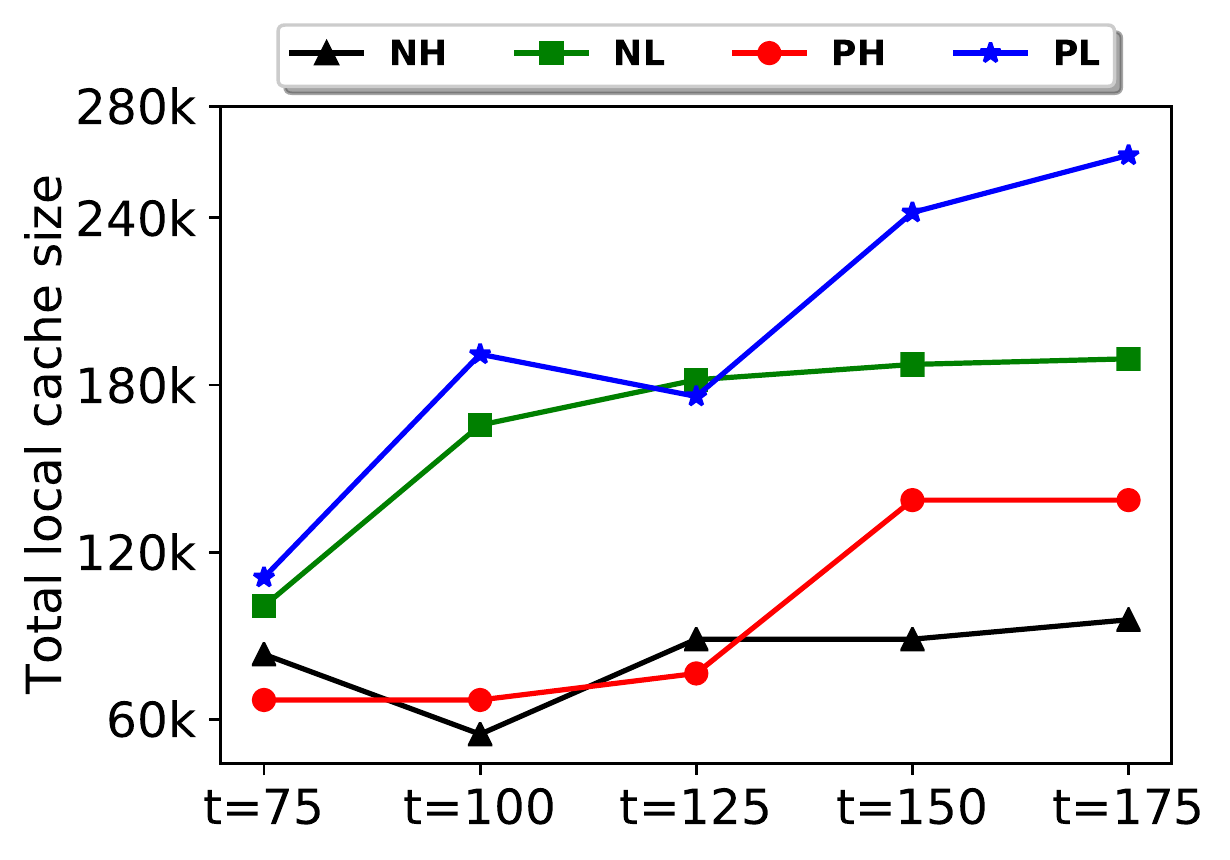}

\end{minipage}
\begin{minipage}[t]{0.33\linewidth}
\centering
\includegraphics[width=0.9\textwidth,height=3.5cm]{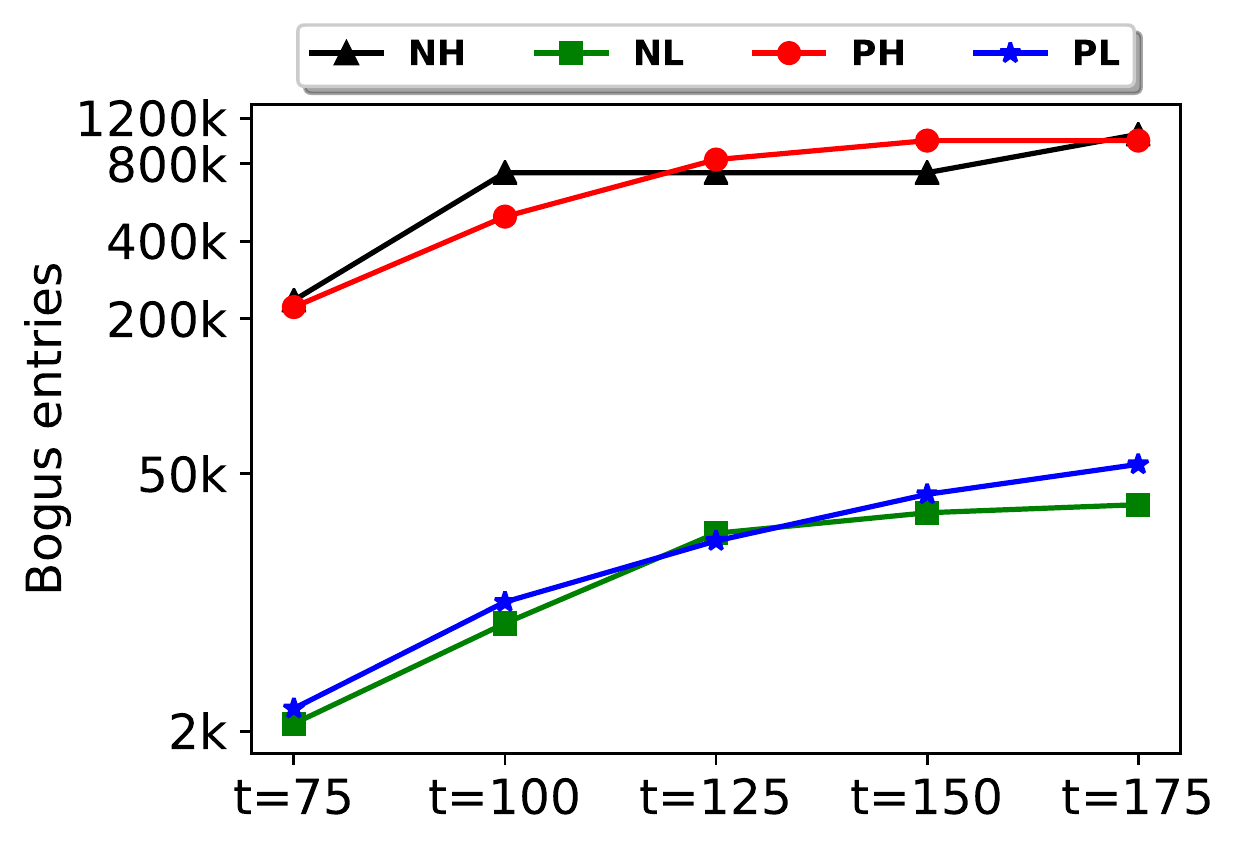}

\end{minipage}

(b) $\alpha=512$

\vspace{-10pt}
\caption{Evaluations on the accumulated throughput, local cache size, and number of used bogus entries, respectively}
\label{fig:client}
\vspace{-10pt}
\end{figure*} 

At \textit{S}, \textit{API Provider} provides RESTful APIs to serve \textit{C}'s HTTP requests. 
%
%
By calling streaming API requests, \textit{API Provider} then passes collected batches in \textsf{streaming} to the \textit{EDB Controller}. 
This component executes the insertion protocol as presented in Figure~\ref{fig:shieldDB_protocols}. 
\system~introduces a component called \textit{EDB Wrapper}, which separates \textit{EDB Controller}'s protocols from any database storage technology.


\vspace{-12pt}
\section{Experimental Evaluation}
\label{sec:exp}
We evaluated ShieldDB to understand the applicability of the padding strategies by investigating (1) the streaming throughput, padding overhead, and local cache load of the \textit{Padding Service} when using different padding strategies against the non-persistent and persistent adversaries, (2) the corresponding EDB size and search latency relating to the padding overhead, and (3) the efficiency of \textsf{flushing} in reducing  the cache load and the padding overhead reduction when \textsf{re-encryption} is applied.

%


\vspace{-12pt}
\subsection{Setup and Overview}
\label{experiment:setup}
ShieldDB is developed using Python and the code is published online\footnote{\url{https://github.com/MonashCybersecurityLab/ShieldDB}}. We use standard packages of Pycrypto (2.6.1) to implement cryptographic primitives (SHA256 for cryptographic hash functions and AES-128 cipher for pseudo-random functions) and NLTK (3.3) for textual processing. We deploy ShieldDB in Azure Cloud and run on an isolated DS15 v2 instance (Intel Xeon E5-2673 2.4GHz CPU with 20 cores and 140G RAM), where Ubuntu Server 17.1 is installed.  The controllers of the padding service are implemented by using Python multiprocessing package. For simplicity, we co-locate the \textit{Client} and \textit{Padding Service} at the same instance. At the server side. \textit{API Controller} works on top of the Flask-a micro web framework, while EDB is realised by RocksDB, a key-value storage.

We select the Enron data set\footnote{Enron email dataset: \url{https://www.cs.cmu.edu/~./enron/}}, and extract 2,418,270 keyword/document id pairs from the top 5,000 most frequent keywords in the dataset as the keyword space in our experiment.
%
%
We group them and allocate the cache capacity for each keyword cluster based on their frequencies, as introduced in Section~\ref{subsec:setup}. Figure~\ref{fig:capacity} presents the normalised cache capacities of these clusters at different values of $\alpha$. 
Recalled that $\alpha$ indicates the minimum number of keywords in each cluster (see Section~\ref{subsec:setup}). 
%
%
During the \textsf{setup}, ShieldDB generates a padding dataset for the keyword set. In our experiments, the dataset is estimated empirically enough to be used in streaming data up to 175 seconds for both $\alpha= 256$ and $\alpha= 512$. In details, the dataset contains 1,859,877 bogus pairs ($\approx$ 389 Kb).

To create the streaming scenario, the \textit{Client} groups every $10$ documents in the Enron data set as a batch (approx. $560$ stemmed keyword/id pairs) and continuously inputs batches to the system.
Note that we do not limit the processing capability of the \textit{Padding Service} \textit{P} and we let it continuously handle batches in a queue sent by the \textit{Client}. In every batch, \textit{P} performs the padding and encryption, and then continuously streams encrypted batches to the \textit{Server}.
To faithfully understand the performance of padding, we deploy the client and server to the same dedicated instance so that the impact of network I/O is minimised. Note that we begin to record the performance of ShieldDB after the cold start period of $75$ seconds. 

We experiment ShieldDB with different combinatorial settings of padding strategies and modes. They are denoted as \emph{NH} (strategy against non-persistent adversary via \emph{high} mode), \emph{NL} (non-persistent padding strategy via \emph{low} mode), \emph{PH} (strategy against persistent adversary via \emph{high} mode, and \emph{PL} (strategy against persistent adversary via \emph{low} mode).
The performance of ShieldDB is evaluated via  a set of measurements such as  system throughput, local cache size, used bogus pairs, EDB size, and search time. Here, the system throughput represents the total accumulated number of real $(w,id)$ pairs that have been encrypted and inserted to EDB. 
%
%

\vspace{4pt}
\noindent\textbf{Remark}:
Our focus is to analyse the system performance with different settings of padding strategies and modes related to the security as mentioned in Section~\ref{subsec:padding}. The empirical settings of 175 second streaming period and 10 documents per batch are used for the evaluation under a stable workload, not causing performance bottleneck, when the \textit{Client} and \textit{Padding Service} are co-located at the same Azure instance. Other parameters of a streaming period expects to result in the same observation as we obtained. The batch size can be adjusted empirically based on the application and client's resources.

\vspace{-10pt}
\subsection{Evaluation}
\label{subsec:evaluation}
We measure the performance of ShieldDB at both \textit{Padding Service} and the untrusted server \textit{S} over a 175-second streaming period.
In details, Fig.~\ref{fig:client} summarises the performance of \textit{Padding Service} with the three different metrics of accumulated throughput, local cache size, and padding overhead when setting $\alpha=256$ and $\alpha=512$.
Then, Fig.~\ref{fig:server} describes the performance of \textit{S} by observing EDB size, search time, and the average result length of query keywords.

\noindent \textbf{System throughput}: We first measure the accumulated throughput over time when ShieldDB is deployed with different padding modes of \emph{NH} (\underline{n}on-persistent using \underline{h}igh padding mode), \emph{NL} (\underline{n}on-persistent using \underline{l}ow padding mode), \emph{PH} (\underline{p}ersistent using \underline{h}igh padding mode), and \emph{PL} (\underline{p}ersistent using \underline{l}ow padding mode). 
We also monitor the number of batch insertions and the average batch processing time of \textit{Padding Controller} to evaluate the throughput difference between these padding strategies.  

Fig.~\ref{fig:client} shows that these padding modes have similar throughput at a lower $\alpha=256$. 
However, the overall throughput reduces nearly a half when setting $\alpha=512$. It is explained that padding overhead and encryption cost are higher when more keywords are allocated in each cluster. Consequently, the throughput will be decreased. 
Table~\ref{table:batch} also supports that finding when fewer batches are inserted to \textit{S} and the average processing time per batch takes a longer time when setting $\alpha=512$.
%

Furthermore, when setting $\alpha=512$, Fig.~\ref{fig:client}  shows that \emph{low} mode promotes more real keyword/id pairs to be inserted to EDB than \emph{high} mode. In details, the throughput of \emph{NL} is $1.23$ times higher than the throughput of \emph{NH}, and \emph{PL}'s is about $1.51$ timer higher than \emph{PH}'s. 
Table~\ref{table:batch} also supports this finding when it reports that \emph{low} mode creates more batch insertions than \emph{high} mode, while its average batch processing time is completely negligible compared to that value of the latter. 
This observation shows the efficiency of \emph{low} mode since it only performs necessarily minimum padding for keywords in every batch. 
In contrast, \textit{Padding Controller} takes longer time under \emph{high} padding mode due to higher padding overhead and the longer encryption time taken by the large number of bogus pairs. 
%

\noindent \textbf{Cache size}: To investigate the overhead at the padding service, we monitor the local cache as shown in Fig.~\ref{fig:client}.
In general, \emph{low} mode results in a larger number of cached pairs in cache clusters than \emph{high} mode, regardless of padding constraints. 
The cache in \emph{NL} consumes $150\%{\sim}197\%$ larger space than the cache in \emph{NH}.
The load of cache in \emph{PL} is $1.8{\sim}2.5 \times$ higher than the load of the cache in \emph{PH}.
%

\noindent \textbf{Padding Overhead}: 
We rely on the number of used bogus entries reported in Fig.~\ref{fig:client} to compute the padding overhead of different combinatorial settings of padding strategies and modes. 
The padding overhead is estimated as the ratio between the bogus and real (throughput) pairs.
We see that although \emph{high} padding mode achieves a lower load of cache than \emph{low} mode, it utilises more bogus pairs from the generated padding dataset than the latter.
In details, the padding overhead of \emph{NH} ranges from $3.8{\sim} 4.1 $ and from $5.6 {\sim} 5.8 $ for $\alpha = 256$ and  $512$, respectively.
In contrast, the padding overhead of \emph{NL} ranges is marginal, varying from $0.07 {\sim} 0.13 $ and $0.06 {\sim} 0.16$ for $\alpha = 256$ and $512$, respectively. The reason is that a portion of streamed keyword/id pairs are still cached at the padding service. 
It also demonstrates that when $\alpha$ is large, \emph{PH} generates a larger padding overhead than \emph{NH} does. Specifically, the padding overhead of \emph{PH} is in the range of $6.4{\sim}8.9$ for $\alpha = 512$. 
The reason is that \emph{PH} will add bogus pairs for keywords that do not appear in the current time interval, while \emph{NH} will not if the keywords have not existed. 
%
%

\begin{figure}
\centering
\includegraphics[width=0.40\textwidth,height=3.5cm]{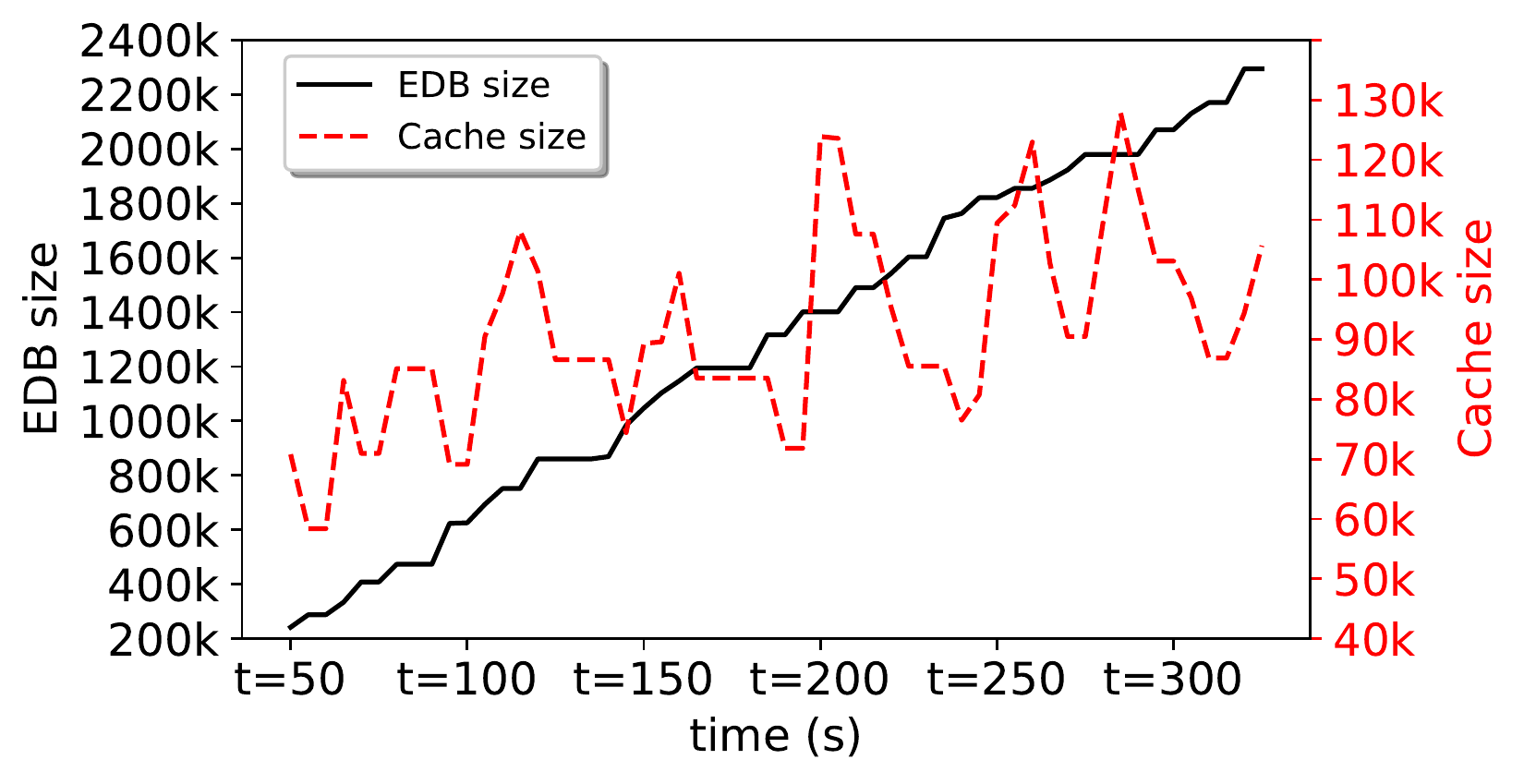} 
\vspace{-5pt}
\caption{Flushing operation with $\alpha=256$}
\vspace{-23pt}
\label{fig:flush}
\end{figure}

\noindent \textbf{EDB size}: We report the number of real and bogus pairs in EDB over the time in Fig.~\ref{fig:server}. It demonstrates that \emph{high} mode generates more data in EDB than \emph{low} mode due to the selection of all cached pairs in clusters for padding and the  large number of used bogus pairs.

\noindent \textbf{Search time}:
To demonstrate the search performance, we configure the client  to query 10\% randomly selected keywords in EDB at timestamps, i.e, $t=75, 100 , 125, 150$, and $175$. Fig.~\ref{fig:server} shows that \emph{high} mode makes querying a keyword take a longer time, because \textit{S} decrypts more bogus pairs. In contrast, the search time in \emph{low} mode is shorter due to the fewer used bogus pairs. The search time in  \emph{NH} and \emph{PH} is fluctuated due to the change of the result lengths of keywords in EDB as given in Table~\ref{table:resultlength256} and Table~\ref{table:resultlength512}. 


\begin{figure*}[t]
	\hspace{-0.3cm}
	\begin{minipage}[t]{0.30\linewidth}
	    \centering
	    \includegraphics[width=\textwidth,height=3.5cm]{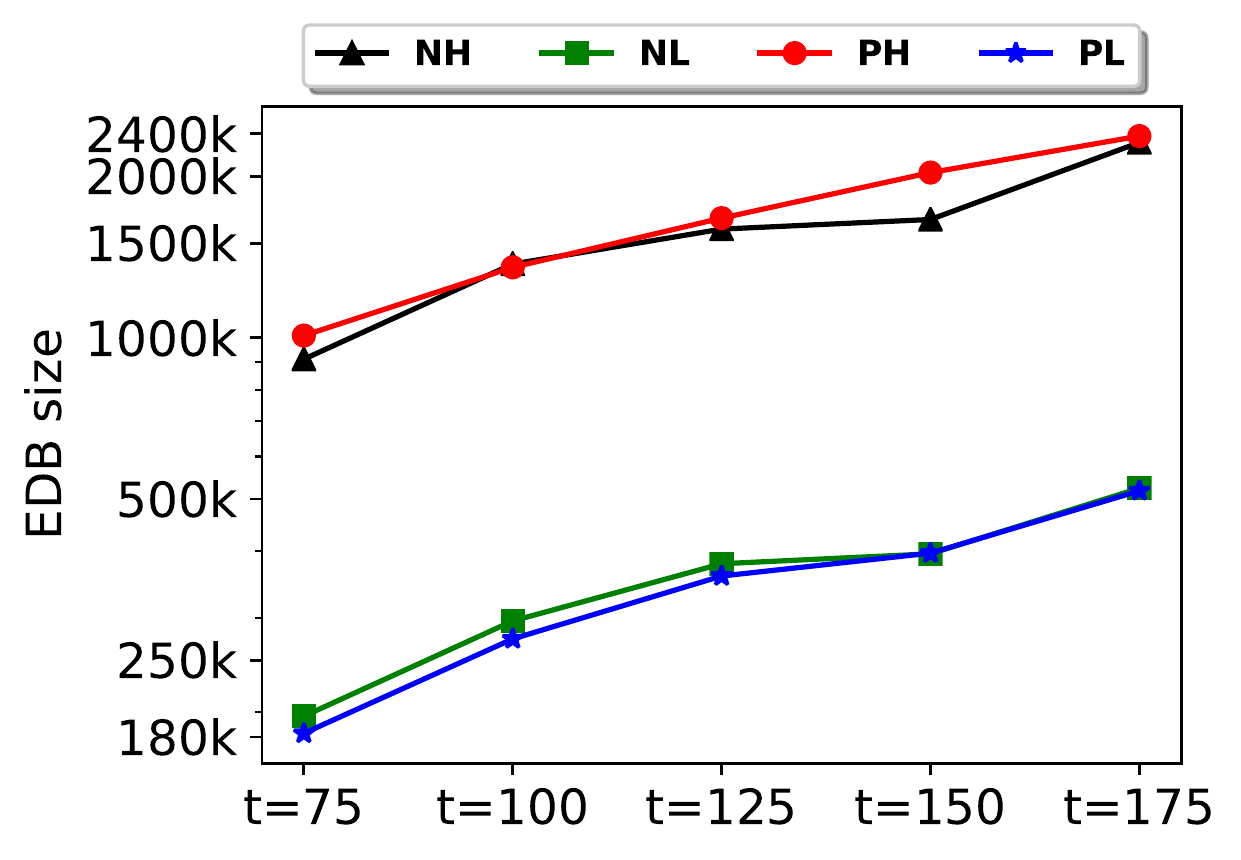} 
  	\end{minipage} 
  	\hspace{-0.3cm}
  \begin{minipage}[t]{0.30\linewidth}
	    \centering
	    \includegraphics[width=\textwidth,height=3.5cm]{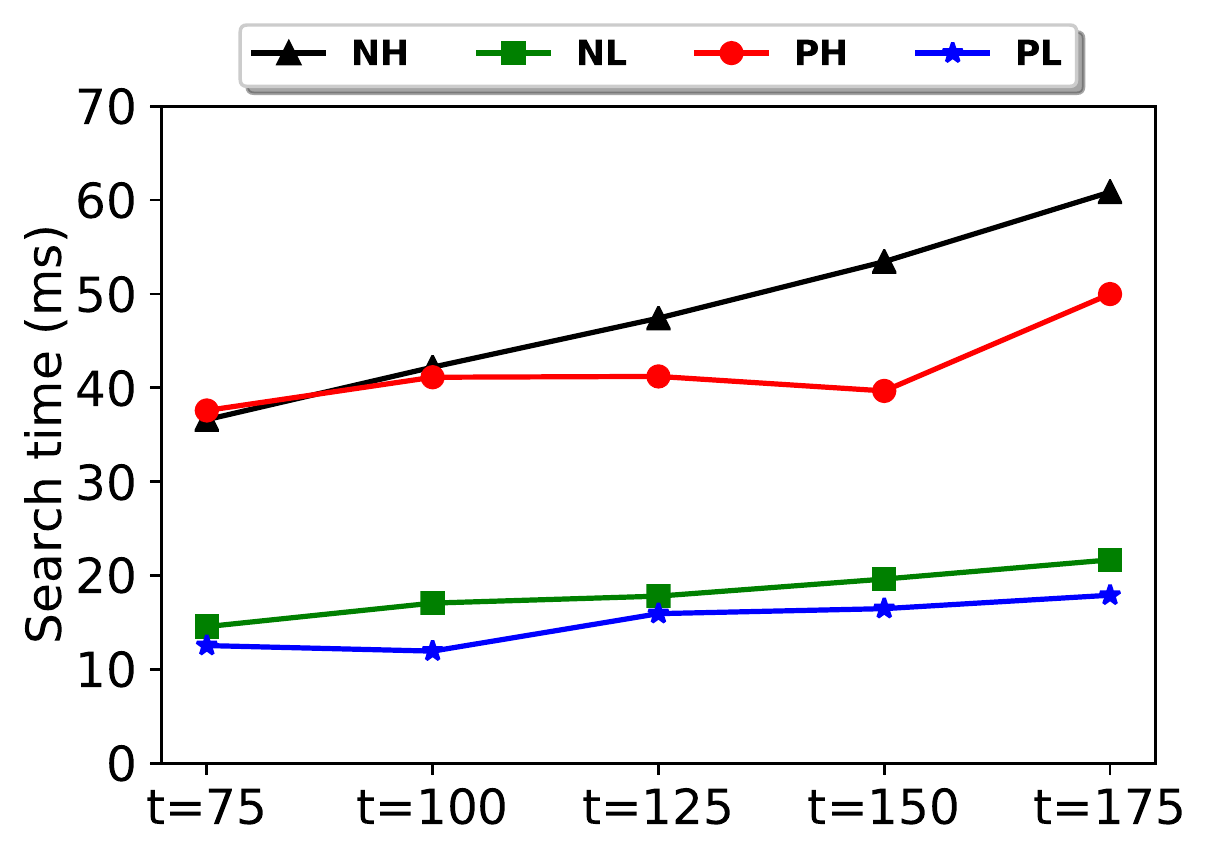} 
  \end{minipage}
  \hspace{-0.4cm}
  \begin{minipage}[t]{0.44\linewidth}
	    \centering
	    \vspace{-100pt}
	    \captionof{table}{Result length with $\alpha=256$}
	    \scalebox{0.8}{
			\begin{tabular}{|c|c|c|c|c|c|}
				\hline
				\multirow{2}{*}{Setting} &  \multicolumn{5}{c|}{\textbf{Time intervals}}\\ 
				\cline{2-6}
				&$t=75$ & $t=100$ & 	$t=125$ & $t=150$ & $t=175$ \\ \hline
				\textbf{NH}      &  	$593.78$ &  $669.94$ 		& 	$778.45$ 		& $811.25$ &$903.53$ \\ \hline
				\textbf{NL}      &   	$109.856$ &  $144.66$ 		& 	$164.40$ 		& $171.56$ & $186.04$ \\ \hline
				\textbf{PH}      &   	$562.86$ &  $660.22$ 		& 	$593.18$ 		& $579.30$ & $714.12$ \\ \hline
				\textbf{PL}      &    	$89.25$ &  $82.38$ 		& 	$107.92$ 		& $110.57$ & $126.43$ \\ \hline
			\end{tabular}
		}
	\label{table:resultlength256}
  \end{minipage}

\hspace{8cm} (a) $\alpha=256$

\hspace{-0.3cm}
	\begin{minipage}[t]{0.30\linewidth}
	    \centering
	    \includegraphics[width=\textwidth,height=3.5cm]{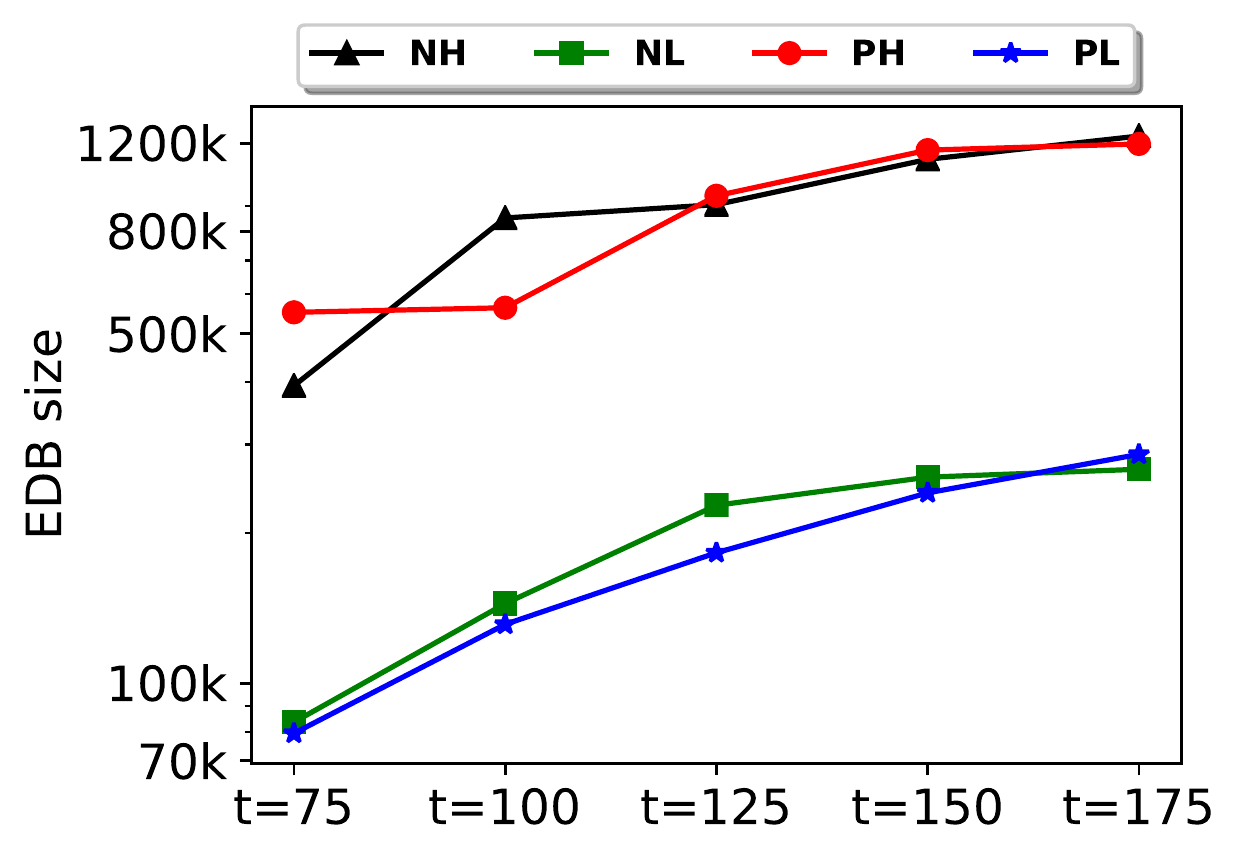} 
  	\end{minipage} 
  	\hspace{-0.3cm}
  	\begin{minipage}[t]{0.30\linewidth}
	    \centering
	    \includegraphics[width=\textwidth,height=3.5cm]{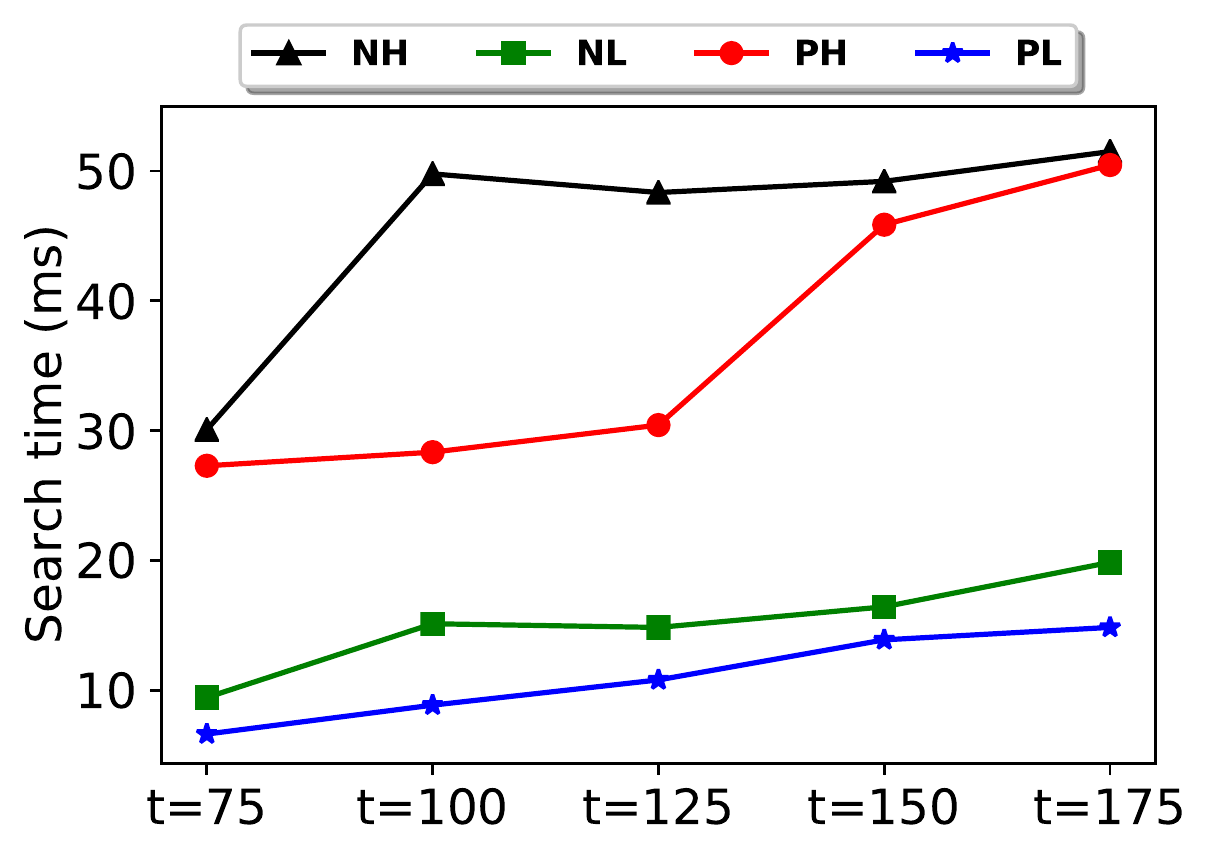} 
  	\end{minipage}
  	\hspace{-0.18cm}
  	\begin{minipage}[t]{0.44\linewidth}
  		\vspace{-105pt}
		\captionof{table}{Result length with $\alpha = 512$}
		\scalebox{0.8}{
			\begin{tabular}{|c|c|c|c|c|c|}
			\hline
			\multirow{2}{*}{Setting} &  \multicolumn{5}{c|}{\textbf{Time intervals}}\\ 
			\cline{2-6}
			&$t=75$ & $t=100$ & 	$t=125$ & $t=150$ & $t=175$ \\ \hline
			\textbf{NH}      &  	$668.02$ &  $843.22$ 		& 	$837.01$ 		& $840.85$ & $861.44$ \\ \hline
			\textbf{NL}      &   	$81.64$ &  $140.95$ 		& 	$147.73$ 		& $168.17$ & $174.40$ \\ \hline
			\textbf{PH}      &   	$577.06$ &  $610.02$ 		& 	$614.34$ 		& $857.73$ & $879.63$ \\ \hline
			\textbf{PL}      &    	$51.66$ &  $85.33$		& 	$89.21$ 		& $93.82$ & $112.32$\\ \hline
			\end{tabular}
		}
	\label{table:resultlength512}
	\end{minipage}

	\hspace{8cm} (b) $\alpha = 512$
    \vspace{-5pt}
	\caption{Evaluation on the EDB size and search time}
	\vspace{-15pt}
	\label{fig:server}
\end{figure*}

\noindent \textbf{Flushing}: We select two largest cache clusters to simulate the \textsf{flushing} operation. In particular, we set a small time window, 20 seconds, to trigger flushing. If these clusters do not exceed up to 75\% of their original capacities, then the \textsf{flushing} operation is invoked. 
Figure~\ref{fig:flush} reports EDB size and cache size over the time with a scanning window of 20 seconds. The operation occurs at $t=73$, $45$, $80$, $121$, $144$, $189$, $222$, $272$, and $331$ seconds. We observe that the cache size drops significantly at these timestamps since \textit{Cache Controller} flushes the cached pairs to \textit{Padding Controller}. Note that the EDB and cache sizes are flat while \textit{Padding Controller} performs padding and encryption.
Empirically, we observe that \textsf{flushing} operation averagely reduced the cache load efficiently up to $1.9\sim 2.8\times$ across the padding strategies in the same streaming period.

\noindent \textbf{Re-encryption}: To investigate the performance of \textsf{re-encryption}, we experiment ShieldDB after 175 seconds operated with \emph{NH} at $\alpha=256$. We select the keyword cluster that has the most entries stored in EDB for the re-encryption. This keyword set is also re-used as the query set to benchmark the query performance before, during, and after re-encryption. There are 180,677 real entries associating with 256 keywords of this cluster. Table~\ref{table:reencryption256} demonstrates the performance of the re-encryption. This operation takes 131.3 seconds for fetching process, and 103.11 seconds for padding and re-insertion. During the operation, the average query time per keyword is the smallest due to the deletion of all entries in the selected cluster. Note that this query time takes into account the search over local cache clusters if the keyword is not available in EDB. After re-encryption, the number of bogus entries used for the cluster is nearly reduced by 64.1\%, making the average search time shorter.

\noindent \textbf{Overall performance}:
%
Table~\ref{table:overalperformance} summarises the performance of \textit{Padding Service} regarding three critical measurements of throughput per second, average cache size at every second, and the overall padding overhead. 
%
As seen, there are no perfect padding strategies that can achieve a great balance. 
 \textit{Low} padding mode makes a higher throughput value and lightens padding overhead, but it incurs a significant cache load. 
In contrast, \textit{high} padding mode makes the cache load lightweight, but it introduces a higher overhead.   

Note that the padding strategies against the persistent adversary are also applicable to the non-persistent adversary. 
The \textit{firstBatch} condition can theoretically make some clusters might be not achieved in a long time if some keywords never appear. 
However, this is not the case in our current experiments. Therefore, the throughput for \emph{PH} and \emph{NH}, and \emph{PL} and \emph{NL} is close, respectively.

The value $\alpha$ relates to the number of keywords in clusters. A higher value indicates that more keywords are co-located in the same cluster. Hence, they all will have the same result length after padded.
From the results, \system~shows the tradeoff when selecting a higher value of $\alpha$. That is, the throughput is declined nearly double while padding overhead increases almost twice (see Table~\ref{table:overalperformance}).

\begin{table*}[!t]
\caption{Overall performance of ~\system~throughout a 175-second streaming period}
\vspace{-15pt}
\begin{center}
\small
\begin{tabular}{|c|c|c|c|c|c|c|c|}
\hline
\multirow{2}{*}{Setting} & \multirow{2}{*}{\textbf{Adversary Target}} & \multicolumn{2}{c|}{\textbf{Throughput per second}} & \multicolumn{2}{c|}{\textbf{Avg. cache load}} & \multicolumn{2}{c|}{\textbf{Padding overhead}} \\ \cline{3-8} 
& & \multicolumn{1}{c|}{$\alpha=256$} & \multicolumn{1}{c|}{$\alpha=512$} & \multicolumn{1}{c|}{$\alpha=256$} & \multicolumn{1}{c|}{$\alpha=512$} & \multicolumn{1}{c|}{$\alpha=256$} & \multicolumn{1}{c|}{$\alpha=512$} \\ \hline
\textbf{NH}& Non-persistent& $2,634.27$ & $1,459.62$ & $99,347.8$ & $82,267.8$ & $3.8\sim4.12$&$ 5.6\sim5.8$ \\ \hline
\textbf{NL}& Non-persistent& $2,779.77$ &$ 1,515.74$ &$ 168,681.4$ &$ 164,960$ & $0.07\sim0.13$ &$0.06\sim0.16$\\ \hline
\textbf{PH}& Persistent& $2,702.05$& $1,289.64$ &$ 97,351.6$ &$97,557.6$& $4.8\sim6.3$ &$6.4\sim8.9$\\ \hline
\textbf{PL}& Persistent& $2,833.46$& $1,590.46$ &$195,702.2$ &$196,413.6$& $0.08\sim0.14$& $0.08\sim0.23$\\ \hline
\end{tabular}%
\end{center}
\vspace{-10pt}
\label{table:overalperformance}
\end{table*}

\begin{table*}[!t]
\caption{Overall performance of the insecure streaming system and the forward-private SSE streaming system}
\vspace{-15pt}
\begin{center}
\small
\begin{tabular}{|c|c|c|c|c|c|c|}
\hline
& \multirow{2}{*}{insecure system} & \multicolumn{5}{c|}{Forward-private system} \\ \cline{3-7} 
& & $t=75$ & $t=100$ & $t=125$& $t=150$ & $t=162$  \\ \hline
Throughput (pairs/s)&  $4.3 \times 10^4$ & $1.73 \times 10^4$ & $1.71 \times 10^4$ & $1.71 \times 10^4$ & $1.73 \times 10^4$  &$1.74 \times 10^4$ \\ \hline
{Avg. result length (\#pairs)} &  $483.66$  &  $281.5$ & $384.8$ & $476.53$ & $513.61$ & $483.61$\\ \hline
{Search latency (ms)} &    $5.12$  &$20.9$ & $24.62$ & $33.75$ & $40.95$ & $43.59$\\ \hline
Storage overhead (\#pairs)&  $2.41 \times 10^6$ & $1.3 \times 10^6$ & $1.7 \times 10^6$ & $2.14 \times 10^6$ & $2.17 \times 10^6$ & $2.41 \times 10^6$ \\\hline
\end{tabular}
\label{table:forward_system}
\end{center}
\vspace{-15pt}
\end{table*}

\begin{table}
\caption{Re-encryption on the largest cluster}
\vspace{-15pt}
\begin{center}
\small
\begin{tabular}{|c|c|c|c|}
\hline
{}& \textbf{Before} & \textbf{During}  & \textbf{After} \\ \hline
Bogus entries used    &  $643,131$ &  $230,715$ & $230,715$\\ \hline
Search time (ms)      & $379.37$  &  $0.03$ & $210.18$\\ \hline
\end{tabular}
\label{table:reencryption256}
\end{center}
\vspace{-10pt}
\end{table}

\noindent  \textbf{Comparison with baselines}:
We further investigate the security and performance trade-off between \system~and two baselines.
The first baseline (aka Baseline-I) is an insecure system for which batches of un-encrypted keyword/id pairs without padding are streamed to the server's storage.
We define the batch size as $256$ pairs. 
The second one (aka Baseline-II) is also the streaming system without padding, but it realises our searchable encryption scheme with forward privacy, presented in Figure~\ref{fig:shieldDB_protocols}, to encrypt the pairs of every batch insertion.

We measure the overall performance of these baselines by using the same streaming database and the measurement metrics as evaluated for ~\system~(in Section \ref{experiment:setup}).
It is clear that Baseline-II brings $2.5\times$ overhead in addition throughput compared to Baseline-I. 
The reason is because that the encrypted entries of keywords in the same batch indeed are generated from the ephemeral key of the batch and keyword's extracted state. 
We note that Baseline-I maintained a constant throughput and completed streaming within $55$ seconds, while Baseline-II finished in $162$ seconds (Table \ref{table:forward_system}).
The throughput overhead ~\system~ brought forward is about $6.04\sim6.5\times$ (resp. $10.7\sim11.74\times)$   lower than performance of Baseline-II when setting $\alpha=256$ (resp. $512$).
We note that overhead is caused by the encryption of additional bogus pairs introduced in every batch insertion.
The average result length for keywords streamed to the EDB of ~\system~ is about $1.8\sim2.3\times$ (\textit{high} padding mode used) greater than that value if Baseline-II is deployed.
The search latency of ~\system ~is almost double ($1.7\sim 2.1\times$), slower than Baseline-II.
%
%
We observe that the security enhanced by the padding and forward privacy would overall bring the streaming throughput per second $\sim 16.3\times$ (resp. $\sim 27\times$) slower than Baseline-I when setting $\alpha=256$ (resp. 512).

\noindent  \textbf{Empirical analysis of streaming distribution}: Next, we investigate how the streaming distribution of real data outsourced by \system~ to EDB changes over the time. 
To do this, we consider the training distribution used to generate the padding dataset and cluster's caches extracted from the training dataset in the \textsf{Setup}, (in Section~\ref{subsec:setup}), as the baseline distribution.
Note that the training distribution was different with respecting to $\alpha$ (i.e., the minimal number of keywords in every cache cluster) (see Equation 1).
%
%
Then, we monitor the streaming distribution of real data at different times $t$  when ~\system~ employs different combinatorial settings of padding strategies and modes.
In particular, we use the Kullback–Leibler (KL) distance~\cite{Zhang17}~%
%
to measure the difference between such streaming distributions and the baseline distribution (Figure~\ref{fig:streamingdist}).

Our observation shows that the streaming distribution was different compared to the baseline at the early streaming time (i.e., $t=75-150$). Then, it tended to converge to the baseline when the streaming dataset was almost outsourced completely ($t\geq 175$).
At this time, the padding dataset was also almost used.
The reason for that is because the completed streaming dataset shares the same distribution with the training dataset.
However, at the earlier time, the difference was large because there was some keywords in the training distribution that did not appear yet in these early streaming batches.
%
%
We note that, with $\alpha=256$, the persistent padding settings (i.e., \textit{PH} and \textit{PL}) have the largest  distribution difference since it requires the existence of all keywords in the cluster at the \textit{firstBatch} (in Algorithm~\ref{alg:paddingstrategy}).
With $\alpha=512$ (i.e., more keywords required in clusters), the distribution difference was double (i.e., $0.8\sim1.0$ KL unit) since the setting causes a longer waiting period before the padding strategies meet, under the same streaming rate.
%


\subsection{Discussion on Deployment}
\label{sec:discussion}

%
%
%
%
%
%
%
%


 We note that the above experiments consider the keyword frequency distribution in \textsf{setup} is similar to the one in a period of \textsf{streaming} operation.
We deem that the assumption and the corresponding setting of \system~ for deployment can be practically held in practice.
%
%
%
%

\begin{figure}
\centering
\includegraphics[width=0.5\textwidth,height=3.5cm]{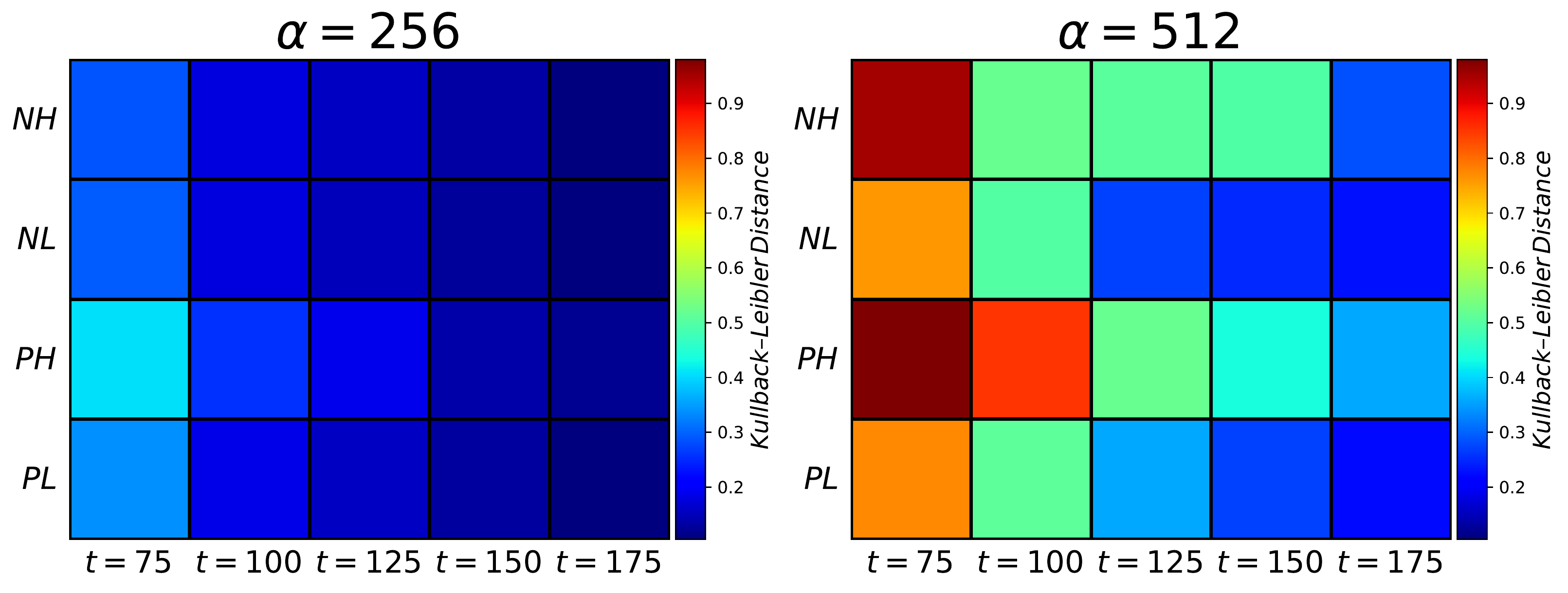} 
\vspace{-15pt}
\caption{The difference in streaming distribution}
\vspace{-15pt}
\label{fig:streamingdist}
\end{figure}

First, this setting is applicable to streaming applications for which the underlying distribution does not change much over time or it is known in advance, like the known data range of the IoT sensors~\cite{Zhang2017StreamingKC,Gomes19}.
Second, in our observation (Figure~\ref{fig:streamingdist}), we show that the streaming distribution changes towards the training distribution over a streaming period, not requiring the exact matching between them for any particular interval. Therefore, the assumption can at least be hold for that duration, and better than assuming that the streaming and training distribution are close for any particular time interval.  
%
%
%
%
%
%
%
In addition, we note that the \textsf{setup} operation can be re-invoked again to re-cluster keywords based on up-to-date streaming data if the streaming distribution is different from the training distribution. In that way, the re-clustering can use that up-to-date streaming distribution as the training distribution. 
%
%
%

We are aware that the keyword distribution difference can cause a long tail effect when applying the proposed padding strategies to low frequent keywords.
For instance, if a keyword only occurs in the first batch of that keyword's cluster and disappears for all subsequent batches, \textit{Padding Controller} still pads that keyword during subsequent batches when padding strategies $NH$ and $NL$ are used.
%
%
%
As seen in the above experiments, such different distribution can happen at the early streaming stages. %
%
We also note that, when the streaming distribution differs from the training distributions, it may cause some ``cold'' clusters and/or the intensive usage of the padding datasets in some  ``hot'' clusters. 
As a results, the overall streaming throughput can be slowed down.

To mitigate the above issues, \system~ offers  \textsf{flushing} and \textsf{re-encryption} operations regarding the highly varied frequency of streaming keywords.
As experimented, \textsf{flushing} could quickly reduce $1.9\sim 2.8\times$ the cache load to boost the ``cold'' clusters. In addition, \textsf{re-encryption} could lowered $64\%$ the amount of bogus pairs used by re-padding all keywords in the ``hot'' cluster.
The above result is obtained when we checked and applied the operations for every fixed time window. 
%
%
Nevertheless, it is non-trivial to optimise the padding overhead in the streaming setting and we leave it as future work.
%
%

\vspace{-10pt}
%
%





\section{Related Works} 
\label{sec:related}

\noindent \textbf{Searchable symmetric encryption}: \system~employs  SSE~\cite{SoWa00} as an underlying building block to enable single-keyword encrypted search. Curtmola et al.~\cite{CurtmolaGKO06} and Kamara et al.~\cite{KamaraPR12} formalise the security of SSE for static and dynamic databases respectively, and devise concrete constructions with sublinear search time. A line of schemes~\cite{ChaseK10,CashJJ13,JareckiJK13,Cash14,Vo20,Vo21} (just to list a few) are proposed to improve performance and expressiveness of SSE. 
Driven by leakage-abuse attacks~\cite{CashGP15,ZhangKP16, Blackstone20}, new schemes~\cite{Bost16,RaphaelBO17,SunYLS18,Sun21} with less leakage in search and update are proposed to achieve forward and backward security.
Note that, although oblivious RAM~\cite{Stefanov12,Stefanov13} provides the highest protection for the \textit{Server}'s memory access pattern, we do not consider it for~\system~due its inefficient capability in the streaming setting. In details, the approach requires more computation and storage at the \textit{Client}, and the communication between the \textit{Client} and the \textit{Server}. Also, ORAM does not hide the size of the query results, unless there is non-trivial padding.
%
%

In the meanwhile, padding countermeasures~\cite{IslamKK12,CashGP15,BostF17} are  considered as an effective approach to obfuscate the leakage during search operations of SSE. 
In particular, Islam et al.~\cite{IslamKK12} propose the first padding countermeasure for SSE; keywords are grouped into different clusters, where each keyword in a cluster matches a set of identical document ids. This requires another data structure to help the client to differentiate real and bogus document ids after search, since all bogus ids are selected from the real ones. 
After that, Cash et al.~\cite{CashGP15} propose another approach; the number of ids in each keyword matching list is padded up to the nearest multiple of an integer, aka padding factor. To guarantee effectiveness, this factor needs to be increased until no unique result size exists. However, this padding factor is a system-wide parameter, and incrementing it introduces redundant padding for all other padded matching lists. 
To reduce padding overhead, Bost and Fouque~\cite{BostF17} propose to pad the keyword matching lists based on clusters of keywords with similar frequency. Their proposed clustering algorithm achieves minimised padding overhead while thwarting the count attack in the static setting. 
Very recently, Xu et al.~\cite{XYWWX19} investigate the formal method to quantify the padding security strength, and propose a padding generation algorithm which makes the bogus and documents similar. 
%
%
Again, all the above padding countermeasures focus on the static setting, where the dataset remains unchanged after the setup. We note that the assumption in this setting is not always true in practice, and therefore \system~is designed to embed padding countermeasures in the dynamic setting, where the keyword existence is a matter in online streaming.  


Recent works~\cite{Kamara19,Patel19} propose volume-hiding encryption schemes to mitigate the leakage-abuse attacks.
We note that those schemes are focused on the static setting, as they resort to specialised data structures and constructions. First, they are not dynamic friendly. Specifically, multi-hashing and cuckoo hashing techniques are adopted as the underlying data structures. It is not easy to insert new data into those data structures, and we are not aware any existing volume-hiding schemes support efficient updates. Second, volume hiding schemes may hide the size of the query result, but it is not clear whether they can protect the relationships between different query keywords when applying them into the context of keyword search. 
%


\vspace{2pt}
\noindent \textbf{Encrypted database systems}: \system~can also be fit into a line of research on designing encrypted database systems. Most of existing encrypted databases~\cite{Poparz11,PappasVK14,PoddarBP16,YuanGWWLJ17, PapadimitriouBCRHSMB16,Zhao16,Panagiotis20,Lei19} focus on supporting rich queries over encrypted data in SQL and NoSQL databases. They mainly target on query functionality and normally integrate different primitives together to achieve the goal. Like the issues in SSE, inference attacks against encrypted databases~\cite{NaveedKW15,KellarisGKA16} are designed to compromise their claimed protection. To address this issue, one approach is to use advanced cryptographic tools such as secure multi-party computation~\cite{PappasVK14,PoddarBP16}. Note that padding can also be adapted to mitigate inference attacks. A system called Seabed~\cite{PapadimitriouBCRHSMB16} proposes a schema for RDBMS that introduces redundant data values in each attribute of data records to hide the frequency of the underlying data values. 
Compared with the above systems, \system~ focuses on the document-oriented data model and supports keyword search over encrypted documents.

%
%

\vspace{-10pt}
\section{Conclusions} %
\label{sec:conclusion}

ShieldDB is an encrypted database system that supports keyword search over encrypted documents with advanced security features. 
Our system employs the SSE framework to implement encrypted data structures for efficient queries. 
To defend against leakage-abuse attacks against SSE, ShieldDB includes effective padding countermeasures targeting adversaries in the dynamic setting. 
To demonstrate the performance of our system, we develop a prototype, and perform  intensive evaluations on various metrics.
We show that our proposed padding strategy is practical and deployable to real-world streaming applications/systems that require the privacy preservation on data stream.

\section{Acknowledgements}

The work was supported in part by the Monash University Postgraduate Publications Award, the Data61-Monash Collaborative Research Project, and the ARC Discovery Project DP200103308. The last author has been supported by the Research Grants Council of Hong Kong under Grant CityU 11217819, Grant CityU 11217620, and R6021-20F. 

\vspace{-10pt}
\bibliographystyle{IEEEtran}
\bibliography{shieldDB}

%
%
%
%
%
%
\newpage 
\appendices
\section{Security of ShieldDB}
\label{subsec:security_shielddb}

\subsection{The leakage of dynamic searchable encryption scheme}

ShieldDB implements a dynamic searchable encryption scheme (DSSE) $\Sigma = (\textsf{Setup}, \textsf{Streaming}, \textsf{Search})$, consisting of three protocols between a padding service $P$, a storage server $S$, and an querying client $C$.
A database DB$_t = (w_i,id_i)_{i=1}^{|DB_t|}$ is defined as a tuple of keyword and document id pairs with $w_i \subseteq \{0,1\}^*$ and $id_i \in \{0,1\}^l$ at the time interval $t \geq0$. 
We first formalise the SSE-based leakage functions of ~\system~ as follows.

\textsf{Setup}$(\textrm{DB}_0)$ is a protocol that takes as input a database DB$_0$, and outputs a tuple of $(k_1,k_2,\{L_1,\dots,L_m\},st,B,\textrm{EDB}_0)$, where $k_1,k_2$ are secret keys to encrypt keywords and document ids, a set $\{L_1,\dots,L_m\}$ contains cache clusters, $st$ maintains keyword states, and $B$ is a bogus dataset to be used for padding, and EDB$_0$ is the encrypted database at $t=0$.

\textsf{Streaming}$(k_1,k_2,L_u,st,B,\{(w_i,id_i)\};\textrm{EDB}_{t-1},\{(u_i,v_i)\})$ is a protocol between $P$ with inputs $k_1$,$k_2$, and $L_u$ ($1\leq u \leq m$) the cache cluster to be updated, the states $st$, the bogus dataset $B$, and the set of keyword and document id pairs $\{(w_i,id_i)\}$ to be streamed, and $S$ with input $\textrm{EDB}_{t-1}$ the encrypted database at time $t-1$ ($t \geq 1$), and $\{(u_i,v_i)\}$ the set of encrypted keyword and document identifier pairs for batch insertion. 
Once $P$ uploads $\{(u_i,v_i)\}$ to $S$, $st$ and $B$ gets updated, $L_u$ is reset. At $S$, once $\textrm{EDB}_{t-1}$ gets updated by $\{(u_i,v_i)\}$, it changes to $\textrm{EDB}_{t}$.

\textsf{Search}$(k_1,k_2,q,st;\textrm{EDB}_{t})$ is a protocol between $C$ with the keys $k_1,k_2$, the query $q$ querying the matching documents of a single keyword $w_i$, and the state $st$, and $S$ with $\textrm{EDB}_{t}$. 
Meanwhile, $C$ queries $P$ for retrieving cached documents of the query keyword.

The security of ShieldDB can be quantified via a leakage function $\mathcal{L}=(\mathcal{L}^{Stp},\mathcal{L}^{Stream},\mathcal{L}^{Srch})$. 
It defines the information exposed in \textsf{Setup}, \textsf{Streaming}, and \textsf{Search}, respectively.
%
The function ensures that ShieldDB does not reveal any information beyond the one that can be inferred from $\mathcal{L}^{Stp}$, $\mathcal{L}^{Stream}$, and $\mathcal{L}^{Srch}$.

In \textsf{Setup}, $\mathcal{L}^{Stp} = |\textrm{EDB}_0|$ presenting the size of $\textrm{EDB}_0$, i.e., the number of encrypted keyword and document id pairs.

In \textsf{Streaming}, ShieldDB is \textit{forward private} as presented in \textsf{Streaming} protocol. Hence $\mathcal{L}^{Stream}$ can be written as
\[
    \mathcal{L}^{Stream}(\{(w,id)\}) = \mathcal{L}^\prime(\{id\})
\]

\noindent where $\{(w,id)\}$ denotes a batch of keyword and id pairs $w$, and $\mathcal{L}^\prime $ is a stateless function. Hence, $\mathcal{L}^{Stream}$ only reveals the number of pairs to be added to EDB. ShieldDB does not leak any information about the updated keywords. 
In particular, $S$ cannot learn that the newly inserted documents match a keyword that being previously queried. 

In \textsf{Search}, $\mathcal{L}^{Srch}$ reveals common leakage functions~\cite{CurtmolaGKO06}: the \textit{access pattern} \textsf{ap} and the \textit{search pattern} \textsf{sp} as follows. 

The \textsf{ap} reveals the encrypted matching document identifiers associated with search tokens. For instance, if an adversary controls $\textrm{EDB}_t$, she monitors the search query list $Q_t=\{q_1,\dots,q_{n-1}\}$ by the time order. 
Then, $\textsf{ap}(q_i)$ (with $1 \leq i \leq n-1 $) for a query keyword $w_i$ is presented as 
\[
\textsf{ap}(q_i)= \textrm{EDB}(w_i) = \{(u_{w_i},v_{w_i})\}
\]
where $u_{w_i}$ and $v_{w_i}$ are an encrypted keyword and document id entry associated with $w_i$ in $\textrm{EDB}_t$. 

The \textsf{sp} leaks the repetition of search tokens sent by \textit{C} to \textit{S}, and hence, the repetition of queried keywords in those search tokens. 
\[
\textsf{sp}(q_i)= \{\forall j \neq i, q_j \in Q_t, w_j = w_i\}
\]

Next, we detail the leakage during the interaction between \textit{C} and \textit{S} over $Q_t$ on a given $\textnormal{DB}_t$. We call an instantiation of the interaction as a \textit{history} $H_t=(\textrm{DB}_t,q_1,\dots,q_{n-1})$. We note that the states of keywords in $\textnormal{DB}_t$ do not change during these queries. The leakage function of $H_t$ is presented as
\[
\mathcal{L}(H_t)=(|\textrm{EDB}(w_i)|,\dots,|\textrm{EDB}(w_{n-1})|,\alpha(H_t),\sigma(H_t))
\] 
where $|\textrm{EDB}(w_i)|$ ($1 \leq i \leq n-1$) is the number of matching documents associated with the keyword $w_i$ mapping to the query $q_i$, $\alpha(H_t)= \{\textsf{ap}(q_1),\dots,\textsf{ap}(q_{n-1})\}$ is the \textit{access pattern} induced by $Q_t$, and $\sigma(H_t)$ is a symmetric binary matrix such that for $1 \leq i,j \leq n-1$, the element at $i^{th}$ row and $j^{th}$ column is 1 if $w_i=w_j$, and 0 otherwise.

\vspace{-5pt}
\subsection{Constrained security in ~\textnormal{\system}}

We note that the database knowledge of non-persistent and persistent adversaries falls outside the traditional SSE formalisation~\cite{BostF17}.
The reason is because the notion is limited by the fact that knowing the DB, the query list is uniquely defined by the acceptable leakage of SSE.
Namely, there is already the uniqueness of a \textit{history} given the knowledge of the adversary.
%
%
Therefore, we want to define new constrained security that can formalise the adversary's knowledge in $H_t$.
But, given the constraint, there are multiple \textit{histories} at time $t$ satisfying the leakage function (i.e., making $H_t$ no longer unique). 
In this way, one needs to find two different lists of queries generating the exact same leakage with the same $\textrm{DB}_t$. 
%
As a starting point, we extend the Definition 3.1 in ~\cite{BostF17} to formalise 
%
$H_t$ satisfying  the constraint $C$ iff $C(H_t)= {\normalfont \tt{true}}$ as in Definition.~\ref{def:np_constraints}.

\begin{definition}\label{def:np_constraints}
A constraint C $=(C_0, C_1,\dots,C_{n-1})$ over a database set $\mathcal{DB}_t$ and a query set $Q_t=\{q_1, \dots, q_{n-1}\}$, is a sequence of algorithms such that, for $\textnormal{DB}_t \in \mathcal{DB}_t$, $C_0(\textnormal{DB}_t)= (flag_0,st_0)$, where $flag_0$ is ${\normalfont \tt{true}}$ or ${\normalfont \tt{false}}$ and $st_0$ captures $C_0$'s state, and for q $\in Q_t$, $C_i(q, flag_{i-1})=(flag_i), (i\geq 1)$. The constraint is consistent if $C_i(.,{\normalfont \tt{false}} ,.) = ({\normalfont \tt{false}} ,.)$ (the constraint remains ${\normalfont \tt{false}}$ if it once evaluates to $ {\normalfont \tt{false}})$.

For a history $H_t = (\textnormal{DB}_{t},q_1, \dots, q_{n-1})$, we note $C(H_t)$ the evaluation of
\begin{equation*}
\begin{split}
C(H_t) := C_{n-1}(q_{n-1},C_{n-2}(\dots,C_0(\textnormal{DB}_t))).
\end{split}
\end{equation*}
If $C(H_t)= {\normalfont \tt{true}}$, we say that $H_t$ satisfies $C$. A constraint C is valid if there exists two different efficiently constructable histories $H_t$ and $H^\prime_t$ satisfying C.
\end{definition}

After defining the knowledge in $H_t$ known by the adversary, we also formalise some elements (i.e., queries) in $H_t$ that are unknown to the adversary. Namely, they are left \textit{free} from the constraint \textit{C}. 
%
%
We note that Bost et. al. ~\cite{BostF17} already defined \textit{free} components in static database setting (i.e., not time interval $t$) for constraint security. Therefore, we extend the Definition 3.2 in ~\cite{BostF17} to formalise \textit{free} component in $C$ regarding $H_t$ in below Def.~\ref{def:np_free}.

\begin{definition}\label{def:np_free}\textnormal{(Free history component)}
We say that C lets the i-th query \textnormal{free} if for history $H_t=(\textnormal{DB}_{t},q_1, \dots, q_{n-1})$ satisfying C, for every search (resp. update) query q if $q_i$ is a search (resp. update) query, $H^\prime_t=(\textnormal{DB}^\prime_{t},q_1, \dots, q_{i-1},q,q_{i+1},\dots, q_{n-1})$ also satisfies C, where $\textnormal{DB}^\prime_{t} \in \mathcal{DB}_t$.
\end{definition} 

The idea behind of letting $i$-$th$ query \textit{free} is that there exists some other queries in the history $H_t^\prime$ such that  $H_t^\prime$ still satisfies $C$ (and both $\mathcal{L}(H_t)=\mathcal{L}(H^\prime_t)$) without modifying the leakage $\mathcal{L}(H^\prime_t)$.

Now, we also define the \textit{acceptable} constraint notion, so that, given a constraint $C$, and a leakage function $\mathcal{L}$, for every history $H_t$, we are able to find another \textit{history} satisfying $C$ with the same leakage.
\begin{definition}\label{def:np_acceptableconstraints}
A constraint C is $\mathcal{L}$-acceptable for some leakage $\mathcal{L}$ if, for every efficiently computable history $H_t$ satisfying C, there exists an efficiently computable $H^\prime_t$ $\neq $ $H_t$ satisfying C, for $H^\prime_t=(\textnormal{DB}^\prime_{t},q_1, \dots, q_{n-1})$, such that $\mathcal{L}(H_t)=\mathcal{L}(H^\prime_t)$.

A set of constraints $\mathfrak{C}$ is said to be $\mathcal{L}$-acceptable if all its elements are $\mathcal{L}$-acceptable.
\end{definition} 

Now, after giving background definitions, we start to investigate the query at time $t$. We recall that the \textit{Client} only triggers the search on completely outsourced data in $\textrm{EDB}_t$, where ($t>0$) is a random interval upon receiving search query tokens.
Therefore, we consider that the leakage function only depends on the query itself, and on the state of $\textrm{DB}_t$: $\mathcal{L}(q)$ can be presented as a stateless function of $f_\mathcal{L}(q,\textrm{DB}_t)$.
We make an observation on  $\textrm{EDB}_t$ that: let $C$ be a constraint, $H_t=(\textrm{DB}_t,q_1,\dots,q_{n-1})$ an \textit{history} satisfying $C$, and $q, q^\prime$ be two queries such that $\widetilde{H}_t=H_t||q=(\textrm{DB}_t,q_1,\dots,q_{n-1},q)$ and $\widetilde{H}_t^\prime = H_t||q^\prime=(\textrm{DB}_t,q_1,\dots,q_{n-1},q^\prime)$. 
Then, if $f_\mathcal{L}(q,\textrm{DB}_t)=f_\mathcal{L}(q^\prime,\textrm{DB}_t)$, then both $\widetilde{H}_t$ and $\widetilde{H}_t^\prime$ with the same leakage satisfying $C$.
This observation can be iterated to create multiple (i.e., more than 2) \textit{histories} using the same $\textrm{DB}_t$ and they are both satisfying $C$ with the same leakage.
Therefore, we can define a clustering $\Gamma_t=\{G_1,\dots,G_m\}$ of queries induced by the leakage $\mathcal{L}$ after history $H_t$ is a partition of a query set $Q_t$, for which, in every cluster, queries share the same leakage after running the history $H_t$ as below. 
\[\bigcup_{i=1}^m G_i= Q_t\] \vspace{-9pt}
\[\forall i \neq j ~~ G_i \cap G_j = \emptyset\] \vspace{-12pt}
\[ \textnormal{and}~\forall q, q^\prime \in G_i, \mathcal{L}(q,H_t) = \mathcal{L}(q^\prime,H_t)\]

\noindent where $\mathcal{L}(H_t,q)$ is the output of $\mathcal{L}(q)$ after having been run on each element of $H_t$. 
Note that, we omit the subscript $t$ in $\Gamma_t$ in the clear context of $H_t$; otherwise, we state it separately. 
We denote $\Gamma_\mathcal{L}(H_t)$ the clustering induced by $\mathcal{L}$ after $H_t$. 
We can see that it is impossible to merge different clusters in  $\Gamma_\mathcal{L}(H_t)$ with the same leakage. 
Therefore, formally, for $\Gamma_\mathcal{L}(H_t)=\{G_1,\dots,G_m\}$, where $m$ is the total number of clusters, we have:
\[ \forall i \neq j, \forall q \in G_i, \forall q^\prime \in G_j, \mathcal{L}(H_t,q) \neq \mathcal{L}(H_t,q^\prime)\]

\noindent We present $\Gamma_{\mathcal{L},C}(H_t)$ the $\mathcal{L}$-induced clustering applied on history $H_t$ satisfying $C$ such that a subset of queries $Q_t$ in queries $q$ gives $C(H_t||q)=\normalfont \tt{true}$.
We can see that, in the singular query $q$ earlier, $\mathcal{L}(q)$ only depends on $q$ and $\textnormal{DB}_{t}$.
Therefore, more generally, when $q$ is a set of queries with $C(H_t||q)=\normalfont \tt{true}$, $\Gamma_{\mathcal{L},C}(H_t)$ only depends on $\textnormal{DB}_{t}$.
Indeed, the clustering $\Gamma_{\mathcal{L},C}(H_t)$ acquires at least two elements in every cluster.
Otherwise,  an \textit{history} $H_t$ can be constructed without any different \textit{history} $H_t^\prime$.
Namely, we need at least $|G_i| > 2$ for $\forall i \in [1,m]$, to make sure that there are at least $2$ constrained histories can be found. 
Therefore, we can extend Def.~\ref{def:np_acceptableconstraints} to have an acceptable constraint $C$ with $\alpha$ histories, where $|G_i| > \alpha$.
We note that the notation $\alpha$ here is inline with the clustering algorithm in \textsf{setup} in ~\system ~ (Section 4.1).

\begin{definition}\label{def:np_acceptableconstraints_k}\textnormal{(Extended acceptable constraint)}
A constraint C is $(\mathcal{L},\alpha)$-acceptable for some leakage $\mathcal{L}$ and integer $\alpha>1$ if, for every efficiently computable history $H_t^0$ satisfying C $(i.e., C(H_t^0) = {\normalfont \tt{true}})$, there exists $(\alpha-1)$ efficiently computable $\{H^i_t\}_{1\leq i \leq \alpha-1}$  such that $H^i_t \neq H^j_t$ for $i\neq j$, that are all satisfying C, and $\mathcal{L}(H^0_t) =\dots =\mathcal{L}(H^{\alpha-1})$.
\end{definition} 

Now, we can see that when $|G_i|\geq \alpha$, for $\forall i \in [1,m]$, i.e.,  strictly more than one element in each cluster of  $\Gamma_{\mathcal{L},C}(H_t)$, $C$ is $(\mathcal{L},\alpha)$-acceptable, as formalised in below Proposition~\ref{prop:k_acceptable}.

\begin{proposition}\label{prop:k_acceptable}
Let C be a constraint, and $\mathcal{L}$ a leakage function. If for every history $H_t$ satisfying $C$, the clustering $\Gamma_{\mathcal{L},C}(H_t)=\{G_1,\dots,G_m\}$ is such that $|G_i|\geq \alpha$ for all i, $C$ is $(\mathcal{L},\alpha)-acceptable$.
\end{proposition}


\subsection{Security against Non-persistent Adversary}
\label{subsec:security_ana_np}
\noindent{\textbf{Prior knowledge of the database}}:
Considering the adversary's knowledge of the database is $\textnormal{DB}_{t}$ when she captures $\mathcal{L}(H_t)$, 
we use the predicate $C^{\textnormal{DB}_t}$ to formalise this knowledge, by adapting the notion of server's knowledge in~\cite{BostF17}.
Formally, we have $C^{\textnormal{DB}_t}(H_t)= {\normalfont \tt{true}}$ if the database of the input history is $\textnormal{DB}_t$.
As used in Definition \ref{def:np_acceptableconstraints}, $C^{\textnormal{DB}_t}$ ensures that all challenge histories' database is the same, i.e., $\textnormal{DB}_t$.
That also leave all queries in $\textnormal{DB}_t$ are left \textit{free}, as defined in Definition~\ref{def:np_free}. 
More generally, we can model the fact that the adversary know the database by considering the constraint set $\mathfrak{C}^{\mathcal{DB}_t}=\{C^{\textnormal{DB}_{t}},{\textnormal{DB}_{t}}\in \mathcal{DB}_t\}$.

Now, we recall that the non-persistent adversary only captures an interval $t$, and \textsf{search} only triggers on encrypted entries inserted in $\textrm{EDB}_t$. Therefore, we consider the scheme  $\Sigma_{NP}=(\textsf{Setup},\textsf{Search})$ at time $t$ for the non-persistent adversary. We start adding a padding mechanism presented in Algorithm 1 (i.e., Padding Strategies) to $\Sigma$ such that, for every keyword in $\textnormal{DB}_{t}$, there are at least different $(\alpha-1)$ keywords with the same number of matching documents. 
%
%
Then, with the knowledge of ${\textnormal{DB}_t}$, the leakage function of $\Sigma_{NP}$ is formally defined as $\mathcal{L}_{NP}=(\mathcal{L}^{Stp},\mathcal{L}^{Srch},\mathcal{L}^{\alpha-pad})$, where $\mathcal{L}^{Stp}$ and $\mathcal{L}^{Srch}$ reveals the leakage in \textsf{Setup} at at $t$ and \textsf{Search} against $\textrm{EDB}_t$, respectively (see Section~\ref{subsec:security_shielddb}), and the new leakage $\mathcal{L}^{\alpha-pad}$ reveals the minimum size of clusters induced by $\mathcal{L}^{\alpha-pad}$.

By using Proposition~\ref{prop:k_acceptable}, we can show that $\mathfrak{C}^{\mathcal{DB}_t}$ is an $(\mathcal{L}_{NP},\alpha)$-acceptable set of constraints, where $\alpha$ is the minimum cluster size (over all constructable databases). 
The reason is that, since constraints in $\mathfrak{C}^{\mathcal{DB}_t}$ leave all queries \textit{free} for every history $H_t=(\textnormal{DB}_{t},q_1, \dots, q_{n-1})$, we can generate a different history $H^\prime_t$ with the same leakage by choosing another  query $q \neq q_1$ that are both matching the same number of documents, and changing all queries $q_i=q_1$ in $H_t$ to $q$. Also, if there is queries $q_j = q$ in $H_t$, we can switch queries in $q_j$ to $q_1$. 
This can give us a history $H^\prime_t\neq H_t$ with the same leakage of $H_t$.
We note that there are at least $\alpha$ choices of $q$ to create  $\Gamma_{\mathcal{L}_{NP}}(H_t)$ we can derive $\mathfrak{C}^{\mathcal{DB}_t}(\mathcal{L}_{NP},\alpha)$-acceptable.
%

Now, we are ready to define the notion of constrained adaptive indistinguishability for $\Sigma_{NP}$ given $\mathfrak{C}^{\mathcal{DB}_t}$ and the leakage $(\mathcal{L}_{NP})$.

%
%

\begin{definition}\label{def:security_n_adversary}
 Let $\Sigma_{NP}$ = {\normalfont(\textsf{Setup,Search})} be the SSE scheme of \system,  $\lambda $ be the security parameter, and $\mathcal{A}$ be a non-persistent adversary. Let $\mathfrak{C}^{\mathcal{DB}_t}$ be a set of  $(\mathcal{L}_{NP},\alpha)$-acceptable constraints.  Let {\normalfont Ind$_{SSE,\mathcal{A},\mathcal{L}_{NP}, \mathfrak{C}^{\mathcal{DB}_t},\alpha}$} be the following game:

\begin{minipage}[t]{6.5cm}

{\normalfont Ind$_{\mathrm{SSE},\mathcal{A},\mathcal{L}_{NP}, \mathfrak{C}^{\mathcal{DB}_t},\alpha}$}($\lambda$) \textit{Game}:\\
$\hspace*{3mm}b\xleftarrow{\$}\{0, \dots, \alpha-1\}$\\
$\hspace*{3mm}(C^{\textnormal{DB}_{t}}_0,\textnormal{DB}^0_t,\dots,\textnormal{DB}^{\alpha-1}_t)\leftarrow \mathcal{A}(1^{\lambda})$\\
$\hspace*{3mm}(K,\textnormal{EDB}^b_t)\leftarrow {\normalfont\textsf{Setup}}(\textnormal{DB}^b_t)$\\
$\hspace*{3mm}(C^{\textnormal{DB}_{t}}_1,q^0_1,\dots,q^{\alpha-1}_1)\leftarrow \mathcal{A}(\textnormal{EDB}^b_t)$\\
$\hspace*{3mm}\tau^b_{1}\leftarrow \normalfont\textsf{Search}(q^b_1)$\\
\hspace*{3mm}\textbf{for} $i = 2$ \textbf{to} $n$ \textbf{do}\\
$\hspace*{6mm}(C^{\textnormal{DB}_{t}}_i,q^0_i,\dots,q^{\alpha-1}_i)\leftarrow \mathcal{A}(q^b_{i-1})$\\
$\hspace*{6mm}\tau^b_{i}\leftarrow \normalfont\textsf{Search}(q^b_i)$\\
\hspace*{3mm}\textbf{end for}\\
$\hspace*{3mm}b^\prime \leftarrow \mathcal{A}(\tau^b_{n})$\\
\hspace*{3mm}\textbf{if} b $=$ b$^\prime$ \textbf{return} 1, \textbf{otherwise return} 0\\ 
\end{minipage}

where $\tau^b_{i}\leftarrow \normalfont\textsf{Search}(q^b_i)$ presents the transcript of the query $q^b_i$, and with the restriction that, for all the $H^i_t=(\textnormal{DB}^i_t,q^0_i,\dots, q^{n-1}_i)$,
\begin{itemize}
  \setlength{\itemsep}{1pt}
  \setlength{\parskip}{0pt}
  \setlength{\parsep}{0pt}
\item $C^{\textnormal{DB}_{t}} \in \mathfrak{C}^{\mathcal{DB}_t}, and~ \forall 0\leq i \leq (\alpha-1), C^{\textnormal{DB}_t}(H^i_t) = {\normalfont \tt{true}}$
\item $\mathcal{L}(H^0_t)=\dots=\mathcal{L}(H^{\alpha-1}_t)$
\end{itemize}

\noindent We say that $\Sigma$ is $(\mathcal{L}_{NP},\mathfrak{C}^{\mathcal{DB}_t},\alpha)$-constrained-adaptively-indistinguishable if for all probabilistic polynomial time adversary $\mathcal{A},$
\begin{equation}
\begin{split}
\label{eq:nonpersistent_alpha}
 & {\normalfont  \textbf{Adv}^{\mathrm{Ind}}_{\mathcal{A},\mathcal{L}_{NP},\mathfrak{C}^{\mathcal{DB}_t},\alpha}}  (\lambda) = \\
      &    \bigl| \mathbb{P}[{\mathrm{Ind}_{\mathrm{SSE},\mathcal{A},\mathcal{L}_{NP},\mathfrak{C}^{\mathcal{DB}_t},\alpha}}(\lambda)=1] - \frac{1}{\alpha} \bigr| \leq negl(\lambda).
\end{split}
\end{equation}
\end{definition}  

We can see that $\Sigma_{NP}$ offers at least $log(\alpha)$ bits of security. Given $\mathfrak{C}^{\mathcal{DB}_t}$ $(\mathcal{L}_{NP},\alpha)$-acceptable, we can analysing the transcripts $\tau^b_{i}$ under the choice of $\alpha$.
%
%
First, we make an observation on the keyword choice in  $(\textnormal{DB}^0_t,\dots,\textnormal{DB}^{\alpha-1}_t)$ as follows.
We denote by $\Delta^i_t=\{w^i_1,\dots,w^i_{n-1}\}$ the keyword space of $\textnormal{DB}^i_t$, where $i \in \{0,\alpha-1\}$.
Then, $C^{\textnormal{DB}_t}(H^i_t) = {\normalfont \tt{true}}$ and all $\mathcal{L}(H^0_t)=\dots=\mathcal{L}(H^{\alpha-1}_t)$ imply  $\Delta^0_t = \dots = \Delta^{\alpha-1}_t$.
Let $f(w)$ be a function returning the frequency of the keyword $w$, we can see that, for all $w^0_j \in \Delta^0_t$, where $j \in |\Delta^0_t|$, there are at least one another $w^i_j$ in $\Delta^i_t$ (i.e., $\forall i\neq 0$) such that $f_{w^0_i} = f_{w^i_i}$. 
This turns out that the \textsf{Setup} needs to groups at least $\alpha$ keywords and pad them to be the same length such that, for a given $q^b_i$ in \textsf{Search}, under the chosen $b$, the transcript $\tau^b_{i}$ can be hardened by at least $(\alpha-1)$ choices.

Now, we adapt the Theorem 2 in~\cite{BostF17} to prove the extended constrained indistinguishability (i.e., Definition~\ref{def:security_n_adversary})  by using regular leakage indistinguishability and extended acceptability of constraint set $\mathfrak{C}^{\mathcal{DB}_t}$ as follows. 

\begin{theorem}\label{theo:theorem_non_persistent}
Let $\Sigma_{NP}$ = {\normalfont(\textsf{Setup,Search})} be our SSE scheme, and $\mathfrak{C}^{\mathcal{DB}_t}$ a set of knowledge constraints. If $\Sigma_{NP}$ is $\mathcal{L}_{NP}$-constrained-adaptively-indistinguishable secure, and  $\mathfrak{C}^{\mathcal{DB}_t}$ is   $(\mathcal{L}_{NP},\alpha)$-acceptable, then $\Sigma_{NP}$ is ($\mathcal{L}_{NP}$,$\mathfrak{C}^{\mathcal{DB}_t},\alpha$)-constrained-adaptively-indistinguishability secure.
\end{theorem}

\begin{proof}

Let $\mathcal{A}$ be an adversary in the {\normalfont Ind$_{\mathrm{SSE},\mathcal{A},\mathcal{L}_{NP}, \mathfrak{C}^{\mathcal{DB}_t},\alpha}$} game. 
We construct an adversary $\mathcal{B}$ against the game. $\mathcal{B}$ first randomly picks two integer $\alpha_0,\alpha_1 \in \{0,\alpha-1\}$. 
Then, $\mathcal{B}$ starts $\mathcal{A}$ and receives $\alpha$ databases $(\textnormal{DB}^0_t,\dots,\textnormal{DB}^{\alpha-1}_t)$. 
Upon giving the pair $(\textnormal{DB}^{\alpha_0}_t,\textnormal{DB}^{\alpha_1}_t)$ to the challenger, where the challenger holds a random secret bit $b$, $\mathcal{B}$ receives the challenge encrypted database $\textnormal{EDB}^*_t$ which she forwards to $\mathcal{A}$. 
Then, $\mathcal{A}$ repeatedly outputs $\alpha$ queries $(q^0_i,\dots,q^{\alpha-1}_i)$ and gives to $\mathcal{B}$. To respond, $\mathcal{B}$ outputs $(q^{\alpha_0}_i,q^{\alpha_1}_i)$ to the game, and receives back the transcript $\tau^*_{i}$ and forwards it to $\mathcal{A}$. 
Then, $\mathcal{A}$ outputs the integer $\alpha^\prime$. If $\alpha^\prime=\alpha_0$, $\mathcal{B}$ outputs $b^\prime=0$, else if $\alpha^\prime=\alpha_1$, $\mathcal{B}$ outputs $b^\prime=1$, and otherwise outputs the probability $1/2$ for the output $0$ and the probability $1/2$ for the output $1$.
%

We first make an observation: for the pair $(H^{\alpha_0}_t,H^{\alpha_1}_t)$, the views of the adversary  $\mathcal{B}$ are indistinguishable due to   $\mathcal{L}_{NP}(H^{\alpha_0}_t)=\mathcal{L}_{NP}(H^{\alpha_1}_t)$, presenting both satisfying $\mathfrak{C}^{\mathcal{DB}_t}$. 
Then we can formalise $\mathcal{B}$ as follows:
\begin{equation}
\begin{split}
\label{eq:adv_b}
{\normalfont  \textbf{Adv}^{\mathrm{Ind}}_{\mathcal{B},\mathcal{L}_{NP},\mathfrak{C}^{\mathcal{DB}_t}}}  (\lambda) =  \bigl|
           \mathbb{P}[b=b^\prime] - \frac{1}{2} \bigr| \leq negl(\lambda) 
\end{split}
\end{equation}
Now, we evaluate $ \mathbb{P}[b=b^\prime]$ as follows.
\begin{equation}
\begin{split}
\label{eg:probability_B}
\mathbb{P}[b=b^\prime] & =  \\
          &~\mathbb{P}[ b=b^\prime| \alpha^\prime \in \{\alpha_0,\alpha_1\}]  \cdot  \mathbb{P}[\alpha^\prime \in \{\alpha_0,\alpha_1\}] \\
        & +\mathbb{P}[b=b^\prime|\alpha^\prime \notin \{\alpha_0,\alpha_1\}] \cdot  \mathbb{P}[\alpha^\prime \notin \{\alpha_0,\alpha_1\}] \\
        & = \mathbb{P}[ b=b^\prime \cap \alpha^\prime \in \{\alpha_0,\alpha_1\}] \\
        & +\mathbb{P}[b=b^\prime|\alpha^\prime \notin \{\alpha_0,\alpha_1\}] \cdot  \mathbb{P}[\alpha^\prime \notin \{\alpha_0,\alpha_1\}] \\
        & = \mathbb{P}[\mathcal{A} \textnormal{ wins the } {\normalfont Ind_{\mathrm{SSE},\mathcal{A},\mathcal{L}_{NP}, \mathfrak{C}^{\mathcal{DB}_t},\alpha}} \textnormal{ game}] \\
        & + \frac{1}{2}\left(1-\mathbb{P}[\alpha^\prime \in \{\alpha_0,\alpha_1\}]\right)
\end{split}
\end{equation}

Now, we evaluate $\mathbb{P}[\alpha^\prime \in \{\alpha_0,\alpha_1\}]$ as follows.
\begin{equation*}
\begin{split}
\mathbb{P}[\alpha^\prime \in \{\alpha_0,\alpha_1\}] & = \mathbb{P}[ \alpha^\prime=\alpha_0] + \mathbb{P}[ \alpha^\prime=\alpha_1]
\end{split}
\end{equation*}

Since we have
\begin{equation*}
\begin{split}
\mathbb{P}[ \alpha^\prime=\alpha_0] + & \mathbb{P}[ \alpha^\prime=\alpha_1]  = \\
& \mathbb{P}[ \alpha^\prime=\alpha_b| b=0] + \mathbb{P}[ \alpha^\prime=\alpha_b |b=1]
\end{split}
\end{equation*}
then,
\begin{equation*}
\begin{split}
\mathbb{P}[\alpha^\prime \in \{\alpha_0,\alpha_1\}] & = \frac{1}{2}\left (\mathbb{P}[ \alpha^\prime=\alpha_b| b=0] +  \mathbb{P}[ \alpha^\prime=\alpha_0]  \right) \\
 & + \frac{1}{2}\left (\mathbb{P}[ \alpha^\prime=\alpha_b| b=1] +  \mathbb{P}[ \alpha^\prime=\alpha_1]  \right)
\end{split}
\end{equation*}
We note that $\mathbb{P}[ \alpha^\prime=\alpha_b]$ is the probability $\mathcal{A}$ wins the 1-out-of-$\alpha$ indistinguishability game, and $\alpha_0$ and $\alpha_1$ are uniformly selected from $\{0,\alpha-1\}$, then we have
\begin{equation}
\begin{split}
\label{eq:alpha_12}
\mathbb{P}[\alpha^\prime & \in \{\alpha_0,\alpha_1\}]  = \\ &\mathbb{P}[\mathcal{A} \textnormal{ wins the } {\normalfont Ind_{\mathrm{SSE},\mathcal{A},\mathcal{L}_{NP}, \mathfrak{C}^{\mathcal{DB}_t},\alpha}} \textnormal{ game}] 
+ \frac{1}{\alpha}
\end{split}
\end{equation}

\noindent Applying Eq.~\ref{eq:alpha_12} to Eq.~\ref{eg:probability_B}, we have
\begin{equation*}
\begin{split}
\mathbb{P}[b=b^\prime] & = \\ & \frac{1}{2} \cdot\mathbb{P}[\mathcal{A} \textnormal{ wins the } {\normalfont Ind_{\mathrm{SSE},\mathcal{A},\mathcal{L}_{NP}, \mathfrak{C}^{\mathcal{DB}_t},\alpha}} \textnormal{ game}] \\
& + \frac{1}{2} - \frac{1}{2\alpha}
\end{split}
\end{equation*}

\noindent Then, from Equation~\ref{eq:adv_b}, we can derive
\begin{equation}
\begin{split}
\label{eq:adv_b_derieved}
& {\normalfont   \textbf{Adv}^{\mathrm{Ind}}_{\mathcal{B},\mathcal{L}_{NP},\mathfrak{C}^{\mathcal{DB}_t}}}  (\lambda) =\\
& \frac{1}{2} \left(\mathbb{P}[\mathcal{A} \textnormal{ wins the } {\normalfont Ind_{\mathrm{SSE},\mathcal{A},\mathcal{L}_{NP}, \mathfrak{C}^{\mathcal{DB}_t},\alpha}} \textnormal{ game}] - \frac{1}{\alpha}\right)
\end{split}
\end{equation}

\noindent Applying Equation~\ref{eq:adv_b_derieved} to Equation~\ref{eq:nonpersistent_alpha}, finally, we can have

\begin{equation}
\begin{split}
\label{eq:adv_b_a_adv}
 {\normalfont   \textbf{Adv}^{\mathrm{Ind}}_{\mathcal{B},\mathcal{L}_{NP},\mathfrak{C}^{\mathcal{DB}_t}}}  (\lambda) =
\frac{1}{2}{\normalfont  \textbf{Adv}^{\mathrm{Ind}}_{\mathcal{A},\mathcal{L}_{NP},\mathfrak{C}^{\mathcal{DB}_t},\alpha}}  (\lambda)
\end{split}
\end{equation}

\end{proof}

\subsection{Security against Persistent Adversary}
\label{subsec:persistentdefinitions}

\noindent \textbf{Prior knowledge of the databases}
We start to generalise the knowledge of the non-persistent adversary over the time to be the persistent adversary's knowledge of databases. 
Namely, we denote by $(\textnormal{DB}_{t=0},\dots, \textnormal{DB}_{t=T})$ such knowledge, where $T$ denotes the streaming period.
We also make use of the constraint set $\mathfrak{C}^{\mathcal{DB}_t}=\{C^{\textnormal{DB}_{t}},{\textnormal{DB}_{t}}\in \mathcal{DB}_t\}$, defined in  Section~\ref{subsec:security_ana_np}, to formulate
$\mathbf{C}^{[1,\dots,T]}=\{\mathfrak{C}^{\mathcal{DB}_0},\dots,\mathfrak{C}^{\mathcal{DB}_T}\}$ the generalisation of constraint sets over the period such that we know every $\mathfrak{C}^{\mathcal{DB}_t}(\mathcal{L}_{NP},\alpha)$-acceptable, $\forall t\in [0,T]$.

%
Let $\delta$ be a stateless function that outputs the keyword set difference in a pairwise different inputs. Then,  $\delta(\textnormal{DB}_{t},\textnormal{DB}_{t^\prime})=W^{t,t^\prime}$, where $W^{t,t^\prime}=\{w_i | w_i \in \textnormal{DB}_{t}, w_i \notin \textnormal{DB}_{t^\prime}\}$. 
We consider the leakage function only depends on the  query itself $(i.e., q(.))$ and on the state of database at the querying time: $\mathcal{L}_{NP}(q)_t$ and $\mathcal{L}_{NP}(q)_{t^\prime}$  be presented as  stateless functions of $f_{\mathcal{L}_{NP}}(q,\textnormal{DB}_{t})$, and $f_{\mathcal{L}_{NP}}(q,\textnormal{DB}_{t^\prime})$, respectively.

%
Let $Q=\{q_1, \dots, q_{n-1}\}$ be a query set, and $C^{\textnormal{DB}_{t}}$ and $C^{\textnormal{DB}_{t^\prime}}$  be  constraints applied on the $H_t=(\textnormal{DB}_{t},Q)$ and $H_{t^\prime}=(\textnormal{DB}_{t^\prime},Q)$, respectively.
%
From Proposition~\ref{prop:k_acceptable}, we denote by  $\Gamma_{\mathcal{L}_{NP},C^{\textnormal{DB}_t}}(H_t)=\{G_{1,t},\dots,G_{m,t}\}$  the 
clustering $\Gamma_t=\{G_{1,t},\dots,G_{m,t}\}$ of queries induced by the leakage $\mathcal{L}_{NP}$ after history $H_t$ is a partition of $Q$, for which, in every cluster, queries share the same leakage after running the history $H_t$. 
Similarly, we also derive the clustering $\Gamma_{t^\prime}=\{G_{1,t^\prime},\dots,G_{m^\prime,t^\prime}\}$ of queries $Q_t$ induced by the leakage $\mathcal{L}_{NP}$ after history $H_t^\prime$.
We make an observation on $\textrm{EDB}_t$  and $\textrm{EDB}_{t^\prime}$ that:
let $q(w)$ (resp.  $q(w^\prime)$) be the query of $w$ (resp. $w^\prime$), where $ w \notin W^{t,t^\prime}, w^\prime \in W^{t,t^\prime}$, 
and  $q(w), q(w^\prime) \in G_{v,t}$, ($\exists v \in [1,m]$), such that: 
\[ (\widetilde{H}^{w}_t=H_t||q(w),
\widetilde{H}^{w^\prime}_{t}=H_{t}||q(w^\prime))\]
and 
\[(\widetilde{H}^{w}_{t^\prime}=H_{t^\prime}||q(w), \widetilde{H}^{w^\prime}_{t^\prime}=H_{t^\prime}||q(w^\prime))\]
We can see that $f_{\mathcal{L}_{NP}}(q(w),\textrm{DB}_t)=f_{\mathcal{L}_{NP}}(q(w^\prime),\textrm{DB}_t)$, but $f_{\mathcal{L}_{NP}}(q(w),\textrm{DB}_{t^\prime})\neq f_{\mathcal{L}_{NP}}(q(w^\prime),\textrm{DB}_{t^\prime})$
due to $w \in \textnormal{DB}_{t^\prime}$ while $w^\prime \notin \textnormal{DB}_{t^\prime}$, causing $\textrm{EDB}_{t^\prime}(w) \neq \textrm{EDB}_{t^\prime}(w^\prime)$.
We can see that: if both $q(w), q(w^\prime) \in G_{v,t^\prime}$, then $f_{\mathcal{L}_{NP}}(q(w),\textrm{DB}_{t^\prime}) = f_{\mathcal{L}_{NP}}(q(w^\prime),\textrm{DB}_{t^\prime})$.
This observation can be iterated for all other pairwise different queries in $G_{v_t}$ and $G_{v_{t^\prime}}$.
%
%
More generally, we need to have $|G_{v,t}| = |G_{v,t^\prime}|$, $\forall t,t \in [0,T]$ such that $\forall q(w),q(w^\prime) \in G_{v,t}$, they always have the same leakage at different $t^\prime \in [1,T]$.
Formally, given $m$ the number of clusters, we define
\[ \forall G_{v,t}, G_{v,t^\prime} \neq \emptyset,  F_{t,t^\prime,v}(G_{v,t},G_{v,t^\prime})=(|G_{v,t}| \stackrel{?}{=} |G_{v,t^\prime}|) \]
%
Now, we also use \textit{constrained} security to formalise that leakage over $\textrm{EDB}_t$ and $\textrm{EDB}_{t-1}$, $\forall t\geq 1$ as in.

\begin{definition}\label{def:knowledge_constraints}
A constraint $\mathbb{F}^t=(F_{t,t-1,1},\dots,F_{t,t-1,m})$ over two clusters $\Gamma_t=\{G_{1,t},\dots,G_{m,t}\}$ and $\Gamma_t^\prime=\{G_{1,t^\prime},\dots,G_{m,t^\prime}\}$ 
induced by the leakage $\mathcal{L}_{NP}$ after histories $H_t$ (resp. $H_{t-1}$) over $Q_t$ (resp. $Q_{t-1}$), is a sequence of algorithms such that $F(t,t-1,i)=flag_i$, where $flag_i$ is ${\normalfont \tt{true}}$ or ${\normalfont \tt{false}}$. The constraint is consistent if $(.,{\normalfont \tt{false}} ,.)=({\normalfont \tt{false}},.)$, then $\mathbb{F}^t= {\normalfont \tt{false}}$ (the constraint remains ${\normalfont \tt{false}}$ if it once evaluates to $ {\normalfont \tt{false}})$.

Let a constraint set $\mathcal{F}=(\mathbb{F}^1,\dots,\mathbb{F}^T)$ is a sequence of algorithms evaluated at every $t\in[1,T]$. The set is consistent if $(.,{\normalfont \tt{false}} ,.)=({\normalfont \tt{false}},.)$, then $\mathcal{F}= {\normalfont \tt{false}}$ (the constraint remains ${\normalfont \tt{false}}$ if it once evaluates to $ {\normalfont \tt{false}})$. In a short form, we write $\mathcal{F}^T$ to present the condition $\mathcal{F}={\normalfont \tt{true}}$. 
\end{definition} 

\vspace{-5pt}
%
%
%
%
%
%
%
The security of a scheme $\Sigma_{P}=(\textsf{Setup},\textsf{Stream},\textsf{Search})$ against the persistent adversary over a streaming period $T$.
We start adding a padding mechanism against the persistent adversary in Algorithm 1 (i.e., Padding Strategies) to $\Sigma_{P}$ such that, $\forall G_{i,t}$ in $\Gamma_t=\{G_{1,t},\dots,G_{m,t}\}$ induced by $\textsf{Search}_t$ (i.e., searching at time $t$) against  $\textrm{EDB}_t$, $G_{i,t}$ always has the same size $|G_{i,t}|=|G_{i,t^\prime}|$, where $G_{i,t^\prime}$ in $\Gamma_{t^\prime}=\{G_{1,t^\prime},\dots,G_{m,t^\prime}\}$, $\forall t^\prime \neq t$.
%
%
Let $\mathcal{L}^{Stream}_{[1,\dots,T]}=\{{L^{Stream}_1},\dots,{{L^{Stream}_T}}\}$, and $\mathcal{L}^{Search}_{[1,\dots,T]}=\{\mathcal{L}^{Srch}_{NP,i}\}$, for $i\in [1,\dots,T]$, where $\mathcal{L}^{Srch}_{NP,i}=(\mathcal{L}^{Srch}_i,\mathcal{L}^{\alpha-pad}_i)$ is the search leakage at time $i$, where  $\mathcal{L}^{\alpha-pad}_i)$ reveals the sizes of clusters induced by $\mathcal{L}^{\alpha-pad}_i$. 
Then, the leakage function of $\Sigma_{P}$ can be quantified via the leakage function $\mathcal{L}_{P}=(\mathcal{L}^{Stp},\mathcal{L}^{Stream}_{[1,\dots,T]},\mathcal{L}^{Search}_{[1,\dots,T]})$. 

By using Definition~\ref{def:knowledge_constraints}, we can show that $\mathbf{C}^{[1,\dots,T]}=\{\mathfrak{C}^{\mathcal{DB}_0},\dots,\mathfrak{C}^{\mathcal{DB}_T}\}$ is an $(\mathcal{L}_{P},\alpha,\mathcal{F})$-acceptable set of constraint sets, where $\alpha$ denotes the minimum cluster size (over all constructable databases) and $\mathcal{F}^T$ denotes the condition   $\mathcal{F}={\normalfont \tt{true}}$. The reason is that, for every time $t\in [0,T]$, $\mathfrak{C}^{\mathcal{DB}_t}$ is $(\mathcal{L}_{NP},\alpha)$-acceptable, and 
$\forall G_{i,t}$ in $\Gamma_t=\{G_{1,t},\dots,G_{m,t}\}$ induced by $\textsf{Search}_t$ (i.e., searching at time $t$) against  $\textrm{EDB}_t$, $G_{i,t}$ always has the same size $|G_{i,t}|=|G_{i,t^\prime}|$, where $G_{i,t^\prime}$ in $\Gamma_{t^\prime}=\{G_{1,t^\prime},\dots,G_{m,t^\prime}\}$, $\forall t^\prime \neq t$
%
%
%
%
%

Now, we are ready to define the notion of constraned adaptive indistinguishability for $\Sigma_{P}$ given $\mathbf{C}^{[1,\dots,T]}$. This security game is formalised in Definition~\ref{def:security_p_adversary}.

\begin{definition}\label{def:security_p_adversary}
 Let $\Sigma$ = {\normalfont(\textsf{Setup,Streaming,Search})} be the DSSE scheme of \system,  $\lambda $ be the security parameter, and $\mathcal{A}$ be a persistent adversary. Let $\mathbf{C}^{[1,\dots,T]}$  be a set of $(\mathcal{L}_P,\alpha,\mathcal{F})$-acceptable constraint sets. %
 Let $u_t$ (\textsf{streaming}) (resp. $Q_t$ (\textsf{search})) be an update (resp. query set, i.e., $Q_t=\{q_{t,1},\dots,q_{t,n}\}$) at time $t$, and $\mathcal{L}^{Stream}$ (resp. $\mathcal{L}^{Srch}$) be the leakage after the query \textit{u}(resp. \textit{q}), respectively. 
 Let {\normalfont Ind$_{DSSE,\mathcal{A},\mathcal{L}_{P}, \mathbf{C}^{[1,\dots,T]},\alpha,\mathcal{F}}$} be the following game:

\begin{minipage}[t]{6.5cm}

{\normalfont Ind$_{DSSE,\mathcal{A},\mathcal{L}_{P}, \mathbf{C}^{[1,\dots,T]},\alpha,\mathcal{F}}$}($\lambda$) ~\textit{Game}:\\
$\hspace*{3mm}b\xleftarrow{\$}\{0, \dots, \alpha-1\}$\\
$\hspace*{3mm}(\mathfrak{C}^{\mathcal{DB}_0},\textnormal{DB}^0_t,\dots,\textnormal{DB}^{\alpha-1}_t)\leftarrow \mathcal{A}(1^{\lambda})$\\
$\hspace*{3mm}(K,\textnormal{EDB}^b_0)\leftarrow {\normalfont\textsf{Setup}}(\textnormal{DB}^b_0)$\\
\hspace*{3mm}\textbf{for} $t = 1$ \textbf{to} $T$ \textbf
{do}\\
$\hspace*{8mm}\textnormal{EDB}^b_t \leftarrow \normalfont\textsf{Streaming}(u_t)$\\
$\hspace*{8mm}(\mathfrak{C}^{\mathcal{DB}_t},Q_t)\leftarrow \mathcal{A}(\textnormal{EDB}^b_t)$\\
$\hspace*{8mm}\{\tau^b_{t,1},\dots,\tau^b_{t,n}\} \leftarrow \normalfont\textsf{Search}(Q_t,\textnormal{EDB}^b_t)$\\
\hspace*{3mm}\textbf{end for}\\
$\hspace*{3mm}b^\prime \leftarrow \mathcal{A}(\tau^b_{T,1},\dots,\tau^b_{T,n})$\\
\hspace*{3mm}\textbf{if} b $=$ b$^\prime$ \textbf{return} 1, \textbf{otherwise return} 0\\ 
\end{minipage}

where $\tau^b_{t,i}$ presents the transcript of the query $q_{t,i}$, and with the restriction that, given $\mathbf{C}^{[1,\dots,T]}=\{\mathfrak{C}^{\mathcal{DB}_0},\dots,\mathfrak{C}^{\mathcal{DB}_T}\}$,  for $\mathfrak{C}^{\mathcal{DB}_t}=\{C^{\textnormal{DB}_{t}},{\textnormal{DB}_{t}}\in \mathcal{DB}_t\}$, for all the $H^i_t=(\textnormal{DB}^i_t,q_{t,1},\dots, q_{t,n})$,
\begin{itemize}
  \setlength{\itemsep}{1pt}
  \setlength{\parskip}{0pt}
  \setlength{\parsep}{0pt}
\item $C^{\textnormal{DB}_{t}} \in \mathfrak{C}^{\mathcal{DB}_t}, and~ \forall 0\leq i \leq (\alpha-1), C^{\textnormal{DB}_t}(H^i_t) = {\normalfont \tt{true}}$
\item $\mathcal{L}(H^0_t)=\dots=\mathcal{L}(H^{\alpha-1}_t)$
\item $\mathcal{F}=(\mathbb{F}^1,\dots,\mathbb{F}^T)={\normalfont \tt{true}}$
\end{itemize}

\noindent We say that $\Sigma$ is $(\mathcal{L}_{P}, \mathbf{C}^{[1,\dots,T]},\alpha,\mathcal{F}   )$-constrained-adaptively-indistinguishable if for all probabilistic polynomial time adversary $\mathcal{A},$
\begin{equation}
\begin{split}
\label{eq:persistent_alpha_f}
 & {\normalfont  \textbf{Adv}^{\mathrm{Ind}}_{\mathcal{A},\mathcal{L}_{P},\mathbf{C}^{[1,\dots,T]},\alpha,\mathcal{F}}}  (\lambda) = \\
      &    \bigl| \mathbb{P}[{\mathrm{Ind}_{\mathrm{DSSE},\mathcal{A},\mathcal{L}_{P},\mathbf{C}^{[1,\dots,T]},\alpha,\mathcal{F}}}(\lambda)=1] - \frac{1}{\alpha} \bigr| \leq negl(\lambda).
\end{split}
\end{equation}
\end{definition}  

%
%
%
%

We can see that $\Sigma_P$ offers at least $log(\alpha)$ bits of security. Given $\mathbf{C}^{[1,\dots,T]}$ is an $(\mathcal{L}_{P},\alpha,\mathcal{F})$-acceptable set of constraint sets, we can analysing every transcript in the set  $\{\tau^b_{t,1},\dots,\tau^b_{t,n}\}$ under the choice of $\alpha$.
Let $u_t$ simply contains a pair of $(w_i,id)$, we can make an observation in~\system~ as follows. 

If $ST[w_i].c=0$, presenting that $w_i$ appears in the first time, then \textit{Padding Controller} checks the cache cluster that expects to have $w_i$ against the \textit{first batch} condition. 
We recall that the condition ensures the existence of all keywords in the cluster before padding. If the condition is fail, we see that both $\textnormal{EDB}_i^0$ and $\textnormal{EDB}_i^1$ are indistinguishable under the choice of $b\xleftarrow{\$}\{0, \dots, \alpha-1\}$. The reason is because the challenger does not send any batch to the server. 
In contrast, if the condition is passed, \textit{Padding Controller} pads all the keywords in the cluster to be the same length and encrypt them before sending a batch to the server.  Meanwhile, $ST[w_i].c$ gets updated. Accordingly, $\textnormal{EDB}_t^b$ is indistinguishable under the choice of because these databases receive the batch of keywords that have the same length. 
The \textit{first batch} condition is crucial when \system~ is against the persistent adversary. It ensures there is no new keyword in the cluster appears in subsequent batches. Hence, the adversary cannot distinguish when a new keyword is added in $\textnormal{EDB}^b_i$.

If $ST[w_i].c >0 $ and \textit{first batch} condition has been met, \textit{Padding Controller} performs padding similarly with the padding strategy against the non-persistent adversary, presented in Algorithm 1. We can also see that $\textnormal{EDB}_t^b$ is  indistinguishable because \textit{Padding Controller} guarantees all the keywords in a cluster have the same length.

Now, we start analysing the query $q_t,i$ that queries the keyword $w_i$, with $ST[w_i].c=0$ or $ST[w_i].c>0$.

If $ST[w_i].c=0$, the adversary cannot guess the picked database because $\tau^b_{t,i}$ return nothing. 

If $ST[w_i].c>0$, $\tau^b_{t,i}$ is indistinguishable to all other query keywords in the same cluster at time $t$.

\begin{theorem}\label{theo:theorem_persistent}
Let $\Sigma_{P}$ = {\normalfont(\textsf{Setup,Streaming,Search})} be our DSSE scheme, and $\mathbf{C}^{[1,\dots,T]}=\{\mathfrak{C}^{\mathcal{DB}_0},\dots,\mathfrak{C}^{\mathcal{DB}_T}\}$ is a set of constraint sets. If $\Sigma$ is $\mathcal{L}_{P}$-constrained-adaptively-indistinguishable secure, and $\mathbf{C}^{[1,\dots,T]}$ is $(\mathcal{L}_{P},\alpha,\mathcal{F})$-acceptable, then $\Sigma$ is $(\mathcal{L}_{P},\mathbf{C}^{[1,\dots,T]},\alpha,\mathcal{F})$-constrained-adaptively-indistinguishability secure.
\end{theorem}

\begin{proof}

Let $\mathcal{A}$ be an adversary in the {\normalfont Ind$_{\mathrm{DSSE},\mathcal{A},\mathcal{L}_{P}, \mathbf{C}^{[1,\dots,T]},\alpha,\mathcal{F}}$} game. 
We construct an adversary $\mathcal{B}$ against the game. $\mathcal{B}$ first randomly picks two integer $\alpha_0,\alpha_1 \in \{0,\alpha-1\}$. 
Then, $\mathcal{B}$ starts $\mathcal{A}$ and receives $\alpha$ databases $(\textnormal{DB}^0_0,\dots,\textnormal{DB}^{\alpha-1}_0)$. 
Upon giving the pair $(\textnormal{DB}^{\alpha_0}_0,\textnormal{DB}^{\alpha_1}_0)$ to the challenger, where the challenger holds a random secret bit $b$, $\mathcal{B}$ receives the challenge encrypted database $\textnormal{EDB}^*_0$ which she forwards to $\mathcal{A}$. 
Then, for every $t\in[1,\dots,T]$, we have:
\begin{itemize}
\item $\mathcal{A}$ sends $u_t$ (stream queries) to  challenger and receives $\textnormal{EDB}^*_t$.
\item $\mathcal{A}$ outputs $((q^0_{t,1},\dots,q^{\alpha-1}_{t,1}),\dots,(q^0_{t,n},\dots,q^{\alpha-1}_{t,n}))$ and gives to $\mathcal{B}$.
\item To respond, $\mathcal{B}$ outputs $((q^{\alpha_0}_{t,1},q^{\alpha_1}_{t,1}),\dots, (q^{\alpha_0}_{t,n},q^{\alpha_1}_{t,n}))$ to the game and receives back the transcripts $(\tau^*_{t,1},\dots,\tau^*_{t,n})$, and forwards them to $\mathcal{A}$.
\end{itemize}

After executing all $t\in[1,\dots,T]$, $\mathcal{A}$ outputs the integer $\alpha^\prime$.
If $\alpha^\prime=\alpha_0$, $\mathcal{B}$ outputs $b^\prime=0$, else if $\alpha^\prime=\alpha_1$, $\mathcal{B}$ outputs $b^\prime=1$, and otherwise outputs the probability $1/2$ for the output $0$ and the probability $1/2$ for the output $1$.
%

We first make an observation: for the pair $(H^{\alpha_0}_t,H^{\alpha_1}_t)$ at $\forall t \in [1,T]$, the views of the adversary  $\mathcal{B}$ are indistinguishable due to   $\mathcal{L}_{P}(H^{\alpha_0}_t)=\mathcal{L}_{P}(H^{\alpha_1}_t)$, presenting both satisfying $\mathbf{C}^{[1,\dots,T]}$. 
Then we can formalise $\mathcal{B}$ as follows:
\begin{equation}
\begin{split}
\label{eq:adv_b_persistent}
{\normalfont  \textbf{Adv}^{\mathrm{Ind}}_{\mathcal{B},\mathcal{L}_{P}, \mathbf{C}^{[1,\dots,T]},\mathcal{F}}}  (\lambda) =  \bigl|
           \mathbb{P}[b=b^\prime] - \frac{1}{2} \bigr| \leq negl(\lambda) 
\end{split}
\end{equation}
Now, we evaluate $ \mathbb{P}[b=b^\prime]$ as follows.
\begin{equation}
\begin{split}
\label{eg:probability_B_persistent}
\mathbb{P}[b=b^\prime] & =  \\
          &~\mathbb{P}[ b=b^\prime| \alpha^\prime \in \{\alpha_0,\alpha_1\}]  \cdot  \mathbb{P}[\alpha^\prime \in \{\alpha_0,\alpha_1\}] \\
        & +\mathbb{P}[b=b^\prime|\alpha^\prime \notin \{\alpha_0,\alpha_1\}] \cdot  \mathbb{P}[\alpha^\prime \notin \{\alpha_0,\alpha_1\}] \\
        & = \mathbb{P}[ b=b^\prime \cap \alpha^\prime \in \{\alpha_0,\alpha_1\}] \\
        & +\mathbb{P}[b=b^\prime|\alpha^\prime \notin \{\alpha_0,\alpha_1\}] \cdot  \mathbb{P}[\alpha^\prime \notin \{\alpha_0,\alpha_1\}] \\
        & = \mathbb{P}[\mathcal{A} \textnormal{ wins the } {\normalfont Ind_{\mathrm{DSSE},\mathcal{A},\mathcal{L}_{P}, \mathbf{C}^{[1,\dots,T]},\alpha,\mathcal{F}}} \textnormal{ game}] \\
        & + \frac{1}{2}\left(1-\mathbb{P}[\alpha^\prime \in \{\alpha_0,\alpha_1\}]\right)
\end{split}
\end{equation}

Now, we evaluate $\mathbb{P}[\alpha^\prime \in \{\alpha_0,\alpha_1\}]$ as follows.
\begin{equation*}
\begin{split}
\mathbb{P}[\alpha^\prime \in \{\alpha_0,\alpha_1\}] & = \mathbb{P}[ \alpha^\prime=\alpha_0] + \mathbb{P}[ \alpha^\prime=\alpha_1]
\end{split}
\end{equation*}

Since we have
\begin{equation*}
\begin{split}
\mathbb{P}[ \alpha^\prime=\alpha_0] + & \mathbb{P}[ \alpha^\prime=\alpha_1]  = \\
& \mathbb{P}[ \alpha^\prime=\alpha_b| b=0] + \mathbb{P}[ \alpha^\prime=\alpha_b |b=1]
\end{split}
\end{equation*}
then,
\begin{equation*}
\begin{split}
\mathbb{P}[\alpha^\prime \in \{\alpha_0,\alpha_1\}] & = \frac{1}{2}\left (\mathbb{P}[ \alpha^\prime=\alpha_b| b=0] +  \mathbb{P}[ \alpha^\prime=\alpha_0]  \right) \\
 & + \frac{1}{2}\left (\mathbb{P}[ \alpha^\prime=\alpha_b| b=1] +  \mathbb{P}[ \alpha^\prime=\alpha_1]  \right)
\end{split}
\end{equation*}
We note that $\mathbb{P}[ \alpha^\prime=\alpha_b]$ is the probability $\mathcal{A}$ wins the 1-out-of-$\alpha$ indistinguishability game, and $\alpha_0$ and $\alpha_1$ are uniformly selected from $\{0,\alpha-1\}$, then we have
\begin{equation}
\begin{split}
\label{eq:alpha_12_persistent}
\mathbb{P}[\alpha^\prime & \in \{\alpha_0,\alpha_1\}]  = \\ &\mathbb{P}[\mathcal{A} \textnormal{ wins the } {\normalfont Ind_{\mathrm{DSSE},\mathcal{A},\mathcal{L}_{P}, \mathbf{C}^{[1,\dots,T]},\alpha,\mathcal{F}}} \textnormal{ game}] 
+ \frac{1}{\alpha}
\end{split}
\end{equation}

\noindent Applying Eq.~\ref{eq:alpha_12_persistent} to Eq.~\ref{eg:probability_B_persistent}, we have
\begin{equation*}
\begin{split}
\mathbb{P}[b=b^\prime] & = \\ & \frac{1}{2} \cdot\mathbb{P}[\mathcal{A} \textnormal{ wins the } {\normalfont Ind_{\mathrm{SSE},\mathcal{A},\mathcal{L}_{P}, \mathbf{C}^{[1,\dots,T]},\alpha,\mathcal{F}}} \textnormal{ game}] \\
& + \frac{1}{2} - \frac{1}{2\alpha}
\end{split}
\end{equation*}

\noindent Then, from Equation~\ref{eq:adv_b_persistent}, we can derive
\begin{equation}
\begin{split}
\label{eq:adv_b_derieved_persistent}
& {\normalfont   \textbf{Adv}^{\mathrm{Ind}}_{\mathcal{B},\mathcal{L}_{P},\mathbf{C}^{[1,\dots,T]},\mathcal{F}}}  (\lambda) =\\
& \frac{1}{2} \left(\mathbb{P}[\mathcal{A} \textnormal{ wins the } {\normalfont Ind_{\mathrm{SSE},\mathcal{A},\mathcal{L}_{P}, \mathbf{C}^{[1,\dots,T]},\alpha,\mathcal{F}}} \textnormal{ game}] - \frac{1}{\alpha}\right)
\end{split}
\end{equation}

\noindent Applying Equation~\ref{eq:adv_b_derieved_persistent} to Equation~\ref{eq:persistent_alpha_f}, finally, we can conclude

\begin{equation}
\begin{split}
\label{eq:adv_b_a_adv_persistent}
 {\normalfont   \textbf{Adv}^{\mathrm{Ind}}_{\mathcal{B},\mathcal{L}_{P},\mathbf{C}^{[1,\dots,T]},\mathcal{F}}}  (\lambda) =
\frac{1}{2}{\normalfont  \textbf{Adv}^{\mathrm{Ind}}_{\mathcal{A},\mathcal{L}_{P},\mathbf{C}^{[1,\dots,T]},\alpha,\mathcal{F}}}  (\lambda)
\end{split}
\end{equation}

%
%
%
%

\end{proof}

\newpage
\begin{IEEEbiography}[{\includegraphics[width=0.9in,height=1.0in,clip,keepaspectratio]{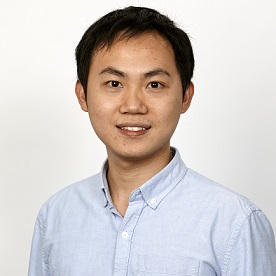}}]{Viet Vo} is a Ph.D. student at the Department of Software and Security, Monash University. He received his B.S. degree from Ho Chi Minh University of Science in 2010, the M.Eng. degree from Chonnam National University in 2013, the MPhil. degree from Monash University in 2016. His research interests include data security and privacy, secure networked system, and data processing. His research has been supported by CSIRO Data61. He was the recipient of the Publication Award at Monash University in 2021.\end{IEEEbiography}
\begin{IEEEbiography}[{\includegraphics[width=0.9in,height=1.0in,clip,keepaspectratio]{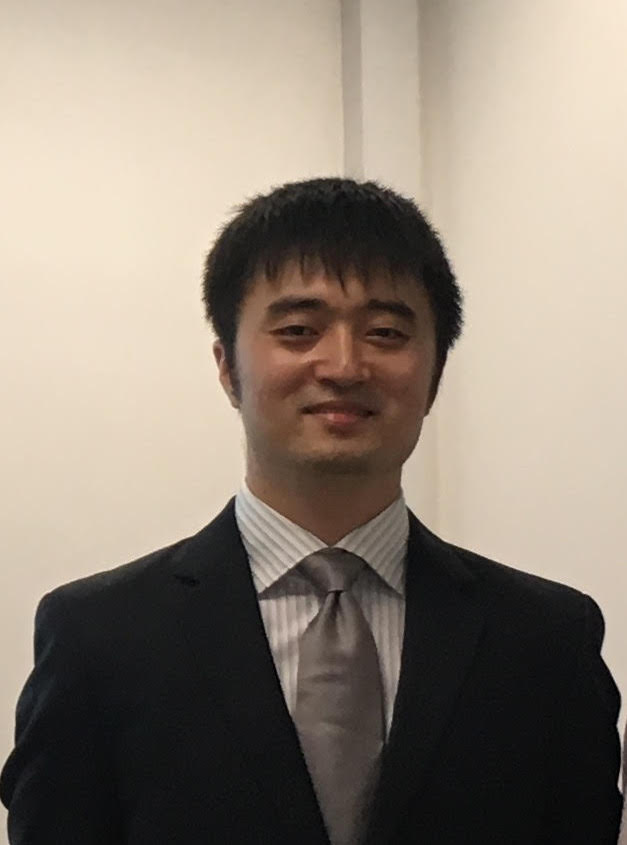}}]{Xingliang Yuan} is currently a Senior Lecturer (aka US Associate Professor) at the Department of Software Systems and Cybersecurity in the Faculty of Information Technology, Monash University, Australia. His research interests include data security and privacy, secure networked system, machine learning security and privacy, confidential computing. His research has been supported by Australian Research Council, CSIRO Data61, and Oceania Cyber Security Centre. In the past few years, his work has appeared in prestigious venues in cybersecurity, computer networks, and distributed systems, such as ACM CCS, NDSS, IEEE INFOCOM, IEEE TDSC, IEEE TIFS, and IEEE TPDS. He was the recipient of the Dean’s Award for Excellence in Research by an Early Career Researcher at Monash Faculty of IT in 2020. He received the best paper award in the European Symposium on Research in Computer Security (ESORICS) 2021, the IEEE Conference on Dependable and Secure Computing (IDSC) 2019, and the IEEE International Conference on Mobility, Sensing and Networking (MSN) 2015.\end{IEEEbiography}
\begin{IEEEbiography}[{\includegraphics[width=0.9in,height=1.0in,clip,keepaspectratio]{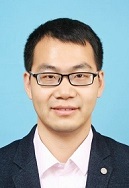}}]{Dr. Shi-Feng Sun} is currently an Associate Professor in the School of Cyber Science and Engineering at Shanghai Jiao Tong University, China. Prior to joining SJTU, he was a Lecturer in the Department of Software Systems and Cybersecurity at Monash University, Australia. Before that, he worked as a Research Fellow at Monash University and CSIRO, Australia. He received his Phd degree in Computer Science from Shanghai Jiao Tong University in 2016 and worked as a visiting scholar in the Department of Computing and Information Systems at the University of Melbourne during his PhD study. His research interest centers on cryptography and data privacy, particularly on provably secure cryptosystems against physical attacks, data privacy-preserving technology in cloud storage, and privacy-enhancing technology in blockchain. He has published more than 40 quality papers, including publications in ACM CCS, NDSS, EUROCRYPT, PKC, ESORICS, AsiaCCS, IEEE TDSC and IEEE TSC, etc.

\end{IEEEbiography}
\begin{IEEEbiography}[{\includegraphics[width=0.9in,height=1.0in,clip,keepaspectratio]{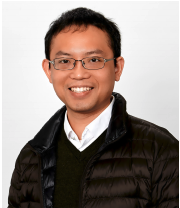}}]{Joseph Liu} is an Associate Professor in the Faculty of Information Technology, Monash University in Melbourne, Australia. He got his PhD from the Chinese University of Hong Kong at 2004. His research areas include cyber security, blockchain and applied cryptography. He has received more than 9800 citations and his H-index is 56, with more than 200 publications in top venues such as CRYPTO, ACM CCS, NDSS, INFOCOM. He is currently the lead of the Monash Cyber Security Discipline Group. He has established the Monash Blockchain Technology Centre at 2019 and serves as the founding director. His remarkable research in linkable ring signature forms the theory basis of Monero (XMR), one of the largest cryptocurrencies in the world with current market capitalization more than US\$6 billion. He has been given the prestigious ICT Researcher of the Year 2018 Award by the Australian Computer Society (ACS),  the largest professional body in Australia representing the ICT sector, for his contribution to the blockchain and cyber security community. He has been invited as the IEEE Distinguished Lecturer in 2021 for the topic of Blockchain in Supply Chain.
\end{IEEEbiography}
\begin{IEEEbiography}[{\includegraphics[width=0.9in,height=1.0in,clip,keepaspectratio]{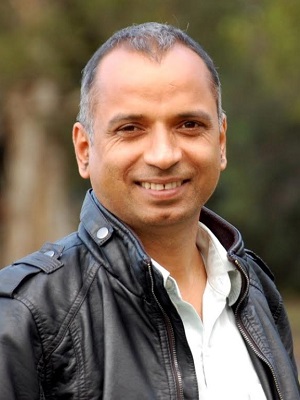}}]{Dr. Surya Nepal} is a Senior Principal Research Scientist at CSIRO Data61 and the Deputy Research Director of Cybersecurity Cooperative Research Centre (CRC). He has been with CSIRO since 2000 and currently leads the distributed systems security group comprising 30 staff and 50 PhD students. His main research focus is on the development and implementation of technologies in the area of distributed systems, with a specific focus on security, privacy and trust. He has more than 250 peer-reviewed publications to his credit. He is a member of the editorial boards of IEEE Transactions on Service Computing, ACM Transactions on Internet Technology, IEEE Transactions on Dependable and Secure Computing,  and Frontiers of Big Data- Security Privacy, and Trust. Dr Nepal also holds an honorary professor position at Macquarie University and a conjoint faculty position at UNSW. \end{IEEEbiography}
\begin{IEEEbiography}[{\includegraphics[width=0.9in,height=1.0in,clip,keepaspectratio]{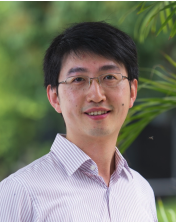}}]{Cong Wang} is a Professor in the Department of Computer Science, City University of Hong Kong. His research interests include data and network security, blockchain and decentralized applications, and privacy-enhancing technologies. He has been one of the Founding Members of the Young Academy of Sciences of Hong Kong since 2017, and has been conferred the RGC Research Fellow in 2021. He received the Outstanding Researcher Award (junior faculty) in 2019, the Outstanding Supervisor Award in 2017 and the President's Awards in 2019 and 2016, all from City University of Hong Kong. He is a co-recipient of the Best Paper Award of IEEE ICDCS 2020, ICPADS 2018, MSN 2015, the Best Student Paper Award of IEEE ICDCS 2017, and the IEEE INFOCOM Test of Time Paper Award 2020. His research has been supported by multiple government research fund agencies, including National Natural Science Foundation of China, Hong Kong Research Grants Council, and Hong Kong Innovation and Technology Commission. He has served as associate editor for IEEE Transactions on Dependable and Secure Computing (TDSC), IEEE Transactions on Services Computing (TSC), IEEE Internet of Things Journal (IoT-J), IEEE Networking Letters, and The Journal of Blockchain Research, and TPC co-chairs for a number of IEEE conferences and workshops. He is a fellow of the IEEE, and member of the ACM.\end{IEEEbiography}

\end{document}